\documentclass[aps,prx,superscriptaddress,twocolumn,twoside,floatfix,a4paper]{revtex4-2}
\usepackage[utf8]{inputenc}
\usepackage{amsmath, amssymb, amsthm}
\usepackage{stackengine,scalerel}
\usepackage{graphicx}
\usepackage{ragged2e}
\usepackage{comment}
\usepackage[colorlinks = true,linkcolor = blue,urlcolor = blue,citecolor = blue,anchorcolor = blue]{hyperref}
\usepackage{url}
\usepackage{float}
\usepackage[table]{xcolor}

\usepackage{wasysym}
\usepackage[braket, qm]{qcircuit}

\newcommand{\MCO}[3]{[#1]^{#2}_{#3}}
\newcommand{\sqO}[2]{[#1]_{#2}}
\newcommand{\floor}[1]{\lfloor #1 \rfloor}
\newcommand{\ceil}[1]{\lceil #1 \rceil}

\newcommand{\abs}[1]{| #1 |}
\newcommand{\rpicl}{\Pi}
\newcommand{\rpi}[1]{\rpicl(#1)}

\newcommand{\rv}[2]{R_{#2}(#1)}
\newcommand{\rz}{R_{\hat{z}}}
\newcommand{\rx}{R_{\hat{x}}}
\newcommand{\rzgate}[1]{\gate{\mathrm{\rz(#1)}}}

\newcommand{\rpigate}[1]{\gate{\mathrm{\rpi{#1}}}}
\newcommand{\rvgate}[2]{\gate{\mathrm{\rv{#1}{#2}}}}
\newcommand\eye{\ensurestackMath{\stackinset{c}{}{c}{-.3pt}%
  {\bullet}{\scriptstyle\bigcirc}}}

\newcommand{\controlt}{*!<0em,.0em>-=-<.2em>{\eye}}
\newcommand{\controlz}{*!<0.em,.0em>-=-<.2em>{\eyeright}}
\newcommand{\controlzl}{*!<0em,.0em>-=-<.2em>{\eyeleft}}
\newcommand{\ctrlzr}[1]{\controlz \qwx[#1] \qw}
\newcommand{\ctrlzl}[1]{\controlzl \qwx[#1] \qw}
\newcommand{\ctrlt}[1]{\controlt \qwx[#1] \qw}

\newtheorem{lemma}{Lemma}

\newtheorem{theorem}{Theorem}

\renewcommand{\sec}[1]{Section~\hyperref[sec:#1]{\ref*{sec:#1}}}
\newcommand{\ssec}[1]{Section~\hyperref[ssec:#1]{\ref*{ssec:#1}}}
\newcommand{\fig}[1]{Figure~\hyperref[fig:#1]{\ref*{fig:#1}}}
\newcommand{\tab}[1]{Table~\hyperref[tab:#1]{\ref*{tab:#1}}}
\newcommand{\lem}[1]{Lemma~\hyperref[lem:#1]{\ref*{lem:#1}}}
\newcommand{\prop}[1]{\hyperref[prop:#1]{Proposition~\ref*{prop:#1}}}
\newcommand{\thm}[1]{Theorem~\hyperref[thm:#1]{\ref*{thm:#1}}}
\newcommand{\apx}[1]{Appendix~\hyperref[apx:#1]{\ref*{apx:#1}}}
\newcommand{\qc}[1]{Circuit~(\hyperref[qc:#1]{\ref*{qc:#1}})}

\newcounter{qcnum}
\newcommand{\qcref}[1]{\refstepcounter{qcnum}\tag{\theqcnum}\label{qc:#1}}

\usepackage{graphicx}

\begin{document}
\title{Efficient Implementation of Multi-Controlled Quantum Gates}

\author{Ben~Zindorf}
\email{ben.zindorf.19@ucl.ac.uk}
\affiliation{Department of Physics and Astronomy, University College London, Gower Street, WC1E 6BT London, United Kingdom}

\author{Sougato~Bose}
\affiliation{Department of Physics and Astronomy, University College London, Gower Street, WC1E 6BT London, United Kingdom}

\date{\today}

\begin{abstract}
We present an implementation of multi-controlled quantum gates which provides significant reductions of cost compared to state-of-the-art methods. The operator applied on the target qubit is a unitary, special unitary, or the Pauli X operator (Multi-Controlled Toffoli), {and requires one clean ancilla, no ancilla, and one dirty ancilla, respectively.}
We generalize our methods for any number of target qubits, and provide further cost reductions if additional ancilla qubits are available.
For each type of multi-controlled gate, we provide implementations for unrestricted (all-to-all) connectivity and for linear-nearest-neighbor.
All of the methods use a linear cost of gates from the Clifford+T (fault-tolerant) set. In the context of linear-nearest-neighbor (LNN) architecture, the cost and depth of our circuits scale linearly irrespective of the position of the qubits on which the gate is applied. Our methods directly improve the compilation process of many quantum algorithms, providing optimized circuits.
{Given the scale of our improvements, for example, quadratic to linear CNOT count for LNN,
they will naturally result in a large reduction of errors.}
\end{abstract}
\maketitle

\section{Introduction}

Multi-controlled single/multiple target quantum gates are pivotal ingredients for quantum computation. For example, a constructive method for implementing arbitrary unitaries, which is, in fact, used for the proof of the sufficiency of fundamental gate sets for universal quantum computation, uses such gates \cite{nielsen_quantum_2010, shende_synthesis_2005}. Of course, that constructive method may not be the most efficient, but nonetheless emphasizes that such gates are at the heart of quantum computation. 

Such gates are also very useful in circuits for quantum chemistry, e.g., particle number conserving Hamiltonians \cite{arrazola_universal_2022}, for constructing a circuit implementing a quantum neuron \cite{de_carvalho_parametrized_2024}, {QAOA \cite{ni_progressive_2024}}, quantum circuits for isometries \cite{iten_quantum_2016, malvetti_quantum_2021}, { port-based teleportation \cite{grinko_efficient_2023}, distributed quantum computation \cite{tanasescu_distribution_2022},} constructing a quantum RAM \cite{park_circuit-based_2019}, in the circuit for Grover's quantum search algorithm \cite{grover_quantum_1997}, and generally, for quantum arithmetics \cite{ali_function_2018}, such as for the arithmetic part of Shor's algorithm \cite{shafaei_cofactor_2014}. 
{ These gates can be seen as the quantum counterparts of the conditional if-then-else statements \cite{lewis_matrix_2022}, which are frequently used in classical computation, making them a convenient logical building block for quantum algorithm design. Thus, one may expect these elementary gates \cite{barenco_elementary_1995} to appear in additional algorithms which will be developed in the future.

Unfortunately, while it may be possible to apply multi-qubit gates directly in some quantum architectures \cite{isenhower_multibit_2011,petrosyan_fast_2024,jandura_time-optimal_2022}, multi-controlled (MC) gates with more than one control qubit are not available for direct implementation in most existing quantum hardware and must be decomposed into smaller available gates. In most near-term NISQ (Noisy Intermediate-Scale Quantum) devices, tunable single-qubit gates (e.g. $R_z(\theta)$,$R_y(\theta)$) are used along with some fixed two-qubit gate (e.g. CNOT,CZ,ECR \cite{abughanem_ibm_2024,abughanem_toffoli_2025,xue_quantum_2022}) as the fundamental building blocks, such that the latter gates are usually characterized by a higher error rate and longer execution time, making them a more "expensive" resource than the former gates. As in most cases, these two-qubit gates are equivalent to a CNOT up to single-qubit rotations, it is common to design a circuit using CNOT as the two-qubit gate - a convenient choice due to its similarities to the classical XOR logic gate. As the errors introduced by each CNOT gate accumulate and propagate, when used to construct a large circuit, minimizing the gate count is an important task, as it directly improves the fidelity of the circuit. Then, a circuit designed using these gates can be directly converted to a circuit which uses another available two-qubit gate by only increasing the number of the cheaper single-qubit gates. 

The CNOT is also part of the Clifford+T gate set \cite{bravyi_universal_2005,nielsen_quantum_2002,jones_logic_2013} which is a prime candidate for the upcoming error-corrected fault-tolerant quantum computation, however, in most error correction codes, the single-qubit T gate is the most expensive resource, followed by the CNOT. A decomposition of multi-controlled gates which minimizes both CNOT and T gates will therefore be beneficial both in the present and in the future, directly improving the fidelity of any circuit which is built using these logical building blocks. Other error correction codes exist, in which the most expensive Clifford+T gate is the Hadamard (H) \cite{vasmer_three-dimensional_2019,butt_fault-tolerant_2024}, followed by the CNOT, and therefore a low count of H gates is also beneficial. 
The Clifford+T gate set is universal and therefore can be used to approximately implement arbitrary single-qubit gates, such that greater accuracy requires a higher number of T and H gates \cite{bravyi_universal_2005}. The use of arbitrary single-qubit gates should therefore be minimized in the fault-tolerant regime. 
 There are other universal sets which are widely used, such as the Toffoli+H and the NCV library ($X,CNOT,CV,CV^\dagger$) \cite{barenco_elementary_1995,cheng_mapping_2018}. As any of the mentioned gate sets can be used to implement the Toffoli, it may be a good strategy to start by decomposing an MC gate using Toffoli gates and then further decomposing these into smaller gates, thus producing a circuit that can be easily converted to the gate set of choice.

In addition to the limited available gates, another obstacle is the restricted connectivity between qubits, which is common to many architectures. While some allow for all-to-all connectivity \cite{schoenberger_shuttling_2024,bluvstein_logical_2024}, others are restricted to neighboring qubits in 2D or 1D arrays. Since there are many options for the arrangement of 2D arrays, such as the square lattice used by Google, and the heavy-hex used by IBM, it can be useful to find an efficient implementation for the highly restrictive 1D linear nearest-neighbor (LNN) connectivity, as it can be used directly in some technologies, and can also be mapped to many 2D arrays with no/low overhead. On the other extreme, while decompositions designed for all-to-all connectivity can be applied in some architectures, these are especially useful as a conceptual unrestricted playground, which can be a good first step before attempting to find decompositions in more restrictive connectivities. Although any all-to-all implementation can be mapped to LNN using SWAP gates, doing so can drastically increase the cost, e.g., increasing the CNOT count of a circuit from linear to quadratic scaling in the number of qubits \cite{pedram_layout_2016,saeedi_synthesis_2011,he_mapping_2019,khan_cost_2008,lukac_optimization_2021,ding_fast_2019,zhang_method_2022}. To this day, this was true for the implementations of MC gates as well - while all-to-all decompositions require a linear CNOT count using a limited number of ancilla qubits \cite{barenco_elementary_1995,vale_decomposition_2023,iten_quantum_2016,he_decompositions_2017,balauca_efficient_2022,maslov_reversible_2011,kole_improved_2017,sasanian_ncv_2011,biswal_improving_2016,ali_quantum_2015,niemann_t-depth_2019,leng_decomposing_2024,abdessaied_technology_2016}, the best known upper bound of CNOT gates in LNN connectivity is quadratic \cite{arsoski_implementing_2024,arsoski_multi-controlled_2025,cheng_mapping_2018,chakrabarti_nearest_2007,miller_elementary_2011,li_quantum_2023,tan_multi-strategy_2018} due to the large number of SWAP gates. Without ancilla, some MC (with a target operator $U\in U(2) \setminus SU(2)$) gate implementations in all-to-all connectivity require a quadratic CNOT count \cite{arsoski_implementing_2024,arsoski_multi-controlled_2025,silva_linear_2023,da_silva_linear-depth_2022,saeedi_linear-depth_2013,leng_decomposing_2024},while others allow for a linear count \cite{gidney_constructing_nodate}, such that in both cases the count of arbitrary single-qubit rotations scales at least linearly and these gates require exponential accuracy. The decomposition of multi-controlled gates with a small number of controls under various connectivity and ancillary constraints has been addressed as well \cite{cruz_shallow_2023,gwinner_benchmarking_2021,nakanishi_decompositions_2024,nakanishi_quantum-gate_2021,nemkov_efficient_2023,song_optimal_2002,zindorf_quantum_2024,heusen_measurement-free_2024}.

In addition to gate count, an important metric is the depth of the circuit, i.e. the number of time slots holding several gates which can be executed in parallel. Minimizing the depth results in a shorter run-time, which reduces decoherence-induced errors.
Some approaches aim to minimize the depth of multi-controlled gates \cite{claudon_polylogarithmic-depth_2024,nie_quantum_2024}, however, this comes with a trade-off, either increasing the gate count or the number of ancillae used. Low-depth implementations in all-to-all connectivity may relay on the ability to execute many two-qubit gates in parallel, however, the number of gates which can be efficiently applied in parallel may be limited in some all-to-all connected devices \cite{figgatt_parallel_2019}, and such implementations may increase the gate count overhead when mapped to more restricted connectivities. As the depth of each gate type is no larger than its count, low gate count implementations may provide the lowest depth as well, in case efficient parallel gate execution is not available or is restricted to a constant number of gates.

The use of logical building blocks when constructing an algorithm can be preferable to constructing the corresponding quantum circuit using the smallest available gates directly. This may be seen as analogous to developing a classical software using assembly - while this might provide more efficient implementations than the use of built-in functions in high-level programming languages like C++ or Python, most would choose the latter for complicated tasks. A compilation step is then required to convert these functions into assembly code, such that an improved compilation of a frequently used function can improve the efficiency of the entire program. On the other hand, one may wish to define the entire quantum program as one large unitary matrix, then apply a compilation step to decompose the entire operator to the smallest gates. Although various methods exist \cite{liu_qcontext_2023,meijer_exploiting_2024,giles_exact_2013}, the resulting gate count upper bound is exponential, due to the number of parameters in the unitary - making it implausible to achieve any improvement over a classical computer. 

In this paper, we focus on minimizing the gate count of CNOT,T and H (Clifford+T) required for the implementation of various types of multi controlled gates, which, as discussed above, is beneficial to both fault-tolerant as well as near-term applications.  Specifically we address the cases in which the target operator is in $SU(2)$, $U(2)$, or the Pauli X. We wish to make minimal assumptions regarding the hardware used and the available resources in order to cover as many use-cases as possible. For each type of gate, we use the minimal number of ancilla qubits allowing for a linear scaling of Clifford+T gates and for a constant number of single-qubit arbitrary rotations. The type of ancilla qubit is important as well, such that a "clean" ancilla which must be kept in a specific state is a more expensive resource than a borrowed "dirty" ancilla which can be any qubit that is unused at the time of the MC gate execution. We implement each gate in all-to-all and LNN connectivity. When the depth can be reduced without increasing the gate count, we apply these reductions as well. Finally, while requiring the minimal number of ancilla qubits, we provide gate count reductions, which can be achieved in case the MC gate is only applied to a subset of all available qubits, and the rest can be borrowed as ancillae.
In all-to-all connectivity we achieve a CNOT count reduction of up to $62.5\%$ compared to the state-of-the-art, and for the multi-controlled $X$ case we reduce the Toffoli count by $50\%$. In LNN connectivity, our methods provide a linear upper bound of the CNOT gate count for any qubit ordering, a {\em quadratic} reduction compared to state-of-the-art methods.
}
    \subsection{Notations}

In this paper, we use the notations which were defined in \cite{zindorf_quantum_2024}. We list some of these notations here for convenience.

Any {single-qubit} $SU(2)$ operator can be represented as a rotation by angle $\lambda$ about an axis $\hat{v}$ as $\rv{\lambda}{\hat{v}}$ \cite{nielsen_quantum_2010}. We use the notation $\rpi{\hat{v}}:=i\rv{\pi}{\hat{v}}$ to describe a Hermitian {gate applying a} $\pi$-rotation { about an arbitrary axis}.
The three-dimensional rotation matrix by angle $\theta$ about the vector $\hat{v}$ is marked as
$\hat{R}_{\hat{v}}(\theta)$.

These notations have been used to describe the decomposition of any $SU(2)$ operator in terms of two $\Pi$ gates in [\cite{zindorf_quantum_2024}\ {\em Lemma 1}] as follows.

\begin{lemma}\label{lem:2_rpi}
Any $\rv{\lambda}{\hat{v}}\in SU(2)$ operator can be implemented as $\rpi{\hat{v}_2}\rpi{\hat{v}_1}$ 
such that $\hat{v}_1$ can be chosen as any unit vector perpendicular to $\hat{v}$, and $\hat{v}_2 = \hat{R}_{\hat{v}}(\frac{\lambda}{2})\hat{v}_1$ with $\frac{\lambda}{2}\in (-\pi,\pi]$.
\end{lemma}
This lemma can be seen as a generalisation of known identities of the Pauli matrices $X = \rpi{\hat{x}}$, $Y = \rpi{\hat{y}}$ and $Z = \rpi{\hat{z}}$. For example, since the unit vector $\hat{y}$ can be obtained by rotating $\hat{x}$ about $\hat{z}$ , as $\hat{y} = \hat{R}_{\hat{z}}(\frac{\pi}{2})\hat{x}$, and $\hat{x}\perp\hat{z}$, we get 
$\rpi{\hat{y}}\rpi{\hat{x}} = \rv{\pi}{\hat{z}}$, i.e. $YX=-iZ$. 

When referring to an operator $O$ applied on a set of qubits $\tau=\{t_1,t_2,..,t_m\}$, we use the notation $\sqO{O}{\tau}$. If the operator is controlled by a qubit set $C=\{c_1,c_2,..,c_n\}$ and applied on a qubit $t$, we write $\MCO{O}{\{c_1,c_2,..,c_n\}}{t}$ or $\MCO{O}{C}{t}$, and use one of the following gate notations:
\[
\Qcircuit @C=1.0em @R=0.0em @!R { 
	 	\nghost{{c}_{1} :  } & \lstick{{c}_{1} :  } & \ctrl{1} & \qw \\
	 	\nghost{{c}_{2} :  } & \lstick{{c}_{2} :  } & \ctrl{1} & \qw \\
        \nghost{{q}_{1} :  } &  & \ar @{.} [1,0]   \\
        \nghost{{q}_{1} :  } &  &  \\
	 	\nghost{{c}_{n} :  } & \lstick{{c}_{n} :  } & \ctrl{1}\ar @{-} [-1,0] & \qw \\
	 	\nghost{{t} :  } & \lstick{{t} :  } & \gate{\mathrm{O}} & \qw\\
   }
 \hspace{5mm}\raisebox{-15mm}{=}\hspace{0mm}
  \Qcircuit @C=1.0em @R=1.7em @!R {
            \nghost{{C} :  } \\
	 	\nghost{{C} :  } & \lstick{{C} :  } & \qw {/^n}& \ctrl{1} & \qw \\
	 	\nghost{{t} :  } & \lstick{{t} :  } & \qw  & \gate{\mathrm{O}} & \qw \\
 }
\]

Both [\cite{zindorf_quantum_2024}\ {\em Lemma 4} and {\em Lemma 13} ] provide methods which allow one to transform the rotation axis of multi-controlled gates. We place these here for convenience.

\begin{lemma}\label{lem:MC_transform} 
$
\MCO{e^{i\psi}\rv{\lambda}{\hat{v}}}{C}{t}
=
\sqO{\rpi{\hat{v}_{M}}}{t}
\MCO{e^{i\psi}\rv{\lambda}{\hat{v}'}}{C}{t}
\sqO{\rpi{\hat{v}_{M}}}{t}
$
for any angles $\lambda,\psi$, unit vectors $\hat{v},\hat{v}'$, and $\hat{v}_M \in M(\hat{v},\hat{v}')$.
\end{lemma}

\begin{lemma}\label{lem:Cpi_zz_yy}
If $\hat{v}=\hat{R}_{\hat{\sigma}}(\phi)\hat{\tau}$ for any unit vectors $\hat{\tau},\hat{\sigma}$ satisfying $\hat{\tau}\perp \hat{\sigma}$, then $\MCO{\rpi{\hat{v}}}{C}{t}=\sqO{\rv{\phi}{\hat{\sigma}}}{t}\MCO{\rpi{\hat{\tau}}}{C}{t}\sqO{\rv{-\phi}{\hat{\sigma}}}{t}$.
\end{lemma}
The set $M(\hat{v}_1,\hat{v}_2)$ is defined for any two unit vectors $\hat{v}_1,\hat{v}_2$   such that $\hat{v}_M \in M(\hat{v}_1,\hat{v}_2)$ iff $\hat{v}_2 = \hat{R}_{\hat{v}_M}(\pi)\hat{v}_1$. In other words, $\hat{v}_M$ can be chosen as a unit vector that is in the "middle" between $\hat{v}_1$ and $\hat{v}_2$.
A well-known special case of \lem{MC_transform} is achieved by noting that $H=\rpi{\hat{h}}$, with $\hat{h}\in M(\hat{x},\hat{z})$.

The phase gate is defined as:
\[
P(\lambda):= 
e^{i\frac{\lambda}{2}}\rv{\lambda}{\hat{z}} = 
\left( 
\begin{matrix} 
1 & 0 \\ 
0 & e^{i\lambda}
\end{matrix} 
\right).
\]
We use the standard notations for two specific phase gates: $S:=P(\frac{\pi}{2}) = \left( \begin{matrix} 1 & 0 \\ 0 & i\end{matrix}\right)$, and  $T:=P(\frac{\pi}{4})$.

The notation $C[a:b] := \{c_a,c_{a+1},..,c_b\}$ is used to define a subset and $C[a:a]=\{c_a\}$ is written as $C[a]$.
When concatenating two sets of qubits, we simply use $\{C,\tau\}:=\{c_1..c_n,t_1..t_m\}$.
For a controlled phase gate, the phase is applied iff all the control qubits {\em and} the target qubit are in the state $\ket{1}$, and therefore it is not required to specify the target. For example, the gate $\MCO{Z}{C}{t}$ can also be written as $\MCO{Z}{\{C,t\}}{}$.
We refer to any gate which can be decomposed using only multi-controlled (or single-qubit) phase gates as a relative phase gate, usually marked as $[\Delta]$.
Finally, a multi-controlled multi-target gate $\prod_{j=1}^{m} \MCO{O}{C}{t_j}$ is marked as follows.
\[
\resizebox{0.9\linewidth}{!}{
\Qcircuit @C=0.6em @R=0.2em @!R { 
	 	\nghost{{C} :  } & \lstick{{C} :  } & \qw {/^n} & \ctrl{1} & \qw \\
	 	\nghost{{t}_{1} :  } & \lstick{{t}_{1} :  } & \qw & \gate{\mathrm{O_1}}\qwx[1] & \qw \\
	 	\nghost{{t}_{2} :  } & \lstick{{t}_{2} :  } & \qw & \gate{\mathrm{O_2}}\qwx[1] & \qw \\
	 	\nghost{{t}_{3} :  } & \nghost{{t}_{3} :  } & &  \ar @{.} [1,0]  \\
        \nghost{{t}_{3} :  } & \nghost{{t}_{3} :  } & &    \\
	   \nghost{{t}_{m} :  } & \lstick{{t}_{m} :  } & \qw  & \gate{\mathrm{O_m}}\qwx[-1] & \qw \\
 }
 \hspace{5mm}\raisebox{-11.5mm}{=}\hspace{0mm}
\Qcircuit @C=0.6em @R=0.2em @!R { 
	 	\nghost{{C} :  } & \lstick{{C} :  } & \qw {/^n} & \ctrl{1} & \ctrl{2} & \qw & \ctrl{5} & \qw \\
	 	\nghost{{t}_{1} :  } & \lstick{{t}_{1} :  } & \qw & \gate{\mathrm{O_1}} & \qw & \qw & \qw & \qw \\
	 	\nghost{{t}_{2} :  } & \lstick{{t}_{2} :  } & \qw & \qw & \gate{\mathrm{O_2}} & \qw & \qw & \qw \\
	 	\nghost{{q}_{3} :  } &  &  &  &  \ar @{.} [1,1] &  \\
   \nghost{{q}_{3} :  } &  &  &  &  &  \\
	 	\nghost{{t}_{m} :  } & \lstick{{t}_{m} :  } & \qw & \qw & \qw 
        & \qw & \gate{\mathrm{O_m}} & \qw \\
 }
 }
\]

\section{All-to-all (ATA) connectivity}\label{sec:ATA_sec}

{We cover the decompositions of three types of MC gates in all-to-all connectivity, with the aim of minimizing the gate count of the Clifford+T gates CNOT,H and T. We start by providing a structure which can implement any multi-controlled $SU(2)$ gate with one or more targets and no ancilla, then using this to implement the multi-controlled Pauli $X$ gate with one dirty ancilla qubit. Finally, using a new trick, we show that multi-controlled $SU(2)$ gate with two targets can be used, along with one clean ancilla to implement any multi-controlled $U(2)$ gate. As all three gate types are based on the implementation of multi-controlled $SU(2)$, their gate counts are the same up to a small constant. The gate count provides an upper bound to the resulting circuit depth; however, we show that the depth can be lowered in case parallel gate execution is allowed.} 

\subsection{Multi-controlled SU(2)}

In this section, we focus on a cost-efficient implementation of multi-controlled $SU(2)$ gates (MCSU2) with a single or multiple targets and no ancilla qubit. { We start by decomposing the MCSU2 using four large many-qubit diagonal gates in \sec{MCSU2_macro_structure}. Then in \sec{decompose_1} we decompose these many-qubit gates using three-qubit gates and finally to single/two-qubit Clifford+T gates. We extend our methods to allow for an efficient multi-controlled multi-target $SU(2)$ decomposition in \apx{MCMTSU2}. For those who prefer a low T depth, we provide some alternative implementations in \apx{T_depth}, which allows one to reduce the depth of the T gate while only increasing the CNOT depth, or to further reduce the T depth by increasing the CNOT count as well. Although our methods do not require any ancilla qubits, we recognize that in some cases the MCSU2 gate may not apply on all available qubits, in which case idle qubits can be used as dirty ancilla to reduce the circuit cost. In \apx{extra_dirty} we address these cost reductions so that each additional ancilla contributes to the cost reduction, until reaching the lowest known gate count achieved for the multi-controlled X gate in \cite{maslov_advantages_2016} using $O(n)$ dirty ancilla, thus finding a bridge between no-ancilla and $O(n)$ dirty ancilla decompositions.}

\subsubsection{MCSU2 macro structure} \label{sec:MCSU2_macro_structure}

We provide a decomposition of any multi-controlled gate 
$\MCO{W}{C}{t}$ such that $W\in SU(2)$. Any such $W$ can be expressed as a rotation about an axis $\hat{v}$ by an angle $\lambda$ as $W=\rv{\lambda}
{\hat{v}}$. The decomposition requires eight single-qubit gates from $\{R_{\hat{x}},R_{\hat{z}}\}$, in addition to four large relative phase gates.
 
 As a first step, we provide the following decomposition.
\[
\resizebox{0.9\linewidth}{!}{
\Qcircuit @C=1.0em @R=0.3em @!R { 
	 	\nghost{{C_1} :  } & \lstick{{C_1} :  } & \qw {/^{n_1}} & \ctrl{1} & \qw \\
	 	\nghost{{C}_{2} :  } & \lstick{{C}_{2} :  } & \qw {/^{n_2}} & \ctrl{1} & \qw \\
	 	\nghost{{t} :  } & \lstick{{t} :  } & \qw & \gate{\mathrm{W}} & \qw \\
 }
\hspace{5mm}\raisebox{-5.5mm}{=}\hspace{0mm}
\Qcircuit @C=0.5em @R=0.0em @!R { 
	 	\nghost{{C}_{1} :  } & \lstick{{C}_{1} :  } & \qw {/^{n_1}} & \ctrl{2} & \qw & \ctrl{2} & \qw & \qw \\
	 	\nghost{{C}_{2} :  } & \lstick{{C}_{2} :  } & \qw {/^{n_2}} & \qw & \ctrl{1} & \qw & \ctrl{1} & \qw \\
	 	\nghost{{t} :  } & \lstick{{t} :  } & \qw & \rpigate{\hat{v}_1} & \rpigate{\hat{v}_2} & \rpigate{\hat{v}_1} & \rpigate{\hat{v}_2} & \qw \\
 }
 }
 \qcref{mcrpi_base_base}
\]
\begin{lemma}\label{lem:ccrv_reg}
Any $\MCO{\rv{\lambda}{\hat{v}}}{C}{t}$ gate with $\rv{\lambda}{\hat{v}}\in SU(2)$ can be implemented as $\MCO{\rpi{\hat{v}_2}}{C_2}{t}
\MCO{\rpi{\hat{v}_1}}{C_1}{t}
\MCO{\rpi{\hat{v}_2}}{C_2}{t}
\MCO{\rpi{\hat{v}_1}}{C_1}{t}$, such that $C_1\cup C_2 = C$. $\hat{v}_1$ can be chosen as any unit vector perpendicular to $\hat{v}$, and $\hat{v}_2 = \hat{R}_{\hat{v}}(\frac{\lambda}{4})\hat{v}_1$.
\end{lemma}
\begin{proof}
We define the pair $(g_1,g_2)$, where $g_{j\in\{1,2\}}=1$ denotes $C_j$ in the computational basis state $\ket{11..1}$, and $g_j=0$ otherwise. 
The operator applied on the target qubit for each option of the pair $(g_1,g_2)$:
\[
\begin{aligned}
&(00): I,
&(01): (\rpi{\hat{v}_2})^2 = I,\\
&(11): (\rpi{\hat{v}_2}\rpi{\hat{v}_1})^2,
&(10): (\rpi{\hat{v}_1})^2 = I.
\end{aligned}
\]
Since $\hat{v}_1 \perp\hat{v}$ and $\hat{v}_2 = \hat{R}_{\hat{v}}(\frac{\lambda}{4})\hat{v}_1$, according to \lem{2_rpi},
$\rpi{\hat{v}_2}\rpi{\hat{v}_1} = \rv{\frac{\lambda}{2}}{\hat{v}}$,
and the operator applied on the target qubit for option $(11)$ is $(\rv{\frac{\lambda}{2}}{\hat{v}})^2=\rv{\lambda}{\hat{v}}$.
\end{proof}
{We could use \lem{Cpi_zz_yy} to decompose the multi-controlled $\Pi$ gates in \qc{mcrpi_base_base} as four multi-controlled Toffoli (MCT, also known as MCX) gates and single-qubit rotations, achieving a structure similar to the one used in \cite{vale_decomposition_2023} to implement a subset of MCSU2, while ours can implement any MCSU2. Following their analysis, this step alone provides a reduction of CNOT count from $20n$ to $16n$. However, we find it useful to allow for an implementation of these MCX gates as multi-controlled Z gates with an added relative phase, as in \qc{mcrpi_base}, resulting in a reduced Toffoli count which provides a further reduction of CNOT count.}
To minimize the number of arbitrary single-qubit rotations from $\{\rx,\rz\}$, it is beneficial to use \lem{ccrv_reg} in order to implement a multi-controlled $SU(2)$ rotation about the $\hat{x}$ axis and then transform the axis using \lem{MC_transform}. Identity gates in the form of relative phase gates $\sqO{\Delta}{C}$, and their inverse $\sqO{\Delta^\dagger}{C}$ can be added to produce the following { decomposition of $\MCO{\rv{\lambda}{\hat{v}}}{\{C_1,C_2\}}{t}$}.
\[
\resizebox{0.9\linewidth}{!}{
\Qcircuit @C=0.5em @R=0.4em @!R { 
	 	\nghost{{C}_{1} :  } & \lstick{{C}_{1} :  } & \qw {/^{n_1}}  & \ctrl{2} & \multigate{1}{\mathrm{\Delta_1}} & \qw  & \qw & \multigate{1}{\mathrm{\Delta_2}} & \qw & \ctrl{2} & \multigate{1}{\mathrm{\Delta_1^\dagger}} & \qw & \qw & \multigate{1}{\mathrm{\Delta_2^\dagger}} & \qw & \qw \\
	 	\nghost{{C}_{2} :  } & \lstick{{C}_{2} :  } & \qw {/^{n_2}} & \qw & \ghost{\mathrm{\Delta_1}} & \qw  & \ctrl{1} & \ghost{\mathrm{\Delta_2}} & \qw & \qw & \ghost{\mathrm{\Delta_1^\dagger}} & \qw & \ctrl{1} & \ghost{\mathrm{\Delta_2^\dagger}} & \qw & \qw \\
	 	\nghost{{t} :  } & \lstick{{t} :  }  & \gate{\mathrm{A_4}} & \ctrl{0} & \qw & \gate{\mathrm{A_2}} & \ctrl{0} & \qw & \gate{\mathrm{A_3}} & \ctrl{0} & \qw & \gate{\mathrm{A_2}} & \ctrl{0} & \qw & \gate{\mathrm{A_1}} & \qw 
        \gategroup{1}{4}{3}{5}{0.3em}{--} 
        \gategroup{1}{7}{3}{8}{0.3em}{--} 
        \gategroup{1}{10}{3}{11}{0.3em}{--} 
        \gategroup{1}{13}{3}{14}{0.3em}{--} 
 }
 \qcref{mcrpi_base}
 }
\]
The additional $\sqO{\Delta}{C}$ gates can be set as any relative phase gates. Decomposing each multi-controlled Z (MCZ) gate, combined with its corresponding $\sqO{\Delta}{C}$ in terms of basic gates, allows for an efficient implementation, as we show in \sec{decompose_1}.
\begin{lemma}\label{lem:ccrv_macro}
Any $\MCO{\rv{\lambda}{\hat{v}}}{C}{t}$ gate with $\rv{\lambda}{\hat{v}}\in SU(2)$ can be implemented with eight gates from $\{\rx,\rz\}$, two $\MCO{Z}{C_1}{t}$ and two $\MCO{Z}{C_2}{t}$ gates as \qc{mcrpi_base} such that $A_2=\rv{-\frac{\lambda}{4}}{\hat{x}}$, $A_3=\rv{\frac{\lambda}{4}}{\hat{x}}$, $A_4=\rv{\theta_2}{\hat{z}}\rv{\theta_1}{\hat{x}}$, $A_1={A_4}^\dagger A_3$, and $C_1\cup C_2 = C$. The  $\sqO{\Delta}{C}$ gates apply a relative phase, and the angles $\theta_1,\theta_2$ are defined by the axis $\hat{v}$.

\end{lemma}
\begin{proof}
   Using \lem{MC_transform} with $\psi = 0$ and $\hat{v}_M\in M(\hat{v},\hat{x})$, we can write 
\[
\MCO{\rv{\lambda}{\hat{v}}}{C}{t}
=
\sqO{\rpi{\hat{v}_{M}}}{t}
\MCO{\rv{\lambda}{\hat{x}}}{C}{t}
\sqO{\rpi{\hat{v}_{M}}}{t}.
\]

The gate $\rpi{\hat{v}_{M}}$ can be decomposed up to a global phase using the $XZX$ Euler angles as
$\rv{\theta_3}{\hat{x}}\rv{\theta_2}{\hat{z}}\rv{\theta_1}{\hat{x}} = \rv{\theta_3}{\hat{x}}A_4$, and since $\rpi{\hat{v}_{M}}$ is hermitian, it can be implemented as $A_4^\dagger\rv{-\theta_3}{\hat{x}}$ as well. 
Substituting and applying cancellations provides
\[
\MCO{\rv{\lambda}{\hat{v}}}{C}{t}
=
\sqO{A_4^\dagger}{t}
\MCO{\rv{\lambda}{\hat{x}}}{C}{t}
\sqO{A_4}{t}.
\]
From \lem{ccrv_reg} with $\hat{v}_1=\hat{z}$ and $\hat{v}_2= \hat{R}_{\hat{x}}(\frac{\lambda}{4})\hat{z}$ we get
\[\MCO{\rv{\lambda}{\hat{x}}}{C}{t} = \MCO{\rpi{\hat{v}_2}}{C_2}{t}
\MCO{Z}{C_1}{t}
\MCO{\rpi{\hat{v}_2}}{C_2}{t}
\MCO{Z}{C_1}{t}.
\]
Finally, each $\MCO{\rpi{\hat{v}_2}}{C_2}{t}$ gate can be implemented as 
\[\MCO{\rpi{\hat{v}_2}}{C_2}{t}
=
\sqO{\rv{\tfrac{\lambda}{4}}{\hat{x}}}{t}
\MCO{Z}{C_2}{t}
\sqO{\rv{-\tfrac{\lambda}{4}}{\hat{x}}}{t}
= 
\sqO{A_3}{t}
\MCO{Z}{C_2}{t}
\sqO{A_2}{t}
\]
according to \lem{Cpi_zz_yy}.
 Since relative phase gates commute with each other and with $\MCO{Z}{C_1}{t},\MCO{Z}{C_2}{t}$, the effects of the additional $\sqO{\Delta}{C}$ gates cancel out.
\end{proof}
    
\subsubsection{MCSU2 decomposition}\label{sec:decompose_1}

In this section, we provide a decomposition of $\sqO{\Delta_1}{C}\MCO{Z}{C_1}{t}$, $\sqO{\Delta_2}{C}\MCO{Z}{C_2}{t}$ and their inverse {(the boxed gates in \qc{mcrpi_base})} in terms of Clifford+T gates, as required for the implementation described in \lem{ccrv_macro}. The following lemma will be used to develop our structure.

\begin{lemma}\label{lem:basic_struct}
    For two sets of qubits $A,B$ and a qubit $d\not\in A\cup B$, the following holds
    $\MCO{\rpi{\hat{v}_2}}{B}{d}\MCO{\rpi{\hat{v}_1}}{A}{d}\MCO{\rpi{\hat{v}_2}}{B}{d} = \MCO{Z}{\{A,B\}}{}\MCO{\rpi{\hat{v}_1}}{A}{d}$ if the unit vectors $\hat{v}_1$ and $\hat{v}_2$ are perpendicular. 
\end{lemma}
\begin{proof}
For any choice of two perpendicular vectors $\hat{v}_1\perp \hat{v}_2$, a unit vector $\hat{v}$ perpendicular to both can be defined as $ \hat{v} = \hat{v}_1\times\hat{v}_2
$ such that $\hat{v}_2 = \hat{R}_{\hat{v}}(\frac{\pi}{2})\hat{v}_1$. Therefore, according to \lem{ccrv_reg},
    \[
\MCO{\rpi{\hat{v}_2}}{B}{d}
\MCO{\rpi{\hat{v}_1}}{A}{d}
\MCO{\rpi{\hat{v}_2}}{B}{d}
\MCO{\rpi{\hat{v}_1}}{A}{d}
=
\MCO{\rv{2\pi}{\hat{v}}}{\{A,B\}}{d}.
\]
We note that for any unit vector $\hat{v}$,  $ \MCO{\rv{2\pi}{\hat{v}}}{\{A,B\}}{d} =\MCO{-I}{\{A,B\}}{d}$. This operator merely applies a phase of $-1$ if all qubits in set $A\cup B$ are in the computational basis state $\ket{1}$, and does not change the state of the target qubit $d$. Therefore, it is equivalent to applying a multi-controlled $Z$ gate on set $A\cup B$, and we can write $\MCO{-I}{\{A,B\}}{d} = \MCO{Z}{\{A,B\}}{}$ as mentioned in \cite{he_decompositions_2017}. By substituting this into the equation above and applying the Hermitian gate $\MCO{\rpi{\hat{v}_1}}{A}{d}$ on the right-hand side, we get the required equation.
\end{proof}
For our construction, we use \lem{basic_struct} with $\hat{v}_1=\hat{z},\hat{v}_2=\hat{x}$, and set $B=\{c,d'\}$, which results in
$
\MCO{X}{\{c,d'\}}{d}
\MCO{Z}{A}{d}
\MCO{X}{\{c,d'\}}{d} = 
\MCO{Z}{\{A,c\}}{d'}\MCO{Z}{A}{d}
$ as follows.
\[
\Qcircuit @C=1.0em @R=0.2em @!R { 
	 	\nghost{{A} :  } & \lstick{{A} :  } & \qw {/^n}  & \ctrl{2} & \qw & \qw \\
	 	\nghost{{c} :  } & \lstick{{c} :  }  & \ctrl{1} & \qw & \ctrl{1} & \qw \\
	 	\nghost{{d} :  } & \lstick{{d} :  } & \targ & \control\qw & \targ & \qw \\
	 	\nghost{{d'} :  } & \lstick{{d'} :  }  & \ctrl{-1} & \qw & \ctrl{-1} & \qw \\
 }
 \hspace{5mm}\raisebox{-6mm}{=}\hspace{0mm}
 \Qcircuit @C=1.0em @R=0.2em @!R { 
	 	\nghost{{A} :  } & \lstick{{A} :  } & \qw {/^n} & \ctrl{1} & \ctrl{2} & \qw \\
	 	\nghost{{c} :  } & \lstick{{c} :  } & \qw & \ctrl{2} & \qw & \qw \\
	 	\nghost{{d} :  } & \lstick{{d} :  } & \qw & \qw & \control\qw & \qw \\
	 	\nghost{{d'} :  } & \lstick{{d'} :  } & \qw & \control\qw & \qw & \qw \\
 }
 \qcref{add_mcz}
\]
We introduce the following notation:
\[
\{ Z\}^m_{C,D} = \prod_{j=2}^m \MCO{Z}{C[1:j]}{d_{j-1}}
\]
where $C=\{c_1,c_2,..,c_n\}$ and $D=\{d_1,d_2,..,d_{n-1}\}$
are two qubit sets, and $m$ can be chosen such that $2\leq m\leq n$.
In addition, we use the $\MCO{X_\Delta}{\{q_1,q_2\}}{q_3}$ notation for an hermitian relative phase Toffoli gate which satisfies
\[
\resizebox{0.9\linewidth}{!}{
\Qcircuit @C=1.0em @R=0.8em @!R { 
	 	\nghost{{q}_{1} :  } & \lstick{{q}_{1} :  } & \ctrlt{1} & \qw \\
	 	\nghost{{q}_{3} :  } & \lstick{{q}_{3} :  } & \targ & \qw \\
	 	\nghost{{q}_{2} :  } & \lstick{{q}_{2} :  } & \ctrl{-1} & \qw \\
 }
 \hspace{5mm}\raisebox{-7mm}{=}\hspace{0mm}
\Qcircuit @C=1.0em @R=0.4em @!R { 
	 	\nghost{{q}_{1} :  } & \lstick{{q}_{1} :  } & \ctrl{1} & \qw & \ctrl{2} & \qw \\
	 	\nghost{{q}_{3} :  } & \lstick{{q}_{3} :  } & \targ & \control\qw & \qw & \qw \\
	 	\nghost{{q}_{2} :  } & \lstick{{q}_{2} :  } & \ctrl{-1} & \ctrl{-1} & \gate{\mathrm{S}} & \qw \\
 }
 \hspace{5mm}\raisebox{-7mm}{=}\hspace{0mm}
 \Qcircuit @C=1.0em @R=0.3em @!R { 
	 	\nghost{{q}_{1} :  } & \lstick{{q}_{1} :  } & \ctrl{2} & \qw & \ctrl{1} & \qw \\
	 	\nghost{{q}_{3} :  } & \lstick{{q}_{3} :  } & \qw & \control\qw & \targ & \qw \\
	 	\nghost{{q}_{2} :  } & \lstick{{q}_{2} :  } & \gate{\mathrm{S^\dagger}} & \ctrl{-1} & \ctrl{-1} & \qw \\
 }
 }
 \qcref{rf_toff_first}
\]
This gate notation allows us to distinguish between the control qubits. We use a symbol which combines the $\ket{0}$-controlled and $\ket{1}$-controlled symbols, noting that for each of these basis states of $q_1$, a different controlled operator is effectively applied on $q_2$ and $q_3$.

We now show that two $[X_\Delta]$ gates can be used to transform $\{ Z\}^{m-1}_{C,D}$ to $\{ Z\}^m_{C,D}$.
\begin{lemma} \label{lem:rec_struct}
$\{ Z\}^m_{C,D} = \MCO{X_\Delta}{\{c_m,d_{m-1}\}}{d_{m-2}}\{ Z\}^{m-1}_{C,D}\MCO{X_\Delta}{\{c_m,d_{m-1}\}}{d_{m-2}} $
    for $3\leq m \leq \abs{C}$, and $\abs{D}\geq m-1$.
\end{lemma} 
\begin{proof}
We prove for the following equivalent equation, in which $[X_{\Delta}]$ are replaced with Toffoli gates :
\[
\{ Z\}^m_{C,D} = \MCO{X}{\{c_m,d_{m-1}\}}{d_{m-2}}\{ Z\}^{m-1}_{C,D}\MCO{X}{\{c_m,d_{m-1}\}}{d_{m-2}}
\]
where the equivalence can be easily verified by substituting each $[X_\Delta]$ gate with its definition which places all of the relative phase gates at the center, then these are commuted and cancelled out.\\  
For $m=3$:
\[
\begin{aligned}
\MCO{X}{\{c_3,d_{2}\}}{d_{1}}&\{ Z\}^{2}_{C,D}\MCO{X}{\{c_3,d_{2}\}}{d_{1}} 
= 
\\
&\MCO{X}{\{c_3,d_{2}\}}{d_{1}}\MCO{Z}{\{c_1,c_2\}}{d_1}\MCO{X}{\{c_3,d_{2}\}}{d_{1}}
\stackrel{*}{=}
\\
&\MCO{Z}{\{c_1,c_2\}}{d_1}\MCO{Z}{\{c_1,c_2,c_3\}}{d_2}
=
\{ Z\}^3_{C,D}
\end{aligned}
\]
where $*$ is according to \lem{basic_struct}.\\
For $m\geq 4$, by definition:
    \[
        \{ Z\}^{m-1}_{C,D} =
        \MCO{Z}{C[1:m-1]}{d_{m-2}}
        \{ Z\}^{m-2}_{C,D}.
    \]
    Since $c_m,d_{m-1},d_{m-2}$ are not accessed in $\{ Z\}^{m-2}_{C,D}$, this gate commutes with $\MCO{X}{\{c_m,d_{m-1}\}}{d_{m-2}}$. Therefore:
    \[
    \begin{aligned}
    &\MCO{X}{\{c_m,d_{m-1}\}}{d_{m-2}}
    \{ Z\}^{m-1}_{C,D}
    \MCO{X}{\{c_m,d_{m-1}\}}{d_{m-2}}
    =\\
    &\MCO{X}{\{c_m,d_{m-1}\}}{d_{m-2}}
    \MCO{Z}{C[1:m-1]}{d_{m-2}}
    \MCO{X}{\{c_m,d_{m-1}\}}{d_{m-2}}
    \{ Z\}^{m-2}_{C,D}.
    \end{aligned}
    \]
    According to \lem{basic_struct}:
    \[
    \begin{aligned}
    \MCO{X}{\{c_m,d_{m-1}\}}{d_{m-2}}
    \MCO{Z}{C[1:m-1]}{d_{m-2}}
    \MCO{X}{\{c_m,d_{m-1}\}}{d_{m-2}}
    &=\\
    &\MCO{Z}{C[1:m]}{d_{m-1}}
    \MCO{Z}{C[1:m-1]}{d_{m-2}}.
    \end{aligned}
    \]
    Since $C[1:m-1]\cup \{c_m\} = C[1:m]$.
 Finally, we can substitute and use 
    \[
    \MCO{Z}{C[1:m]}{d_{m-1}}
    \MCO{Z}{C[1:m-1]}{d_{m-2}}
    \{ Z\}^{m-2}_{C,D}
    =
    \{ Z\}^{m}_{C,D}.
    \]
\end{proof}

The following lemma is simply achieved by repeatedly applying the recursive formula presented in \lem{rec_struct}, and as a final step, applying $\MCO{S}{c_1}{c_2}$ to both sides. This added gate commutes with any other gate in the circuit and transforms $\MCO{Z}{\{c_1,c_2\}}{d_1}$ to $\MCO{iZ}{\{c_1,c_2\}}{d_1}$:

\begin{lemma}\label{lem:Vchain_fullstruct}
$\{ Z\}^{n}_{C,D}\MCO{S}{c_1}{c_2} = (\prod_{j=3}^{n} \MCO{X_\Delta}{\{c_j,d_{j-1}\}}{d_{j-2}})^\dagger
\MCO{iZ}{\{c_1,c_2\}}{d_1}
(\prod_{j=3}^{n} \MCO{X_\Delta}{\{c_j,d_{j-1}\}}{d_{j-2}})$
for $n\geq 3$, $\abs{C}= n$, and $\abs{D}= n-1$.
\end{lemma}
The following circuits describes the decomposition of $\{ Z\}^5_{C,D}\MCO{S}{c_1}{c_2}$ in terms of $\MCO{X_\Delta}{\{q_1,q_2\}}{q_3}$ and $\MCO{iZ}{\{q_1,q_2\}}{q_3}$ gates, or $[X_\Delta]$ and $[iZ]$ for short:
\[
\resizebox{0.9\linewidth}{!}{
\Qcircuit @C=1.0em @R=0.38em @!R {
	 	\nghost{{c}_{1} :  } & \lstick{{c}_{1} :  } & \ctrl{1}  & \multigate{7}{\mathrm{\Delta}}  & \qw \\
	 	\nghost{{c}_{2} :  } & \lstick{{c}_{2} :  } & \ctrl{1} & \ghost{\mathrm{\Delta}}  & \qw \\
	 	\nghost{{c}_{3} :  } & \lstick{{c}_{3} :  } & \ctrl{1} & \ghost{\mathrm{\Delta}}  & \qw \\
	 	\nghost{{c}_{4} :  } & \lstick{{c}_{4} :  } & \ctrl{1} & \ghost{\mathrm{\Delta}}  & \qw \\
	 	\nghost{{c}_{5} :  } & \lstick{{c}_{5} :  } & \ctrl{4} & \ghost{\mathrm{\Delta}} & \qw \\
	 	\nghost{{d}_{1} :  } & \lstick{{d}_{1} :  } & \qw  & \ghost{\mathrm{\Delta}} & \qw \\
	 	\nghost{{d}_{2} :  } & \lstick{{d}_{2} :  } & \qw & \ghost{\mathrm{\Delta}}  & \qw \\
	 	\nghost{{d}_{3} :  } & \lstick{{d}_{3} :  } & \qw & \ghost{\mathrm{\Delta}} & \qw \\
	 	\nghost{{d}_{4} :  } & \lstick{{d}_{4} :  } & \control\qw  & \qw & \qw  \\
}
   \hspace{5mm}\raisebox{-18mm}{=}\hspace{0mm}
\Qcircuit @C=1.0em @R=0.em @!R {
	 	\nghost{{c}_{1} :  } & \lstick{{c}_{1} :  } & \ctrl{1} \barrier[0em]{8} & \qw & \ctrl{1} & \ctrl{1} & \ctrl{1} & \ctrl{1} & \qw \\
	 	\nghost{{c}_{2} :  } & \lstick{{c}_{2} :  } & \ctrl{1} & \qw & \ctrl{1} & \ctrl{1} & \ctrl{4} &\gate{\mathrm{S}} & \qw \\
	 	\nghost{{c}_{3} :  } & \lstick{{c}_{3} :  } & \ctrl{1} & \qw & \ctrl{1} & \ctrl{4} & \qw & \qw & \qw \\
	 	\nghost{{c}_{4} :  } & \lstick{{c}_{4} :  } & \ctrl{1} & \qw & \ctrl{4} & \qw & \qw & \qw & \qw\\
	 	\nghost{{c}_{5} :  } & \lstick{{c}_{5} :  } & \ctrl{4} & \qw & \qw & \qw & \qw & \qw & \qw\\
	 	\nghost{{d}_{1} :  } & \lstick{{d}_{1} :  } & \qw & \qw & \qw & \qw & \control\qw & \qw & \qw\\
	 	\nghost{{d}_{2} :  } & \lstick{{d}_{2} :  } & \qw & \qw & \qw & \control\qw & \qw & \qw & \qw\\
	 	\nghost{{d}_{3} :  } & \lstick{{d}_{3} :  } & \qw & \qw & \control\qw & \qw & \qw & \qw & \qw\\
	 	\nghost{{d}_{4} :  } & \lstick{{d}_{4} :  } & \control\qw & \qw & \qw & \qw & \qw & \qw & \qw\\
}
\hspace{5mm}\raisebox{-18mm}{=}\hspace{0mm}
   \Qcircuit @C=0.7em @R=0.em @!R { 
	 	\nghost{{c}_{1} :  } & \lstick{{c}_{1} :  } & \qw & \qw & \qw & \ctrl{1} & \qw & \qw & \qw & \qw \\
	 	\nghost{{c}_{2} :  } & \lstick{{c}_{2} :  } & \qw & \qw & \qw & \ctrl{4} & \qw & \qw & \qw & \qw \\
	 	\nghost{{c}_{3} :  } & \lstick{{c}_{3} :  } & \qw & \qw & \ctrlt{3} & \qw & \ctrlt{3} & \qw & \qw & \qw \\
	 	\nghost{{c}_{4} :  } & \lstick{{c}_{4} :  } & \qw & \ctrlt{3} & \qw & \qw & \qw & \ctrlt{3} & \qw & \qw \\
	 	\nghost{{c}_{5} :  } & \lstick{{c}_{5} :  } & \ctrlt{3} & \qw & \qw & \qw & \qw & \qw & \ctrlt{3} & \qw \\
	 	\nghost{{d}_{1} :  } & \lstick{{d}_{1} :  } & \qw & \qw & \targ & \gate{\mathrm{iZ}} & \targ & \qw & \qw & \qw \\
	 	\nghost{{d}_{2} :  } & \lstick{{d}_{2} :  } & \qw & \targ & \ctrl{-1} & \qw & \ctrl{-1} & \targ & \qw & \qw \\
	 	\nghost{{d}_{3} :  } & \lstick{{d}_{3} :  } & \targ & \ctrl{-1} & \qw & \qw & \qw & \ctrl{-1} & \targ & \qw \\
	 	\nghost{{d}_{4} :  } & \lstick{{d}_{4} :  } & \ctrl{-1} & \qw & \qw & \qw & \qw & \qw & \ctrl{-1} & \qw \\
 }
 }
 \qcref{v_chain_z_rel_phase}
\]
This structure requires $2n-4$ $[X_\Delta]$ gates and one $[iZ]$ gate. { It can be applied using the Toffoli+H gate set by replacing $[X_\Delta]\rightarrow [X]$ and $[iZ]\rightarrow [Z]$, providing a MCSU2 implementation using $4n+O(1)$ Toffoli gates when applied in \qc{mcrpi_base}, in addition to $O(1)$ arbitrary single qubit gates. Compared to the implementation of multi-controlled $R_{x/y/z}\in SU(2)$ gates in \cite{vale_decomposition_2023}, this construction uses half the number of Toffoli gates while implementing any MCSU2. It can be noted that the structure in \qc{v_chain_z_rel_phase} resembles in shape to the first half of the structure used in \cite{vale_decomposition_2023} for 
 a similar purpose in the MCSU2 construction. This Toffoli-based structure may be applied directly in some architectures\cite{jandura_time-optimal_2022}, or further decomposed in terms of single/two-qubit gates from gate sets such as NCV. As we focus on Clifford+T implementations we will decompose these gates using $T,T^\dagger,H$ and CNOT gates.}
We use the following decompositions for $\MCO{X_\Delta}{\{q_1,q_2\}}{q_3}$ \cite{maslov_advantages_2016}, and $\MCO{iZ}{\{q_1,q_2\}}{q_3}$.

\[
\resizebox{0.9\linewidth}{!}{
\Qcircuit @C=1.0em @R=0.8em @!R { 
	 	\nghost{{q}_{1} :  } & \lstick{{q}_{1} :  } & \ctrlt{1} & \qw \\
	 	\nghost{{q_3} :  } & \lstick{{q_3} :  } & \targ & \qw \\
	 	\nghost{{q}_{2} :  } & \lstick{{q}_{2} :  } & \ctrl{-1} & \qw \\
 }
 \hspace{5mm}\raisebox{-6mm}{=}\hspace{0mm}
 \Qcircuit @C=1.0em @R=0.2em @!R { 
	 	\nghost{{q}_{1} :  } & \lstick{{q}_{1} :  } & \qw & \qw & \ctrl{1}  & \qw\barrier[0em]{2}  & \qw & \qw & \qw & \ctrl{1} & \qw & \qw & \qw \\
	 	\nghost{{q_3} :  } & \lstick{{q_3} :  } & \gate{\mathrm{H}} & \gate{\mathrm{T^\dagger}} & \targ  & \gate{\mathrm{T}} & \qw & \targ & \gate{\mathrm{T^\dagger}} & \targ & \gate{\mathrm{T}} & \gate{\mathrm{H}} & \qw \\
	 	\nghost{{q}_{2} :  } & \lstick{{q}_{2} :  } & \qw & \qw & \qw & \qw & \qw & \ctrl{-1}  & \qw & \qw & \qw & \qw & \qw \\
 }
 }
 \qcref{rf_toffoli}
\]
\[
\resizebox{0.9\linewidth}{!}{
\Qcircuit @C=1.0em @R=0.1em @!R { 
	 	\nghost{{q}_{1} :  } & \lstick{{q}_{1} :  } & \ctrl{1} & \qw \\
	 	\nghost{{q}_{2} :  } & \lstick{{q}_{2} :  } & \ctrl{1} & \qw \\
	 	\nghost{{q_3} :  } & \lstick{{q_3} :  } & \gate{\mathrm{iZ}} & \qw \\
 }
\hspace{5mm}\raisebox{-5mm}{=}\hspace{0mm}
\Qcircuit @C=1.0em @R=0.em @!R { 
	 	\nghost{{q}_{1} :  } & \lstick{{q}_{1} :  } & \qw & \ctrl{2} & \qw & \qw & \qw & \ctrl{2} & \qw & \qw & \qw \\
	 	\nghost{{q}_{2} :  } & \lstick{{q}_{2} :  } & \qw & \qw & \qw & \ctrl{1} & \qw & \qw & \qw & \ctrl{1} & \qw \\
	 	\nghost{{q_3} :  } & \lstick{{q_3} :  } & \gate{\mathrm{T^\dagger}} & \targ & \gate{\mathrm{T}} & \targ & \gate{\mathrm{T^\dagger}} & \targ & \gate{\mathrm{T}} & \targ & \qw \\
 }
 }
 \qcref{iz_cliff_T}
\]
Each $[X_\Delta]$ requires two H, three CNOT and four T gates.
It can be seen that, when applied as part of \lem{Vchain_fullstruct}, the gates to the left of the barrier in \qc{rf_toffoli} can commute with any gate that is executed before. This, along with additional similar commutations, allows for depth reductions, as shown in \apx{deph_reduce}. A simple summation provides the following lemma.
\begin{lemma}\label{lem:vchainz_cost_depth}
A $\{ Z\}^{n}_{C,D}\MCO{S}{c_1}{c_2}$ gate with $n\geq 3$, $\abs{C}= n$, and $\abs{D}= n-1$ can be implemented with \\
CNOT cost $6n-8$ ~and depth $4n-2$,\\
T ~~~~~ cost $8n-12$ and depth $4n$,\\
H ~~~~~ cost $4n-8$ ~and depth $2n-2$.
\end{lemma}
We wish to use this to decompose the structure described in \lem{ccrv_macro} which implements any $\MCO{\rv{\lambda}{\hat{v}}}{C}{t}$ (MCSU2) gate. 
We choose $\sqO{\Delta_1}{C}=\MCO{S}{C_1[1:2]}{}\{Z\}^{n_1-1}_{C_1,C_2'}$ and $\sqO{\Delta_2}{C}=\MCO{S}{C_2[1:2]}{}\{Z\}^{n_2-1}_{C_2,C_1'}$ as relative phase gates applied on the set $C$ of size $n$.
For this choice, we get
\[
\begin{aligned}
&\sqO{\Delta_1}{C}\MCO{Z}{C_1}{t}=\MCO{S}{C_1[1:2]}{}\{Z\}^{n_1}_{C_1,\{C_2',t\}}\text{, and   
 }\\
&\sqO{\Delta_2}{C}\MCO{Z}{C_2}{t}=\MCO{S}{C_2[1:2]}{}\{Z\}^{n_2}_{C_2,\{C_1',t\}}.
\end{aligned}
\]
The inverted versions of these gates are implemented in a similar way.

The subsets $C_1,C_2,C_1'$ and $C_2'$ are of size $n_1,n_2, n'_1$ and $n'_2$ respectively, and satisfy $C_1\cup C_2=C$, $C_1'\in C_1$, $C_2'\in C_2$, $n'_1=n_2-2$, $n'_2=n_1-2$ and $n_1+n_2=n$. We can set $n_1=\floor{\frac{n}{2}}$ and $n_2=n-n_1=\ceil{\frac{n}{2}}$.

 These structures can be used to decompose \qc{mcrpi_base}.
Each cost and depth from \lem{vchainz_cost_depth} can be written as $an+b$, and therefore the counterparts for MCSU2 equal to $2(an_1+b+an_2+b)=2an+4b$. As shown in \apx{deph_reduce}, the sets $C_1',C_2'$ can be chosen such that additional depth reductions are achieved. The following theorem presents the total costs and depths.
\begin{theorem}\label{thm:MCSU2_costs}
    Any multi-controlled $SU(2)$ gate with $n\geq 6$ controls can be implemented without ancilla using eight gates from $\{\rx,\rz\}$ in addition to\\
CNOT cost $12n-32$ and depth $8n-8$,\\
T ~~~~~ cost $16n-48$ and depth $8n-6$ ($8n-3$ for odd $n$),\\
H ~~~~~ cost $8n-32$ ~and depth $4n-11$.
\end{theorem}
We note that if $2 \leq n \leq 5$, the costs and depths can be easily found using \lem{ccrv_macro}, such that \lem{Vchain_fullstruct} is used for a subset of controls of size $3$, \qc{iz_cliff_T} is used for size $2$, and a CZ (as H-CNOT-H) gate is used for size $1$. The number of single-qubit gates can be reduced in some of these cases.

Moreover, we show in \apx{T_depth} that a trade-off between the depth of CNOT and T gates can be achieved.
With no change to any cost, it is possible to reduce the depth of T gates to $5n+O(1)$, while increasing the depth of CNOT gates to $11n+O(1)$.
The T depth can be further reduced to $4n$, however, this results in a CNOT cost and depth of $12.5n+O(1)$ and $12n+O(1)$ respectively.

In \apx{extra_dirty} we show that if a set $\chi$ of $n_{\chi}\leq \floor{\frac{n-6}{2}}$ dirty ancilla qubits is available, {\em each} ancilla can be used to reduce the cost of Clifford+T gates, as presented in \fig{fig_ancil_CNOT_cost}. If $n_{\chi}= \floor{\frac{n-6}{2}}$, the cost of both the T and CNOT gates is reduced to $8n+O(1)$.

{ Moreover, \apx{MCMTSU2} provides a method that extends our structure and allows multiple target qubits (MCMTSU2).}

\subsection{Multi-controlled X}

In this section, we use the MCSU2 structure in order to efficiently implement the Multi-controlled $X$ gate with a single or multiple targets, and with a single dirty ancilla qubit. { In this case, the dirty ancilla is required for an exact implementation of any MCX gate with $n\geq 3$ controls using only Clifford+T gates, from the determinant argument stated in \cite{maslov_advantages_2016}.}

In order to implement a $\MCO{X}{C}{t}$ gate with one dirty ancilla qubit $a$, we use the {MCSU2} implementation of $\MCO{\rv{\lambda}{\hat{v}}}{C'}{t'}$ with $\lambda=2\pi$, $C'=\{C,t\}$, and $t'=a$.
As previously shown, $\MCO{\rv{2\pi}{\hat{v}}}{\{C,t\}}{a}=\MCO{-I}{\{C,t\}}{a} =\MCO{Z}{C}{t}$ for any choice of $\hat{v}$, and simply using the well-known special case of \lem{MC_transform}, we get $\MCO{X}{C}{t} = \sqO{H}{t}\MCO{\rv{2\pi}{\hat{v}}}{\{C,t\}}{a}\sqO{H}{t}$.
Setting $\hat{v}=\hat{y}$ and implementing $\MCO{\rv{2\pi}{\hat{y}}}{\{C,t\}}{a}$ according to \lem{ccrv_reg}, with $\hat{v}_1=\hat{z}$, and $\hat{v}_2=\hat{x}$ is a convenient choice. The relative phase $[\Delta]$ gates are added similarly to \lem{ccrv_macro}. This provides the following circuit.
\[
\resizebox{0.9\linewidth}{!}{$
\begin{aligned}
&\Qcircuit @C=1.0em @R=0.8em @!R { 
  \nghost{{C}_{1} :  } & \lstick{{C}_{1} :  } & \qw {/^{n_1}} & \ctrl{1} & \qw \\
	 	\nghost{{C}_{2} :  } & \lstick{{C}_{2} :  } & \qw {/^{n_2}} & \ctrl{1} & \qw \\
	 	\nghost{{t} :  } & \lstick{{t} :  } & \qw & \targ & \qw \\
	 	\nghost{{a} :  } & \lstick{{a} :  } & \qw & \qw & \qw \\
 }
 \hspace{5mm}\raisebox{-9mm}{=}\hspace{0mm}
 \Qcircuit @C=0.5em @R=0.1em @!R { 
  \nghost{{C}_{1} :  } & \lstick{{C}_{1} :  } & \qw {/^{n_1}} & \ctrl{1} & \qw & \qw \\
	 	\nghost{{C}_{2} :  } & \lstick{{C}_{2} :  } & \qw {/^{n_2}}  & \ctrl{1} & \qw & \qw \\
	 	\nghost{{t} :  } & \lstick{{t} :  }  & \gate{\mathrm{H}} & \ctrl{1} & \gate{\mathrm{H}} & \qw \\
	 	\nghost{{a} :  } & \lstick{{a} :  }  & \qw & \rvgate{2\pi}{\hat{y}} & \qw & \qw \\
 }
  \hspace{5mm}\raisebox{-9mm}{=}\hspace{0mm}\\
&\Qcircuit @C=0.5em @R=0.5em @!R { 
  \nghost{{C}_{1} :  } & \lstick{{C}_{1} :  } & \qw {/^{n_1}} & \ctrl{3} & \multigate{2}{\mathrm{\Delta_1}} & \qw & \multigate{2}{\mathrm{\Delta_2}} & \ctrl{3} & \multigate{2}{\mathrm{\Delta_1^\dagger}} & \qw & \multigate{2}{\mathrm{\Delta_2^\dagger}} & \qw & \qw \\
	 	\nghost{{C}_{2} :  } & \lstick{{C}_{2} :  } & \qw {/^{n_2}} & \qw & \ghost{\mathrm{\Delta_1}} & \ctrl{1} & \ghost{\mathrm{\Delta_1}} & \qw & \ghost{\mathrm{\Delta_1}} & \ctrl{1} & \ghost{\mathrm{\Delta_1}} & \qw & \qw \\
	 	\nghost{{t} :  } & \lstick{{t} :  }  & \gate{\mathrm{H}} & \qw & \ghost{\mathrm{\Delta_1}} & \ctrl{1} & \ghost{\mathrm{\Delta_1}} & \qw & \ghost{\mathrm{\Delta_1}} & \ctrl{1} & \ghost{\mathrm{\Delta_1}} & \gate{\mathrm{H}} & \qw \\
	 	\nghost{{a} :  } & \lstick{{a} :  } & \qw & \control\qw & \gate{\mathrm{H}} & \control\qw & \gate{\mathrm{H}} & \control\qw & \gate{\mathrm{H}} & \control\qw & \gate{\mathrm{H}} & \qw & \qw \\
 }
 \end{aligned}$
 }
 \qcref{mcx_ata_macro}
\]
The cost and depth of this implementation is similar to the multi-controlled $SU(2)$ gate with $n+1$ controls. Arbitrary $\{\rx,\rz\}$ gates are removed and six H gates are added instead.
Notice that the control set $C'$ can be ordered such that two H gates are cancelled out in the full Clifford+T decomposition.
\thm{MCX_costs_short} simply follows.

\begin{theorem}\label{thm:MCX_costs_short}
    Any multi-controlled $X$ gate with $n\geq 5$ controls and a single target can be implemented with one dirty ancilla qubit at the same costs and depths in \thm{MCSU2_costs}, with $n$ replaced by $n+1$ and four additional H gates. Arbitrary $\{\rz,\rx\}$ gates are not needed.
\end{theorem}

{As mentioned above, the MCSU2 gate can also be implemented using the Toffoli+H gate set with only $4n+O(1)$ Toffoli gates, and the same holds for the MCX with one dirty ancilla.  {We thus provide a Toffoli count reduction of roughly $50\%$ for the first time compared to the best known count of $8n+O(1)$ which was introduced in 1995 by Barenco et al. \cite{barenco_elementary_1995}.} Furthermore, our Toffoli count scales similarly to the best known method which requires $O(n)$ dirty ancilla \cite{maslov_advantages_2016}.
Although the number of Toffoli gates is the same, each of our Toffoli gates requires three CNOT and four T gates, while $n+O(1)$ dirty ancilla allows each Toffoli to be implemented at a reduced cost using two CNOT and two T gates \cite{maslov_advantages_2016}. In \apx{extra_dirty} we show that, starting from our one dirty ancilla implementation, each added ancilla effectively allows several Toffoli gates to be applied at this reduced cost, and once $n+O(1)$ ancillae are added, we reach the same cost scaling for CNOT and T gates as \cite{maslov_advantages_2016}. Then, three-controlled Toffoli gates are introduced to double the reduction allowed by each ancilla, reaching the same gate count of $8n+O(1)$ CNOT and T gates using only $\tfrac{n}{2}+O(1)$ dirty ancillae, as shown in \cite{maslov_advantages_2016} as well. Although we were unable to reduce the cost compared to $O(n)$ ancilla methods, this provides a bridge between $1$ and $\tfrac{n}{2}+O(1)$ dirty ancilla methods, such that any number of ancilla between those can also be used to reduce the cost (\fig{fig_ancil_CNOT_cost}), thus harnessing all available resources.
In \apx{ata_MCMTX} we extend the multi-controlled $X$ structure to a multi-control multi-target $X$ gate (MCMTX) such that if the number of target qubits is larger than one, an ancilla qubit is not required. Moreover, additional target qubits can also act as dirty ancilla and provide cost reductions due to the structure of \qc{mcmtX_a2a_targets}.}
  
    \subsection{Multi-controlled U(2)}

In this section, we use the MCMTSU2 structures in order to efficiently implement the multi-controlled $U(2)$ gate (MCU2) using one clean ancilla, with a single or multiple targets. { The clean ancilla, while being an expensive resource, seems necessary here since there are no known ways to achieve a linear count of Clifford+T gates without it. Known no-ancilla implementations require either a quadratic ($O(n^2)$) count of Clifford+T gates \cite{pedram_layout_2016,saeedi_synthesis_2011,he_mapping_2019,khan_cost_2008,lukac_optimization_2021,ding_fast_2019,zhang_method_2022}, or a linear count ($O(n)$) with a linear number of arbitrary single-qubit gates which require exponential precision \cite{gidney_constructing_nodate} since the angles of rotation scale as $O(2^{-n})$, making them hard to implement using finite gate sets.}

In order to implement any $\MCO{U}{C}{t}$ gate, with $U\in U(2)$ { we present a new trick which allows to use our MCMTSU2 decomposition with two targets and a clean ancilla to apply any MCU2 gate.} We recall that $U$ can be expressed as $U=e^{i\psi}\rv{\lambda}{\hat{v}}$, and therefore $\MCO{U}{C}{t} = \MCO{\rv{\lambda}{\hat{v}}}{C}{t} \MCO{P(\psi)}{C}{}$. We now show that $\MCO{U}{C}{t} = \MCO{\rv{\lambda}{\hat{v}}}{C}{t} \MCO{\rv{-2\psi}{\hat{z}}}{C}{a_{\ket{0}}}$ with $a_{\ket{0}}$ a clean ancilla qubit in state $\ket{0}$. { This follows directly from \lem{P_a_0}, which shows that any multi-controlled phase gate can be implemented as one MCSU2 using one clean ancilla.}
\begin{lemma}\label{lem:P_a_0}
$\MCO{P(\psi)}{C}{} = \MCO{\rz(-2\psi)}{C}{a_{\ket{0}}}$ for any angle $\psi$ with $C=\{c_1,c_2,..,c_n\}$, and $a_{\ket{0}}$ is a clean ancilla qubit in state $\ket{0}$.
\end{lemma}
\begin{proof}
The gate $\MCO{\rz(-2\psi)}{C}{a_{\ket{0}}}$ applies the following operator on the ancilla qubit iff set $C$ is in state $\ket{11..1}$:
\[
\rz(-2\psi)=\left(
\begin{matrix}
e^{i\psi} & 0\\
0 & e^{-i\psi}
\end{matrix}
\right) =e^{i\psi}\left(
\begin{matrix}
1 & 0\\
0 & e^{i(-2\psi)}
\end{matrix}
\right)
\]
Therefore, it can be written as $\MCO{\rz(-2\psi)}{C}{a_{\ket{0}}} = \MCO{\mathrm{P(-2\psi)}}{C}{a_{\ket{0}}}\MCO{\mathrm{P(\psi)}}{C}{}$. Since $a_{\ket{0}}$ is known to be in the state $\ket{0}$, the controlled gate $\MCO{\mathrm{P(-2\psi)}}{C}{a_{\ket{0}}}$ has no effect and can be removed.
\end{proof}

The gate $\MCO{\rv{2\psi}{\hat{z}}}{C}{a_{\ket{0}}}$ can be implemented using \lem{ccrv_macro}, with $A_4$ transforming the $\hat{z}$ axis to $\hat{x}$ and thus can be replaced with $H$. The gate $\MCO{\rv{-2\psi}{\hat{z}}}{C}{a_{\ket{0}}}$ can therefore be implemented as the inverse of this structure as:
\[
\resizebox{0.9\linewidth}{!}{
\Qcircuit @C=1.0em @R=0.2em @!R { 
	 	\nghost{{C_1} :  } & \lstick{{C_1} :  } & \qw {/^{n_1}} & \ctrl{1} & \qw \\
	 	\nghost{{C}_{2} :  } & \lstick{{C}_{2} :  } & \qw {/^{n_2}} & \ctrl{1} & \qw \\
	 	\nghost{{a_{\ket{0}}} :  } & \lstick{{a_{\ket{0}}} :  } & \qw & \rvgate{-2\psi}{\hat{z}} & \qw \\
 }
\hspace{5mm}\raisebox{-6mm}{=}\hspace{0mm}
\Qcircuit @C=1.0em @R=0.4em @!R { 
	 	\nghost{{C}_{1} :  } & \lstick{{C}_{1} :  } & \qw {/^{n_1}} & \qw  & \ctrl{2} & \multigate{1}{\mathrm{\Delta_1}}  & \qw & \multigate{1}{\mathrm{\Delta_2}} & \ctrl{2} & \multigate{1}{\mathrm{\Delta_1^\dagger}} & \qw & \multigate{1}{\mathrm{\Delta_2^\dagger}} & \qw \\
	 	\nghost{{C}_{2} :  } & \lstick{{C}_{2} :  } & \qw {/^{n_2}} & \qw & \qw & \ghost{\mathrm{\Delta_1}}  & \ctrl{1} & \ghost{\mathrm{\Delta_2}} & \qw & \ghost{\mathrm{\Delta_1^\dagger}} & \ctrl{1} & \ghost{\mathrm{\Delta_2^\dagger}} & \qw \\
	 	\nghost{{a_{\ket{0}}} :  } & \lstick{{a_{\ket{0}}} :  }  & \gate{\mathrm{H}} & \gate{\mathrm{B_2}} & \ctrl{0} & \gate{\mathrm{B_1}} & \ctrl{0} & \gate{\mathrm{B_2}} & \ctrl{0} & \gate{\mathrm{B_1}} & \ctrl{0} & \gate{\mathrm{H}} & \qw \\
 }
 }
 \qcref{mcp_base}
\]
with $B_1=\rv{\frac{\psi}{2}}{\hat{x}}, B_2=\rv{-\frac{\psi}{2}}{\hat{x}}$. The first $B_2$ gate can commute with the first H gate, and transform to an $\rz$ gate operating on the state $\ket{0}$. This gate therefore only applies a global phase and can be removed. Since $\MCO{\rv{\lambda}{\hat{v}}}{C}{t} \MCO{\rv{-2\psi}{\hat{z}}}{C}{a_{\ket{0}}}$ is an MCMTSU2 gate with $m=2$ targets, it can be implemented up to a global phase using 
\lem{mcmtsu2_struct} and \lem{mcmtz_struct_log} { from \apx{MCMTSU2}} as follows.
\[
\resizebox{0.85\linewidth}{!}{$
\begin{aligned}
&\Qcircuit @C=1.0em @R=0.55em @!R { 
	 	\nghost{{C_1} :  } & \lstick{{C_1} :  } & \qw {/^{n_1}} & \ctrl{1} & \qw \\
	 	\nghost{{C}_{2} :  } & \lstick{{C}_{2} :  } & \qw {/^{n_2}} & \ctrl{1} & \qw \\
	 	\nghost{{t} :  } & \lstick{{t} :  } & \qw & \gate{\mathrm{U}}  & \qw \\
   \nghost{{a}_{\ket{0}} :  } & \lstick{{a}_{\ket{0}} :  } & \qw & \qw & \qw \\
 }
 \hspace{5mm}\raisebox{-10mm}{=}\hspace{0mm}
 \Qcircuit @C=1.0em @R=0.2em @!R { 
	 	\nghost{{C_1} :  } & \lstick{{C_1} :  } & \qw {/^{n_1}} & \ctrl{1} & \qw \\
	 \nghost{{C_2} :  } & \lstick{{C_2} :  } & \qw {/^{n_2}} & \ctrl{1} & \qw \\	\nghost{{t} :  } & \lstick{{t} :  } & \qw & \rvgate{\lambda}{\hat{v}}\qwx[1] & \qw \\
	 	\nghost{{a}_{\ket{0}} :  } & \lstick{{a}_{\ket{0}} :  } & \qw  & \rzgate{-2\psi} & \qw \\
 }
\hspace{5mm}\raisebox{-10mm}{=}\hspace{0mm}\\
&\Qcircuit @C=0.5em @R=0.4em @!R { 
	 	\nghost{{C}_{1} :  } & \lstick{{C}_{1} :  } & \qw {/^{n_1}}  & \qw  & \ctrl{2} & \qw & \multigate{1}{\mathrm{\Delta_1}}  & \qw & \qw  & \qw & \multigate{1}{\mathrm{\Delta_2}} & \qw  & \ctrl{2} & \qw & \multigate{1}{\mathrm{\Delta_1^\dagger}} & \qw & \qw & \qw & \multigate{1}{\mathrm{\Delta_2^\dagger}} & \qw \\
	 	\nghost{{C}_{2} :  } & \lstick{{C}_{2} :  } & \qw {/^{n_2}}  & \qw & \qw & \qw & \ghost{\mathrm{\Delta_1}}   & \qw  & \ctrl{1} & \qw& \ghost{\mathrm{\Delta_1}} & \qw & \qw  & \qw & \ghost{\mathrm{\Delta_1}} & \qw  & \ctrl{1} & \qw & \ghost{\mathrm{\Delta_1}} & \qw \\
	 	\nghost{{t} :  } & \lstick{{t} :  }  & \gate{\mathrm{A_4}} & \targ & \ctrl{0} & \targ & \gate{\mathrm{A_2}} & \targ & \ctrl{0} & \targ & \gate{\mathrm{A_3}} & \targ & \ctrl{0} & \targ & \gate{\mathrm{A_2}} & \targ & \ctrl{0} & \targ & \gate{\mathrm{A_1}} & \qw \\
\nghost{{a}_{\ket{0}} :  } & \lstick{{a}_{\ket{0}} :  }  & \gate{\mathrm{H}} & \ctrl{-1} & \qw & \ctrl{-1} & \gate{\mathrm{B_1}} & \ctrl{-1} & \qw & \ctrl{-1} & \gate{\mathrm{B_2}} & \ctrl{-1} & \qw & \ctrl{-1} & \gate{\mathrm{B_1}} & \ctrl{-1} & \qw & \ctrl{-1} & \gate{\mathrm{H}} & \qw \\   
 }
 \end{aligned}$
 }\qcref{mcU_decomposed}
\]
Moreover, any MCMTU2 gate $\prod_{j=1}^{m}\MCO{U_j}{C}{t_j}$ with $m\geq 1$ targets, such that $U_j=e^{i\psi_j}\rv{\lambda_j}
{\hat{v}_j}\in U(2)$, can be implemented as the MCMTSU2 gate
$\MCO{\rv{-2\psi}
{\hat{z}}}{C}{a_{\ket{0}}} \prod_{j=1}^{m}\MCO{\rv{\lambda_j}
{\hat{v}_j}}{C}{t_j}$ such that $\rv{-2\psi}
{\hat{z}},\rv{\lambda_j}
{\hat{v}_j}\in SU(2)$ and $\psi = \sum_{j=1}^{m} \psi_j$. 
\thm{MCU2_costs} follows.

\begin{theorem}\label{thm:MCU2_costs}
     Any multi-controlled multi-target $U(2)$ gate with $n\geq 6$ controls, and $m\geq 1$ targets can be implemented with one ancilla qubit in state $\ket{0}$ in the same costs and depths as an MCMTSU2 gate implemented using \thm{MCMTSU2_costs} with $m+1$ targets and five $\{\rz/\rx\}$ gates replaced with two H gates.
\end{theorem}

Our results for the discussed multi-control single-target gates are summarized in \tab{costs_table_ATA}. We compare our { results to sate of the art methods which achieved linear Clifford+T gate count and a constant number of arbitrary rotations}. We list the leading terms for the cost and depth of the CNOT and T gates.
As can be seen, we provide improvements in each category. Our cost reductions for the CNOT and T gates are $25\%-62.5\%$, and up to $50\%$, respectively, and the depth reductions are between $50\%-75\%$.

\begin{table}[H]
    \centering
    \resizebox{1.0\linewidth}{!}{
    \begin{tabular}{|p{1.2cm}|p{1.5cm}|p{1.2cm}|p{1cm}|p{1cm}|p{1cm}|p{1cm}|p{1cm}|p{1cm}|p{0.5cm}|p{0.5cm}|p{0.5cm}|p{0.5cm}|}
 \hline
 \multicolumn{3}{|c|}{Gate} & \multicolumn{2}{c|}{CNOT} & \multicolumn{2}{c|}{T} \\
  \hline
   Type & Ancilla & Source & Cost & Depth & Cost & Depth \\
 \hline
   MCSU2     &  None    & \cellcolor{lightgray!60} \cite{barenco_elementary_1995}       & \cellcolor{lightgray!60} $48n$    & \cellcolor{lightgray!60} $48n$  & \cellcolor{lightgray!60} $48n$    & \cellcolor{lightgray!60} $48n$ \\
     &     & \cellcolor{lightgray!60} \cite{maslov_advantages_2016}       & \cellcolor{lightgray!60} $32n$    & \cellcolor{lightgray!60} $32n$  & \cellcolor{lightgray!60} $32n$    & \cellcolor{lightgray!60} $32n$ \\
  &     & \cellcolor{lightgray!60} \cite{iten_quantum_2016}       & \cellcolor{lightgray!60} $28n$    & \cellcolor{lightgray!60} $28n$  & \cellcolor{lightgray!60} $28n$    & \cellcolor{lightgray!60} $28n$ \\
           &   & \cellcolor{lightgray!60} \cite{vale_decomposition_2023}       & \cellcolor{lightgray!60} $20n$    & \cellcolor{lightgray!60} $20n$  & \cellcolor{lightgray!60} $20n$    & \cellcolor{lightgray!60} $20n$ \\
   &  & Ours   & $12n$    &  $8n$   &   $16n$   &  $8n$ \\
 \hline
    MCX    &   One dirty     & \cellcolor{lightgray!60} \cite{barenco_elementary_1995}       & \cellcolor{lightgray!60} $24n$    & \cellcolor{lightgray!60}  $24n$  & \cellcolor{lightgray!60}   $24n$   & \cellcolor{lightgray!60}  $24n$ \\ 
            &        & \cellcolor{lightgray!60} \cite{maslov_advantages_2016,iten_quantum_2016}       & \cellcolor{lightgray!60} $16n$    & \cellcolor{lightgray!60}  $16n$  & \cellcolor{lightgray!60}   $16n$   & \cellcolor{lightgray!60}  $16n$ \\
     &   & Ours   & $12n$    &  $8n$   &   $16n$   &  $8n$ \\
 \hline
    MCU2    &  One clean     & \cellcolor{lightgray!60} \cite{barenco_elementary_1995}       & \cellcolor{lightgray!60} $48n$   & \cellcolor{lightgray!60}  $48n$  & \cellcolor{lightgray!60}  $48n$   & \cellcolor{lightgray!60} $48n$ \\
            &       & \cellcolor{lightgray!60} \cite{maslov_advantages_2016,iten_quantum_2016}       & \cellcolor{lightgray!60} $32n$   & \cellcolor{lightgray!60}  $32n$  & \cellcolor{lightgray!60}  $32n$   & \cellcolor{lightgray!60} $32n$ \\
    &  & Ours   & $12n$    &  $8n$   &   $16n$   &  $8n$ \\
       
 \hline
    \end{tabular}
    }
    \caption{A summary of our ATA multi-controlled single-target results compared to previous methods.}
    \label{tab:costs_table_ATA}
\end{table}

\tab{tcosts_table_ATA} presents our results for an alternative method that favors the depth of T gates, as discussed in \apx{T_depth}.

 We discuss cost reductions which can be applied, in case additional dirty ancilla qubits are available, in \apx{extra_dirty}. These reductions can also be applied to MCMTX gates with $m>2$ targets without ancilla, as noted in \apx{ata_MCMTX}. The results for our multi-controlled single-target gate implementations with additional ancilla are presented in \fig{fig_ancil_CNOT_cost}. 
    
\section{Linear-nearest-neighbor (LNN) connectivity}

{We focus on the same gate types that we covered for all-to-all connectivity, this time under the restrictions of LNN connectivity. One main difference between these two challenges is that in the LNN case, when a MC gate with $n$ controls is implemented over a set of $k>n$ qubits, we need to account for any possible choice of the control and target qubits. 
The approach taken in known methods is to first assume a specific choice for these qubits, and then applying SWAP gates in order to move these qubits to their place \cite{cheng_mapping_2018,chakrabarti_nearest_2007,miller_elementary_2011,li_quantum_2023,tan_multi-strategy_2018}. However, this approach results in an overhead of the CNOT gate count with a quadratic upper bound ($O(kn)$).
As we will show, an alternative approach allows for a linear ($O(n+k)$) scaling of the gate count, by only assuming the location of one or two qubits instead of all of them. As moving one qubit to its place only requires a linear ($O(k)$) number of SWAP gates, this overhead maintains the linear scaling. As before, our MCX and MCU2 structures are built upon the MCSU2 and the MCMTSU2 structures, and thus their gate counts scale similarly, only differing by the linear overhead defined by the number (one or two) and type (target or ancilla) of qubits which should be moved into place. }
\subsection{Multi-controlled SU(2)}
{We start by constructing the MCSU2 structure, first showing how to map our macro structure from all-to-all to LNN connectivity, while only assuming the location of the target qubit. Then in \sec{MCSU2_decomposition} we show how to decompose the structure using a linear count of three- and two-qubit gates and finally decomposing these into the Clifford+T gate set. We continue to generalize our gate and allow for multiple targets in \apx{MCMTSU2_lnn}, achieving a linear count as well. This specific task may demonstrate some of the challenges in implementing efficient LNN circuits, compared to their all-to-all versions.}
\subsubsection{MCSU2 macro structure}\label{sec:MCSU2_macro_lnn}
In this section, we provide an implementation of any LNN-restricted multi-controlled $SU(2)$ gate with $n$ control qubits $C=\{c_1,c_2,..,c_n\}$ and a single target qubit $t$. The gate is implemented in a circuit over $k'\geq n+1$ qubits $Q'=\{q'_1,q'_2,..,q'_{k'}\}$ such that the qubits $q'_l$ and $q'_{l+1}$ are nearest neighbors for $1\leq l<k'$ (LNN qubits). 

We define the set $Q=\{q_1,q_2,..,q_{k}\}$ of size $k $, such that $Q\in Q'$, as the {\em smallest} set of LNN qubits that satisfies $\{C,t\}\in Q$. Each control qubit in the set $C$ corresponds to an arbitrary qubit in the set $Q \setminus t$, and it is assumed that the target qubit is located at the bottom, i.e. $t=q_{k}$, which implies $q_1\in C$. We will discuss the additional cost (communication overhead) required in case $t\not = q_{k}$ in \sec{MCSU2_decomposition}.

The macro structure described in \lem{ccrv_macro} can be used directly.
In order to allow for an efficient LNN implementation, we require an additional constraint on the subsets $C_1$ and $C_2$, such that in each subset {\em only} the first two qubits may be nearest neighbors. { The importance of this constraint will be demonstrated in \sec{MCSU2_decomposition}.}
For any given sets $C,Q$, the subsets can be defined such that this constraint is met and $C_1\cup C_2 = C$. We use the following procedure to define the sets.

Starting with $C_1=\{q_1\}$, $C_2=\emptyset$, and the iteration index $l=2$. At each iteration, $q_l$ is added to $C_1$ if
\[
(q_l\in C)\land \left(l=2 \lor (q_{l-1}\not\in C_1 \land q_{l-1} \not = C_2[1])\right).
\]
If this condition is not met and $q_l\in C$, the qubit $q_l$ is added to $C_2$. Finally, $l$ increases by 1 for the next iteration, stopping when $C_1 \cup C_2$ holds all the control qubits in $C$, i.e. $n_1+n_2=n$.

We define two subsets $Q_1,Q_2\in Q$ of sizes $k_1$ and $k_2$ as the {\em smallest} sets of LNN qubits which satisfy $\{C_1,t\}\in Q_1$ and $\{C_2,t\}\in Q_2$. The relative phase gates $[\Delta_1]$ and $[\Delta_2]$ from \lem{ccrv_macro} operate on $Q_1\setminus t$ and  $Q_2\setminus t$.

For a specific example of a given set $Q=\{q_1,q_2,..,q_{19}\}$, such that $t=q_{19}$, we use a binary string to describe the control sets $C,C_{1},C_{2}$ such that the string holds "1" for index $l$ iff $q_l$ is in the control set. The following provides an arbitrarily chosen set $C$, and the corresponding subsets $C_1,C_2$ achieved using the procedure above.
\[
\begin{aligned}
&C_{\text{\, }} = [1,0,1,0,0,1,0,1,1,1,1,0,1,1,1,0,1,1,0]\\
&C_1 = [1,0,1,0,0,1,0,1,0,0,1,0,1,0,1,0,1,0,0]\\
&C_2 = [0,0,0,0,0,0,0,0,1,1,0,0,0,1,0,0,0,1,0]
\end{aligned}
\]
As can be seen, for each set $C_1$,$C_2$ it is true that there are no neighboring controls, excluding the first two. In this example, $Q_1=Q$, and $Q_2=Q[9:19]$ are of size $k_1=19$, and $k_2=11$.

The following circuit describes the structure achieved for this example.
\[
\resizebox{0.85\linewidth}{!}{
\Qcircuit @C=1.0em @R=0.em @!R { 
	 	\nghost{{q}_{1} :  } & \lstick{{q}_{1} :  } & \ctrl{2} & \qw \\
	 	\nghost{{q}_{2} :  } & \lstick{{q}_{2} :  } & \qw & \qw \\
	 	\nghost{{q}_{3} :  } & \lstick{{q}_{3} :  } & \ctrl{3} & \qw \\
	 	\nghost{{q}_{4} :  } & \lstick{{q}_{4} :  } & \qw & \qw \\
	 	\nghost{{q}_{5} :  } & \lstick{{q}_{5} :  } & \qw & \qw \\
	 	\nghost{{q}_{6} :  } & \lstick{{q}_{6} :  } & \ctrl{2} & \qw \\
	 	\nghost{{q}_{7} :  } & \lstick{{q}_{7} :  } & \qw & \qw \\
	 	\nghost{{q}_{8} :  } & \lstick{{q}_{8} :  } & \ctrl{1} & \qw \\
	 	\nghost{{q}_{9} :  } & \lstick{{q}_{9} :  } & \ctrl{1} & \qw \\
	 	\nghost{{q}_{10} :  } & \lstick{{q}_{10} :  } & \ctrl{1} & \qw \\
	 	\nghost{{q}_{11} :  } & \lstick{{q}_{11} :  } & \ctrl{2} & \qw \\
	 	\nghost{{q}_{12} :  } & \lstick{{q}_{12} :  } & \qw & \qw \\
	 	\nghost{{q}_{13} :  } & \lstick{{q}_{13} :  } & \ctrl{1} & \qw \\
	 	\nghost{{q}_{14} :  } & \lstick{{q}_{14} :  } & \ctrl{1} & \qw \\
	 	\nghost{{q}_{15} :  } & \lstick{{q}_{15} :  } & \ctrl{2} & \qw \\
	 	\nghost{{q}_{16} :  } & \lstick{{q}_{16} :  } & \qw & \qw \\
	 	\nghost{{q}_{17} :  } & \lstick{{q}_{17} :  } & \ctrl{1} & \qw \\
	 	\nghost{{q}_{18} :  } & \lstick{{q}_{18} :  } & \ctrl{1} & \qw \\
	 	\nghost{{q}_{19} :  } & \lstick{{q}_{19} :  } & \rvgate{\lambda}{\hat{v}} & \qw \\
 }
 \hspace{5mm}\raisebox{-50mm}{=}\hspace{0mm}
 \Qcircuit @C=1.0em @R=0.18em @!R { 
	 	\nghost{{q}_{1} :  } & \lstick{{q}_{1} :  } & \qw & \ctrl{2} & \multigate{17}{\mathrm{\Delta_1}} & \qw & \qw & \ctrl{2} & \multigate{17}{\mathrm{\Delta_1^\dagger}} & \qw & \qw & \qw & \qw\\
	 	\nghost{{q}_{2} :  } & \lstick{{q}_{2} :  } & \qw & \qw & \ghost{\mathrm{\Delta_1}} & \qw & \qw & \qw & \ghost{\mathrm{\Delta_1}} & \qw & \qw & \qw & \qw\\
	 	\nghost{{q}_{3} :  } & \lstick{{q}_{3} :  } & \qw & \ctrl{3} & \ghost{\mathrm{\Delta_1}} & \qw & \qw & \ctrl{3} & \ghost{\mathrm{\Delta_1}} & \qw & \qw & \qw & \qw\\
	 	\nghost{{q}_{4} :  } & \lstick{{q}_{4} :  } & \qw & \qw & \ghost{\mathrm{\Delta_1}} & \qw & \qw & \qw & \ghost{\mathrm{\Delta_1}} & \qw & \qw & \qw & \qw\\
	 	\nghost{{q}_{5} :  } & \lstick{{q}_{5} :  } & \qw & \qw & \ghost{\mathrm{\Delta_1}} & \qw & \qw & \qw & \ghost{\mathrm{\Delta_1}} & \qw & \qw & \qw & \qw\\
	 	\nghost{{q}_{6} :  } & \lstick{{q}_{6} :  } & \qw & \ctrl{2} & \ghost{\mathrm{\Delta_1}} & \qw & \qw & \ctrl{2} & \ghost{\mathrm{\Delta_1}} & \qw & \qw & \qw & \qw\\
	 	\nghost{{q}_{7} :  } & \lstick{{q}_{7} :  } & \qw & \qw & \ghost{\mathrm{\Delta_1}} & \qw & \qw & \qw & \ghost{\mathrm{\Delta_1}} & \qw & \qw & \qw & \qw\\
	 	\nghost{{q}_{8} :  } & \lstick{{q}_{8} :  } & \qw & \ctrl{3} & \ghost{\mathrm{\Delta_1}} & \qw & \qw & \ctrl{3} & \ghost{\mathrm{\Delta_1}} & \qw & \qw & \qw & \qw\\
	 	\nghost{{q}_{9} :  } & \lstick{{q}_{9} :  } & \qw & \qw & \ghost{\mathrm{\Delta_1}} & \ctrl{1} & \multigate{9}{\mathrm{\Delta_2}} & \qw & \ghost{\mathrm{\Delta_1}} & \ctrl{1} & \multigate{9}{\mathrm{\Delta_2^\dagger}} & \qw & \qw\\
	 	\nghost{{q}_{10} :  } & \lstick{{q}_{10} :  } & \qw & \qw & \ghost{\mathrm{\Delta_1}} & \ctrl{4} & \ghost{\mathrm{\Delta_1}} & \qw & \ghost{\mathrm{\Delta_1}} & \ctrl{4} & \ghost{\mathrm{\Delta_1}} & \qw & \qw\\
	 	\nghost{{q}_{11} :  } & \lstick{{q}_{11} :  } & \qw & \ctrl{2} & \ghost{\mathrm{\Delta_1}} & \qw & \ghost{\mathrm{\Delta_1}} & \ctrl{2} & \ghost{\mathrm{\Delta_1}} & \qw & \ghost{\mathrm{\Delta_1}} & \qw & \qw\\
	 	\nghost{{q}_{12} :  } & \lstick{{q}_{12} :  } & \qw & \qw & \ghost{\mathrm{\Delta_1}} & \qw & \ghost{\mathrm{\Delta_1}} & \qw & \ghost{\mathrm{\Delta_1}} & \qw & \ghost{\mathrm{\Delta_1}} & \qw & \qw\\
	 	\nghost{{q}_{13} :  } & \lstick{{q}_{13} :  } & \qw & \ctrl{2} & \ghost{\mathrm{\Delta_1}} & \qw & \ghost{\mathrm{\Delta_1}} & \ctrl{2} & \ghost{\mathrm{\Delta_1}} & \qw & \ghost{\mathrm{\Delta_1}} & \qw & \qw\\
	 	\nghost{{q}_{14} :  } & \lstick{{q}_{14} :  } & \qw & \qw & \ghost{\mathrm{\Delta_1}} & \ctrl{4} & \ghost{\mathrm{\Delta_1}} & \qw & \ghost{\mathrm{\Delta_1}} & \ctrl{4} & \ghost{\mathrm{\Delta_1}} & \qw & \qw\\
	 	\nghost{{q}_{15} :  } & \lstick{{q}_{15} :  } & \qw & \ctrl{2} & \ghost{\mathrm{\Delta_1}} & \qw & \ghost{\mathrm{\Delta_1}} & \ctrl{2} & \ghost{\mathrm{\Delta_1}} & \qw & \ghost{\mathrm{\Delta_1}} & \qw & \qw\\
	 	\nghost{{q}_{16} :  } & \lstick{{q}_{16} :  } & \qw & \qw & \ghost{\mathrm{\Delta_1}} & \qw & \ghost{\mathrm{\Delta_1}} & \qw & \ghost{\mathrm{\Delta_1}} & \qw & \ghost{\mathrm{\Delta_1}} & \qw & \qw\\
	 	\nghost{{q}_{17} :  } & \lstick{{q}_{17} :  } & \qw & \ctrl{2} & \ghost{\mathrm{\Delta_1}} & \qw & \ghost{\mathrm{\Delta_1}} & \ctrl{2} & \ghost{\mathrm{\Delta_1}} & \qw & \ghost{\mathrm{\Delta_1}} & \qw & \qw\\
	 	\nghost{{q}_{18} :  } & \lstick{{q}_{18} :  } & \qw & \qw & \ghost{\mathrm{\Delta_1}} & \ctrl{1} & \ghost{\mathrm{\Delta_1}} & \qw & \ghost{\mathrm{\Delta_1}} & \ctrl{1} & \ghost{\mathrm{\Delta_1}} & \qw & \qw\\
	 	\nghost{{q}_{19} :  } & \lstick{{q}_{19} :  } & \gate{\mathrm{A_4}} & \control\qw & \gate{\mathrm{A_2}} & \control\qw & \gate{\mathrm{A_3}} & \control\qw & \gate{\mathrm{A_2}} & \control\qw & \gate{\mathrm{A_1}} & \qw & \qw\\
 }
 }
 \qcref{mcsu2_struct_lnn}
\]
We note that in case $C_2= \emptyset$, each structure controlled by $C_2$ is implemented as a single Z gate which can be cancelled out. This allows to reduce the number of $\{\rx,\rz\}$ to five, similarly to the "ABC" structure described in \cite{barenco_elementary_1995}.

\subsubsection{MCSU2 decomposition}\label{sec:MCSU2_decomposition}

In this section, we provide a decomposition of $\sqO{\Delta_1}{\{Q_1\setminus t\}}\MCO{Z}{C_1}{t}$, $\sqO{\Delta_2}{\{Q_2\setminus t\}}\MCO{Z}{C_2}{t}$ and their inverse gates in terms of Clifford+T gates, as required for the structure described in \sec{MCSU2_macro_lnn}. As mentioned, we can assume that only the first two control bits in each of the subsets $C_1$ and $C_2$ may be neighboring, and the target qubit $t$ is located at the bottom.
First, we define the following notation. 
\[
\{\overline{ Z }\}^{l}_{C,Q} = \prod_{j=l_0}^{l}\MCO{\overline{Z}}{C^{j}}{q_j}
\text{ , }
\MCO{\overline{Z}}{C^{j}}{q_j} =
\begin{cases}
    I & q_j\in C\\
    \MCO{Z}{C^{j}}{q_j} & q_j\not\in C
\end{cases}
\]
with $Q=\{q_1,q_2,..,q_{k}\}$, $C=\{c_1,c_2,..,c_n\}\in Q\setminus q_k$ and  $l_0\leq l\leq k$.
The subset $C^{j}:=Q[1:j-1]\cap C$ with $2\leq j\leq k$ holds all control qubits above $q_j$. 
In addition, $l_0$ is defined as the index of the first non-control qubit in $Q$, i.e., the smallest index which satisfies $q_{l_0}\not \in C \land C^{l_0}\not = \emptyset$ and $2 \leq l_0 \leq k$.
{ \lem{rec_struct_lnn} provides a way to transform $\{\overline{ Z }\}^{l-1}_{C,Q}$ to $\{\overline{ Z }\}^{l}_{C,Q}$ using two MCX gates. As we will show, if none of the control qubits in $C$ are neighboring (or only the first two), each of the used MCX gates is applied on neighboring qubits and only has one or two controls, i.e. LNN CNOT or LNN Toffoli. This provides the LNN version of \lem{rec_struct}, with the Toffoli gates serving a similar purpose of "adding a control" to the $\{\overline{ Z }\}^{l}_{C,Q}$ gate, while the CNOT gates are added in order to "skip" non-control qubits (as we exemplify in \qc{C_1_LNN_decompose} and \qc{C_2_LNN_decompose}). This ability to skip non-control qubits allows us to apply our structure for any choice of the control qubits under the mentioned constraints. This in turn removes the need to assume a specific choice and then apply many qubit swaps which result in a quadratic overhead of CNOT gates, a strategy that is used in state-of-the-art implementations \cite{cheng_mapping_2018,chakrabarti_nearest_2007,miller_elementary_2011,li_quantum_2023,tan_multi-strategy_2018}.}

\begin{lemma} \label{lem:rec_struct_lnn} If $l_0<k$, and $q_l\not\in C$
with $l_0< l\leq k$, then $\{\overline{ Z }\}^{l}_{C,Q} = \MCO{X}{Q[f_l+1:l]}{q_{f_l}}\{\overline{ Z }\}^{l-1}_{C,Q} \MCO{X}{Q[f_l+1:l]}{q_{f_l}}$, with $f_l$ as the smallest index that satisfies
$
\{\overline{ Z }\}^{l-1}_{C,Q} = \{\overline{ Z }\}^{f_l}_{C,Q}
$ and $l_0 \leq f_l<l$.
\end{lemma} 
\begin{proof}
The definition of $f_l$ implies $q_{f_l}\not \in C$, and $Q[f_l+1:l]\setminus q_l \in C$. In other words, $q_{f_l}$ is the nearest non-control qubit above $q_{l}$.
Therefore, 
\[
\{\overline{ Z }\}^{l-1}_{C,Q}
=
\{\overline{ Z }\}^{f_l}_{C,Q}
=
\{\overline{ Z }\}^{f_l-1}_{C,Q}\MCO{Z}{C^{f_l}}{q_{f_l}}.
\]
    Since the qubits $Q[f_l:l]$ are not used in $\{\overline{ Z }\}^{f_l-1}_{C,Q}$, this gate commutes with $\MCO{X}{Q[f_l+1:l]}{q_{f_l}}$. Therefore:
    \[
    \begin{aligned}
    &\MCO{X}{Q[f_l+1:l]}{q_{f_l}}\{\overline{ Z }\}^{l-1}_{C,Q} \MCO{X}{Q[f_l+1:l]}{q_{f_l}}
    =\\
    &\{\overline{ Z }\}^{f_l-1}_{C,Q}
    \MCO{X}{Q[f_l+1:l]}{q_{f_l}}
    \MCO{Z}{C^{f_l}}{q_{f_l}}
    \MCO{X}{Q[f_l+1:l]}{q_{f_l}}.
    \end{aligned}
    \]
    According to \lem{basic_struct}:
    \[
    \MCO{X}{Q[f_l+1:l]}{q_{f_l}}
    \MCO{Z}{C^{f_l}}{q_{f_l}}
    \MCO{X}{Q[f_l+1:l]}{q_{f_l}}
    =
    \MCO{Z}{C^{f_l}}{q_{f_l}}
    \MCO{Z}{C^{l}}{q_{l}}
    \]
    where we used $\MCO{Z}{C^{l}}{q_{l}}= \MCO{Z}{\{C^{f_l},Q[f_l+1:l]\}}{}$, since $Q[f_l+1:l]\setminus 
    q_l \in C$ and $q_{f_l}\not \in C$ due to the definition of $f_l$.
     Finally, we can substitute and use 
    \[
    \{\overline{ Z }\}^{f_l-1}_{C,Q}
    \MCO{Z}{C^{f_l}}{q_{f_l}}
    \MCO{Z}{C^{l}}{q_{l}}
    =
    \{\overline{ Z }\}^{l-1}_{C,Q}
    \MCO{Z}{C^{l}}{q_{l}}
    =
    \{\overline{ Z }\}^{l}_{C,Q}.
    \]
\end{proof}

Since there are no two neighboring control qubits for $l_0<l \leq k$, the only possible values of $f_l$ are 
\[
f_l =
\begin{cases}
    l-1 & q_l,q_{l-1}\not\in C\\
    l-2 & q_l,q_{l-2}\not\in C\land q_{l-1}\in C
\end{cases}
\]
by the definition of $f_l$.
Therefore, the $\MCO{X}{Q[f_l+1:l]}{q_{f}}$ gates in \lem{rec_struct_lnn} are restricted to $\MCO{X}{q_l}{q_{l-1}}$ or $\MCO{X}{\{q_{l-1},q_l\}}{q_{l-2}}$, i.e., either LNN CNOT or LNN Toffoli gates, which can be implemented up to a relative phase.

Using the following notation
\[
\MCO{\overline{X}}{}{q_l} =
\begin{cases}
    I & q_l\in C\\
    \MCO{X}{q_l}{q_{l-1}} & q_l,q_{l-1}\not\in C\\
    \MCO{X_{\Delta}}{\{q_l,q_{l-1}\}}{q_{l-2}} & q_l,q_{l-2}\not\in C\land q_{l-1}\in C
\end{cases}
\]
and repeatedly applying the recursive formula presented in \lem{rec_struct_lnn}, we achieve the following lemma:

\begin{lemma}\label{lem:Vchain_fullstruct_lnn}
If $l_0<k$, $q_k\not\in C$
and $q_l \not\in C\lor q_{l+1}\not\in C$ for any $l_0<l<k$, then\\
$\{\overline{ Z }\}^{k}_{C,Q}=
\left(\prod_{l=l_0+1}^{k}\MCO{\overline{X}}{}{q_l}\right)^{\dagger}
\MCO{Z}{C^{l_0}}{q_{l_0}}
\left(\prod_{l=l_0+1}^{k}\MCO{\overline{X}}{}{q_l}\right)$
for $n=\abs{C}\geq 1$ and $k=\abs{Q}>n$.
\end{lemma}
Since the first two controls might be neighboring, we get
$|C^{l_0}|\in \{1,2\}$.
Therefore, the gate $\MCO{Z}{C^{l_0}}{q_{l_0}}$ is restricted to  $\MCO{Z}{q_{l_0-1}}{q_{l_0}}$ or $\MCO{Z}{\{q_{l_0-2},q_{l_0-1}\}}{q_{l_0}}$, i.e. either LNN CZ or LNN CCZ.

The total number of $\MCO{\overline{X}}{}{q_l}$ gates used to implement $\{\overline{ Z }\}^{k}_{C,Q}$ according to \lem{Vchain_fullstruct_lnn} is $2(k-l_0)$.
The number of LNN $[X_\Delta]$ gates equals $2|Q[l_0:k]\cap C|=2(n-|C^{l_0}|)$, i.e. twice the number of control qubits located below $q_{l_0}$. The $I$ gates are not required to be implemented and their number is the same as the number of $[X_\Delta]$ gates. Therefore, the number of LNN CNOT gates is $2(k-l_0)-4(n-|C^{l_0}|)$.

 Since $Q[1]\in C$, we get $l_{0}=|C^{l_0}|+1$ with $l_0\in\{2,3\}$, and we can define $l'_0=l_0-2$ such that $l'_0=0$ indicates the use of a CZ gate and $l'_0=1$ indicates the use of a CCZ gate. The number of Toffoli gates can be written as $2(n-l'_0-1)$, and the number of CNOT gates as $2k-4n+2l'_0$.\\
{\qc{C_1_LNN_decompose} and \qc{C_2_LNN_decompose}} describe the implementations of
\[
\begin{aligned}
&\sqO{\Delta_1}{\{Q_1\setminus t\}}\MCO{Z}{C_1}{t}= \{\overline{ Z }\}^{k_1-1}_{C_1,Q_1}\MCO{Z}{C_1}{t}=\{\overline{ Z }\}^{k_1}_{C_1,Q_1}
\\ &\text{ and   
 }\\
&\sqO{\Delta_2}{\{Q_2\setminus t\}}\MCO{Z}{C_2}{t}= \MCO{S}{C_2[1:2]}{}\{\overline{ Z }\}^{k_2-1}_{C_2,Q_2}\MCO{Z}{C_2}{t}=\\
&\MCO{S}{C_2[1:2]}{}\{\overline{ Z }\}^{k_2}_{C_2,Q_2}
\end{aligned}
\]

according to \lem{Vchain_fullstruct_lnn}, for $C_1,C_2,Q_1,Q_2$ as the example defined in \sec{MCSU2_macro_lnn}.
\[
\resizebox{0.85\linewidth}{!}{
\Qcircuit @C=1.0em @R=0.2em @!R { 
	 	\nghost{{q}_{1} :  } & \lstick{{q}_{1} :  } & \ctrl{2} & \multigate{17}{\mathrm{\Delta_1}} & \qw\\
	 	\nghost{{q}_{2} :  } & \lstick{{q}_{2} :  } & \qw & \ghost{\mathrm{\Delta_1}} & \qw\\
	 	\nghost{{q}_{3} :  } & \lstick{{q}_{3} :  } & \ctrl{3} & \ghost{\mathrm{\Delta_1}} & \qw\\
	 	\nghost{{q}_{4} :  } & \lstick{{q}_{4} :  } & \qw & \ghost{\mathrm{\Delta_1}} & \qw\\
	 	\nghost{{q}_{5} :  } & \lstick{{q}_{5} :  } & \qw & \ghost{\mathrm{\Delta_1}} & \qw\\
	 	\nghost{{q}_{6} :  } & \lstick{{q}_{6} :  } & \ctrl{2} & \ghost{\mathrm{\Delta_1}} & \qw\\
	 	\nghost{{q}_{7} :  } & \lstick{{q}_{7} :  } & \qw & \ghost{\mathrm{\Delta_1}} & \qw\\
	 	\nghost{{q}_{8} :  } & \lstick{{q}_{8} :  } & \ctrl{3} & \ghost{\mathrm{\Delta_1}} & \qw\\
	 	\nghost{{q}_{9} :  } & \lstick{{q}_{9} :  } & \qw & \ghost{\mathrm{\Delta_1}} & \qw\\
	 	\nghost{{q}_{10} :  } & \lstick{{q}_{10} :  } & \qw & \ghost{\mathrm{\Delta_1}} & \qw\\
	 	\nghost{{q}_{11} :  } & \lstick{{q}_{11} :  } & \ctrl{2} & \ghost{\mathrm{\Delta_1}} & \qw\\
	 	\nghost{{q}_{12} :  } & \lstick{{q}_{12} :  } & \qw & \ghost{\mathrm{\Delta_1}} & \qw\\
	 	\nghost{{q}_{13} :  } & \lstick{{q}_{13} :  } & \ctrl{2} & \ghost{\mathrm{\Delta_1}} & \qw\\
	 	\nghost{{q}_{14} :  } & \lstick{{q}_{14} :  } & \qw & \ghost{\mathrm{\Delta_1}} & \qw\\
	 	\nghost{{q}_{15} :  } & \lstick{{q}_{15} :  } & \ctrl{2} & \ghost{\mathrm{\Delta_1}} & \qw\\
	 	\nghost{{q}_{16} :  } & \lstick{{q}_{16} :  } & \qw & \ghost{\mathrm{\Delta_1}} & \qw\\
	 	\nghost{{q}_{17} :  } & \lstick{{q}_{17} :  } & \ctrl{2} & \ghost{\mathrm{\Delta_1}} & \qw\\
	 	\nghost{{q}_{18} :  } & \lstick{{q}_{18} :  } & \qw & \ghost{\mathrm{\Delta_1}} & \qw\\
	 	\nghost{{q}_{19} :  } & \lstick{{q}_{19} :  } & \control\qw & \qw & \qw\\
 }
  \hspace{5mm}\raisebox{-35mm}{=}\hspace{0mm}
\Qcircuit @C=0.5em @R=0.2em @!R { 
	 	\nghost{{q}_{1} :  } & \lstick{{q}_{1} :  } & \ctrl{2} & \ctrl{2} & \ctrl{2} & \ctrl{2} & \ctrl{2} & \ctrl{2} & \ctrl{2} & \ctrl{2} & \ctrl{2} & \ctrl{2} & \ctrl{1} & \qw & \qw\\
	 	\nghost{{q}_{2} :  } & \lstick{{q}_{2} :  } & \qw & \qw & \qw & \qw & \qw & \qw & \qw & \qw & \qw & \qw & \control\qw & \qw & \qw\\
	 	\nghost{{q}_{3} :  } & \lstick{{q}_{3} :  } & \ctrl{3} & \ctrl{3} & \ctrl{3} & \ctrl{3} & \ctrl{3} & \ctrl{3} & \ctrl{3} & \ctrl{3} & \ctrl{2} & \ctrl{1} & \qw & \qw & \qw\\
	 	\nghost{{q}_{4} :  } & \lstick{{q}_{4} :  } & \qw & \qw & \qw & \qw & \qw & \qw & \qw & \qw & \qw & \control\qw & \qw & \qw & \qw\\
	 	\nghost{{q}_{5} :  } & \lstick{{q}_{5} :  } & \qw & \qw & \qw & \qw & \qw & \qw & \qw & \qw & \control\qw & \qw & \qw & \qw & \qw\\
	 	\nghost{{q}_{6} :  } & \lstick{{q}_{6} :  } & \ctrl{2} & \ctrl{2} & \ctrl{2} & \ctrl{2} & \ctrl{2} & \ctrl{2} & \ctrl{2} & \ctrl{1} & \qw & \qw & \qw & \qw & \qw\\
	 	\nghost{{q}_{7} :  } & \lstick{{q}_{7} :  } & \qw & \qw & \qw & \qw & \qw & \qw & \qw & \control\qw & \qw & \qw & \qw & \qw & \qw\\
	 	\nghost{{q}_{8} :  } & \lstick{{q}_{8} :  } & \ctrl{3} & \ctrl{3} & \ctrl{3} & \ctrl{3} & \ctrl{3} & \ctrl{2} & \ctrl{1} & \qw & \qw & \qw & \qw & \qw & \qw\\
	 	\nghost{{q}_{9} :  } & \lstick{{q}_{9} :  } & \qw & \qw & \qw & \qw & \qw & \qw & \control\qw & \qw & \qw & \qw & \qw & \qw & \qw\\
	 	\nghost{{q}_{10} :  } & \lstick{{q}_{10} :  } & \qw & \qw & \qw & \qw & \qw & \control\qw & \qw & \qw & \qw & \qw & \qw & \qw & \qw\\
	 	\nghost{{q}_{11} :  } & \lstick{{q}_{11} :  } & \ctrl{2} & \ctrl{2} & \ctrl{2} & \ctrl{2} & \ctrl{1} & \qw & \qw & \qw & \qw & \qw & \qw & \qw & \qw\\
	 	\nghost{{q}_{12} :  } & \lstick{{q}_{12} :  } & \qw & \qw & \qw & \qw & \control\qw & \qw & \qw & \qw & \qw & \qw & \qw & \qw & \qw\\
	 	\nghost{{q}_{13} :  } & \lstick{{q}_{13} :  } & \ctrl{2} & \ctrl{2} & \ctrl{2} & \ctrl{1} & \qw & \qw & \qw & \qw & \qw & \qw & \qw & \qw & \qw\\
	 	\nghost{{q}_{14} :  } & \lstick{{q}_{14} :  } & \qw & \qw & \qw & \control\qw & \qw & \qw & \qw & \qw & \qw & \qw & \qw & \qw & \qw\\
	 	\nghost{{q}_{15} :  } & \lstick{{q}_{15} :  } & \ctrl{2} & \ctrl{2} & \ctrl{1} & \qw & \qw & \qw & \qw & \qw & \qw & \qw & \qw & \qw & \qw\\
	 	\nghost{{q}_{16} :  } & \lstick{{q}_{16} :  } & \qw & \qw & \control\qw & \qw & \qw & \qw & \qw & \qw & \qw & \qw & \qw & \qw & \qw\\
	 	\nghost{{q}_{17} :  } & \lstick{{q}_{17} :  } & \ctrl{2} & \ctrl{1} & \qw & \qw & \qw & \qw & \qw & \qw & \qw & \qw & \qw & \qw & \qw\\
	 	\nghost{{q}_{18} :  } & \lstick{{q}_{18} :  } & \qw & \control\qw & \qw & \qw & \qw & \qw & \qw & \qw & \qw & \qw & \qw & \qw & \qw\\
	 	\nghost{{q}_{19} :  } & \lstick{{q}_{19} :  } & \control\qw & \qw & \qw & \qw & \qw & \qw & \qw & \qw & \qw & \qw & \qw & \qw & \qw\\
 }
  \hspace{5mm}\raisebox{-35mm}{=}\hspace{0mm}
 \Qcircuit @C=0.2em @R=0.2em @!R { 
	 	\nghost{{q}_{1} :  } & \lstick{{q}_{1} :  } & \qw & \qw & \qw & \qw & \qw & \qw & \qw & \qw & \qw & \qw & \ctrl{1} & \qw & \qw & \qw & \qw & \qw & \qw & \qw & \qw & \qw & \qw & \qw & \qw\\
	 	\nghost{{q}_{2} :  } & \lstick{{q}_{2} :  } & \qw & \qw & \qw & \qw & \qw & \qw & \qw & \qw & \qw & \targ & \control\qw & \targ & \qw & \qw & \qw & \qw & \qw & \qw & \qw & \qw & \qw & \qw & \qw\\
	 	\nghost{{q}_{3} :  } & \lstick{{q}_{3} :  } & \qw & \qw & \qw & \qw & \qw & \qw & \qw & \qw & \qw & \ctrlt{-1} & \qw & \ctrlt{-1} & \qw & \qw & \qw & \qw & \qw & \qw & \qw & \qw & \qw & \qw & \qw\\
	 	\nghost{{q}_{4} :  } & \lstick{{q}_{4} :  } & \qw & \qw & \qw & \qw & \qw & \qw & \qw & \qw & \targ & \ctrl{-1} & \qw & \ctrl{-1} & \targ & \qw & \qw & \qw & \qw & \qw & \qw & \qw & \qw & \qw & \qw\\
	 	\nghost{{q}_{5} :  } & \lstick{{q}_{5} :  } & \qw & \qw & \qw & \qw & \qw & \qw & \qw & \targ & \ctrl{-1} & \qw & \qw & \qw & \ctrl{-1} & \targ & \qw & \qw & \qw & \qw & \qw & \qw & \qw & \qw & \qw\\
	 	\nghost{{q}_{6} :  } & \lstick{{q}_{6} :  } & \qw & \qw & \qw & \qw & \qw & \qw & \qw & \ctrlt{-1} & \qw & \qw & \qw & \qw & \qw & \ctrlt{-1} & \qw & \qw & \qw & \qw & \qw & \qw & \qw & \qw & \qw\\
	 	\nghost{{q}_{7} :  } & \lstick{{q}_{7} :  } & \qw & \qw & \qw & \qw & \qw & \qw & \targ & \ctrl{-1} & \qw & \qw & \qw & \qw & \qw & \ctrl{-1} & \targ & \qw & \qw & \qw & \qw & \qw & \qw & \qw & \qw\\
	 	\nghost{{q}_{8} :  } & \lstick{{q}_{8} :  } & \qw & \qw & \qw & \qw & \qw & \qw & \ctrlt{-1} & \qw & \qw & \qw & \qw & \qw & \qw & \qw & \ctrlt{-1} & \qw & \qw & \qw & \qw & \qw & \qw & \qw & \qw\\
	 	\nghost{{q}_{9} :  } & \lstick{{q}_{9} :  } & \qw & \qw & \qw & \qw & \qw & \targ & \ctrl{-1} & \qw & \qw & \qw & \qw & \qw & \qw & \qw & \ctrl{-1} & \targ & \qw & \qw & \qw & \qw & \qw & \qw & \qw\\
	 	\nghost{{q}_{10} :  } & \lstick{{q}_{10} :  } & \qw & \qw & \qw & \qw & \targ & \ctrl{-1} & \qw & \qw & \qw & \qw & \qw & \qw & \qw & \qw & \qw & \ctrl{-1} & \targ & \qw & \qw & \qw & \qw & \qw & \qw\\
	 	\nghost{{q}_{11} :  } & \lstick{{q}_{11} :  } & \qw & \qw & \qw & \qw & \ctrlt{-1} & \qw & \qw & \qw & \qw & \qw & \qw & \qw & \qw & \qw & \qw & \qw & \ctrlt{-1} & \qw & \qw & \qw & \qw & \qw & \qw\\
	 	\nghost{{q}_{12} :  } & \lstick{{q}_{12} :  } & \qw & \qw & \qw & \targ & \ctrl{-1} & \qw & \qw & \qw & \qw & \qw & \qw & \qw & \qw & \qw & \qw & \qw & \ctrl{-1} & \targ & \qw & \qw & \qw & \qw & \qw\\
	 	\nghost{{q}_{13} :  } & \lstick{{q}_{13} :  } & \qw & \qw & \qw & \ctrlt{-1} & \qw & \qw & \qw & \qw & \qw & \qw & \qw & \qw & \qw & \qw & \qw & \qw & \qw & \ctrlt{-1} & \qw & \qw & \qw & \qw & \qw\\
	 	\nghost{{q}_{14} :  } & \lstick{{q}_{14} :  } & \qw & \qw & \targ & \ctrl{-1} & \qw & \qw & \qw & \qw & \qw & \qw & \qw & \qw & \qw & \qw & \qw & \qw & \qw & \ctrl{-1} & \targ & \qw & \qw & \qw & \qw\\
	 	\nghost{{q}_{15} :  } & \lstick{{q}_{15} :  } & \qw & \qw & \ctrlt{-1} & \qw & \qw & \qw & \qw & \qw & \qw & \qw & \qw & \qw & \qw & \qw & \qw & \qw & \qw & \qw & \ctrlt{-1} & \qw & \qw & \qw & \qw\\
	 	\nghost{{q}_{16} :  } & \lstick{{q}_{16} :  } & \qw & \targ & \ctrl{-1} & \qw & \qw & \qw & \qw & \qw & \qw & \qw & \qw & \qw & \qw & \qw & \qw & \qw & \qw & \qw & \ctrl{-1} & \targ & \qw & \qw & \qw\\
	 	\nghost{{q}_{17} :  } & \lstick{{q}_{17} :  } & \qw & \ctrlt{-1} & \qw & \qw & \qw & \qw & \qw & \qw & \qw & \qw & \qw & \qw & \qw & \qw & \qw & \qw & \qw & \qw & \qw & \ctrlt{-1} & \qw & \qw & \qw\\
	 	\nghost{{q}_{18} :  } & \lstick{{q}_{18} :  } & \targ & \ctrl{-1} & \qw & \qw & \qw & \qw & \qw & \qw & \qw & \qw & \qw & \qw & \qw & \qw & \qw & \qw & \qw & \qw & \qw & \ctrl{-1} & \targ & \qw & \qw\\
	 	\nghost{{q}_{19} :  } & \lstick{{q}_{19} :  } & \ctrl{-1} & \qw & \qw & \qw & \qw & \qw & \qw & \qw & \qw & \qw & \qw & \qw & \qw & \qw & \qw & \qw & \qw & \qw & \qw & \qw & \ctrl{-1} & \qw & \qw\\
 }
}
\qcref{C_1_LNN_decompose}
\]

\[
\resizebox{0.85\linewidth}{!}{
\Qcircuit @C=1.0em @R=0.38em @!R { 
	 	\nghost{{q}_{9} :  } & \lstick{{q}_{9} :  } & \ctrl{1} & \multigate{9}{\mathrm{\Delta_2}} & \qw\\
	 	\nghost{{q}_{10} :  } & \lstick{{q}_{10} :  } & \ctrl{4} & \ghost{\mathrm{\Delta_1}} & \qw\\
	 	\nghost{{q}_{11} :  } & \lstick{{q}_{11} :  } & \qw & \ghost{\mathrm{\Delta_1}} & \qw\\
	 	\nghost{{q}_{12} :  } & \lstick{{q}_{12} :  } & \qw & \ghost{\mathrm{\Delta_1}} & \qw\\
	 	\nghost{{q}_{13} :  } & \lstick{{q}_{13} :  } & \qw & \ghost{\mathrm{\Delta_1}} & \qw\\
	 	\nghost{{q}_{14} :  } & \lstick{{q}_{14} :  } & \ctrl{4} & \ghost{\mathrm{\Delta_1}} & \qw\\
	 	\nghost{{q}_{15} :  } & \lstick{{q}_{15} :  } & \qw & \ghost{\mathrm{\Delta_1}}& \qw\\
	 	\nghost{{q}_{16} :  } & \lstick{{q}_{16} :  } & \qw & \ghost{\mathrm{\Delta_1}} & \qw\\
	 	\nghost{{q}_{17} :  } & \lstick{{q}_{17} :  } & \qw & \ghost{\mathrm{\Delta_1}} & \qw\\
	 	\nghost{{q}_{18} :  } & \lstick{{q}_{18} :  } & \ctrl{1} & \ghost{\mathrm{\Delta_1}} & \qw\\
	 	\nghost{{q}_{19} :  } & \lstick{{q}_{19} :  } & \control\qw & \qw & \qw\\
 }
   \hspace{5mm}\raisebox{-22mm}{=}\hspace{0mm}
\Qcircuit @C=0.5em @R=0.em @!R { 
	 	\nghost{{q}_{9} :  } & \lstick{{q}_{9} :  } & \ctrl{1} & \ctrl{1} & \ctrl{1} & \ctrl{1} & \ctrl{1} & \ctrl{1} & \ctrl{1} & \ctrl{1} & \qw & \qw\\
	 	\nghost{{q}_{10} :  } & \lstick{{q}_{10} :  } & \ctrl{4} & \ctrl{4} & \ctrl{4} & \ctrl{4} & \ctrl{3} & \ctrl{2} & \ctrl{1} & \gate{\mathrm{S}} & \qw & \qw\\
	 	\nghost{{q}_{11} :  } & \lstick{{q}_{11} :  } & \qw & \qw & \qw & \qw & \qw & \qw & \control\qw & \qw & \qw & \qw\\
	 	\nghost{{q}_{12} :  } & \lstick{{q}_{12} :  } & \qw & \qw & \qw & \qw & \qw & \control\qw & \qw & \qw & \qw & \qw\\
	 	\nghost{{q}_{13} :  } & \lstick{{q}_{13} :  } & \qw & \qw & \qw & \qw & \control\qw & \qw & \qw & \qw & \qw & \qw\\
	 	\nghost{{q}_{14} :  } & \lstick{{q}_{14} :  } & \ctrl{4} & \ctrl{3} & \ctrl{2} & \ctrl{1} & \qw & \qw & \qw & \qw & \qw & \qw\\
	 	\nghost{{q}_{15} :  } & \lstick{{q}_{15} :  } & \qw & \qw & \qw & \control\qw & \qw & \qw & \qw & \qw & \qw & \qw\\
	 	\nghost{{q}_{16} :  } & \lstick{{q}_{16} :  } & \qw & \qw & \control\qw & \qw & \qw & \qw & \qw & \qw & \qw & \qw\\
	 	\nghost{{q}_{17} :  } & \lstick{{q}_{17} :  } & \qw & \control\qw & \qw & \qw & \qw & \qw & \qw & \qw & \qw & \qw\\
	 	\nghost{{q}_{18} :  } & \lstick{{q}_{18} :  } & \ctrl{1} & \qw & \qw & \qw & \qw & \qw & \qw & \qw & \qw & \qw\\
	 	\nghost{{q}_{19} :  } & \lstick{{q}_{19} :  } & \control\qw & \qw & \qw & \qw & \qw & \qw & \qw & \qw & \qw & \qw\\
 }
 \hspace{5mm}\raisebox{-22mm}{=}\hspace{0mm}
 \Qcircuit @C=0.2em @R=0em @!R { 
	 	\nghost{{q}_{9} :  } & \lstick{{q}_{9} :  } & \qw & \qw & \qw & \qw & \qw & \qw & \ctrl{1} & \qw & \qw & \qw & \qw & \qw & \qw & \qw & \qw\\
	 	\nghost{{q}_{10} :  } & \lstick{{q}_{10} :  } & \qw & \qw & \qw & \qw & \qw & \qw & \ctrl{1} & \qw & \qw & \qw & \qw & \qw & \qw & \qw & \qw\\
	 	\nghost{{q}_{11} :  } & \lstick{{q}_{11} :  } & \qw & \qw & \qw & \qw & \qw & \targ & \gate{\mathrm{iZ}} & \targ & \qw & \qw & \qw & \qw & \qw & \qw & \qw\\
	 	\nghost{{q}_{12} :  } & \lstick{{q}_{12} :  } & \qw & \qw & \qw & \qw & \targ & \ctrl{-1} & \qw & \ctrl{-1} & \targ & \qw & \qw & \qw & \qw & \qw & \qw\\
	 	\nghost{{q}_{13} :  } & \lstick{{q}_{13} :  } & \qw & \qw & \qw & \targ & \ctrl{-1} & \qw & \qw & \qw & \ctrl{-1} & \targ & \qw & \qw & \qw & \qw & \qw\\
	 	\nghost{{q}_{14} :  } & \lstick{{q}_{14} :  } & \qw & \qw & \qw & \ctrlt{-1} & \qw & \qw & \qw & \qw & \qw & \ctrlt{-1} & \qw & \qw & \qw & \qw & \qw\\
	 	\nghost{{q}_{15} :  } & \lstick{{q}_{15} :  } & \qw & \qw & \targ & \ctrl{-1} & \qw & \qw & \qw & \qw & \qw & \ctrl{-1} & \targ & \qw & \qw & \qw & \qw\\
	 	\nghost{{q}_{16} :  } & \lstick{{q}_{16} :  } & \qw & \targ & \ctrl{-1} & \qw & \qw & \qw & \qw & \qw & \qw & \qw & \ctrl{-1} & \targ & \qw & \qw & \qw\\
	 	\nghost{{q}_{17} :  } & \lstick{{q}_{17} :  } & \targ & \ctrl{-1} & \qw & \qw & \qw & \qw & \qw & \qw & \qw & \qw & \qw & \ctrl{-1} & \targ & \qw & \qw\\
	 	\nghost{{q}_{18} :  } & \lstick{{q}_{18} :  } & \ctrlt{-1} & \qw & \qw & \qw & \qw & \qw & \qw & \qw & \qw & \qw & \qw & \qw & \ctrlt{-1} & \qw & \qw\\
	 	\nghost{{q}_{19} :  } & \lstick{{q}_{19} :  } & \ctrl{-1} & \qw & \qw & \qw & \qw & \qw & \qw & \qw & \qw & \qw & \qw & \qw & \ctrl{-1} & \qw & \qw\\
 }
}
\qcref{C_2_LNN_decompose}
\]
Where we add the controlled $S$ (CS) gate in order to transform the CCZ gate to a $[iZ]$, similarly to the ATA case in \sec{decompose_1}.

We now wish to decompose the structure in terms of basic Clifford+T gates, under the condition that CNOT gates can only apply on nearest neighboring qubits. All of the CNOT and CZ gates applied in the structure are already applied on nearest neighbors. All $[X_\Delta]$ gates are of the same form, with controls $q_{l},q_{l-1}$ and target $q_{l-2}$. Each such gate can be implemented as the following LNN restricted circuit:
\[
\resizebox{0.85\linewidth}{!}{
\Qcircuit @C=1.0em @R=0.8em @!R { 
	 	\nghost{{q}_{l-2} :  } & \lstick{{q}_{l-2} :  } & \targ & \qw \\
	 	\nghost{{q}_{l-1} :  } & \lstick{{q}_{l-1} :  } & \ctrlt{-1} & \qw \\
	 	\nghost{{q}_{l} :  } & \lstick{{q}_{l} :  } & \ctrl{-1} & \qw \\
 }
\hspace{5mm}\raisebox{-5mm}{=}\hspace{0mm}
\Qcircuit @C=0.5em @R=0.2em @!R { 
	 	\nghost{{q}_{l-2} :  } & \lstick{{q}_{l-2} :  } & \gate{\mathrm{H}} & \ctrl{1} & \gate{\mathrm{T^\dagger}} & \targ \barrier[0em]{2} & \qw & \qw & \targ & \gate{\mathrm{T}} & \ctrl{1} & \gate{\mathrm{H}} & \qw & \qw\\
	 	\nghost{{q}_{l-1} :  } & \lstick{{q}_{l-1} :  } & \qw & \targ & \gate{\mathrm{T}} & \ctrl{-1} & \qw & \targ & \ctrl{-1} & \gate{\mathrm{T^\dagger}} & \targ & \qw & \qw & \qw\\
	 	\nghost{{q}_{l} :  } & \lstick{{q}_{l} :  } & \qw & \qw & \qw & \qw & \qw & \ctrl{-1}  & \qw & \qw & \qw & \qw & \qw & \qw\\
 }
 }
 \qcref{rf_toffoli_LNN}
\]
As in the ATA case (\apx{deph_reduce}), the gates on the left of the barrier can commute with other gates in the structure, and allow for a reduced depth.

In case a CZ gate is used in the circuit, it is simply replaced by a CNOT and two H gates.
In case a $[iZ]$ gate is used, it can be replaced by a $[iZ]$ with an additional SWAP gate (which is cancelled out in the full MCSU2 structure) operating on the controls. This gate can be implemented as follows.
\[
\resizebox{0.85\linewidth}{!}{
\Qcircuit @C=1.0em @R=0.3em @!R { 
	 	\nghost{{q}_{l_0-2} :  } & \lstick{{q}_{l_0-2} :  } & \ctrl{1} & \qswap & \qw \\
	 	\nghost{{q}_{l_0-1} :  } & \lstick{{q}_{l_0-1} :  } & \ctrl{1} & \qswap \qwx[-1] & \qw \\
	 	\nghost{{q}_{l_0} :  } & \lstick{{q}_{l_0} :  } & \gate{\mathrm{iZ}} & \qw & \qw \\
 }
 \hspace{5mm}\raisebox{-6mm}{=}\hspace{0mm}
\Qcircuit @C=0.5em @R=0.2em @!R { 
	 	\nghost{{q}_{l_0-2} :  } & \lstick{{q}_{l_0-2} :  } & \qw & \targ & \qw & \qw & \ctrl{1} & \gate{\mathrm{T^\dagger}} & \targ & \qw & \qw  \\
	 	\nghost{{q}_{l_0-1} :  } & \lstick{{q}_{l_0-1} :  } & \targ & \ctrl{-1} & \gate{\mathrm{T}} & \targ & \targ & \gate{\mathrm{T}} & \ctrl{-1} & \targ & \qw \\
	 	\nghost{{q}_{l_0}{t} :  } & \lstick{{q}_{l_0} :  } & \ctrl{-1} & \qw & \gate{\mathrm{T^\dagger}}  & \ctrl{-1} & \qw & \qw & \qw & \ctrl{-1} & \qw \\
 }
 }
 \qcref{rz_toffoli_LNN}
\]

Using the $l'_0$ notation, the gate $\MCO{Z}{C^{l_0}}{q_{l_0}}$ can be replaced with two H and one CNOT gate for $l'_0=0$, or six CNOT cost/depth along with four T gates in depth two for $l'_0=1$. Therefore, the CNOT depth/cost is $5l'_0+1$, T cost is $4l'_0$ and depth $2l'_0$, and H cost is $2(1-l'_0)$. The H gates can be applied on the top qubit and not contribute to the depth. \lem{vchainz_cost_depth_LNN} follows.

\begin{lemma}\label{lem:vchainz_cost_depth_LNN}
A $\{\overline{ Z }\}^{k}_{C,Q}$ gate with $n\geq 3$ and $k>n$ such that only the first two control qubits may be neighboring, and $q_{k}\not\in C$ can be implemented up to a SWAP-CS gate on two neighboring controls using\\
CNOT cost $2k+6n-3l'_0-9$ and depth $2k+2n+l'_0-1$,\\
T ~~~~~ cost $8n-4l'_0-8$ ~~~~~~ and depth $2n$,\\
H ~~~~~ cost $4n-6l'_0-2$ ~~~~~~ and depth $2n-2l'_0$. 
\end{lemma}

This structure can be used in order to efficiently decompose any multi-controlled $SU(2)$ operator using the LNN restricted version of \lem{ccrv_macro}, as discussed in \sec{MCSU2_macro_lnn}. The only assumption is that the target qubit is located at the bottom of the circuit. Each of the large MCZ gates followed by a relative phase gate can be implemented as $\{\overline{ Z }\}^{k_1}_{C_1,Q_1}$ or $\{\overline{ Z }\}^{k_2}_{C_2,Q_2}$ with or without a CS and SWAP gate applied on two neighboring control qubits.

If applied on the structures corresponding to $C_1$, the SWAP gates are simply cancelled out since the first two qubits are unaffected by the structure corresponding to $C_2$. If applied on the structures corresponding to $C_2$, the SWAP gates must commute with the $[\Delta_1^\dagger]$ gate in order to cancel out. 
From the definition of $\{\overline{ Z }\}^{k_1}_{C_1,Q_1}$, it is clear that swapping any two neighboring non-control qubits has no effect, and thus the SWAP gates cancel out in this case as well.\\
Each cost and depth described in \lem{vchainz_cost_depth_LNN} can be written as $a_{k}k+a_{n}n+b_{l}l'_0+b$. Therefore, the total costs and depths of the entire MCSU2 structure can be written as
$2a_k(k_1+k_2)+2a_{n}n+2b_l(l'_{0,1}+l'_{0,2})+4b$.\\
Noting that $Q[1:2]\not\in C_2$ for $l'_{0,1}=1$, and  $Q[1:3]\not\in C_2$ for $l'_{0,1}=0$, we can write $k_2\leq k_1-3+l'_{0,1}$ and $k_1=k$. So, the costs and depths in the worst case can be written as
\[
4a_{k}k
+2a_{n}n
+2a_k(l'_{0,1}-3)
+2b_l(l'_{0,1}+l'_{0,2})
+4b.
\]
The depth of H gates is reduced by three due to paralleled gates in the full structure, and in addition, if $l'_{0,1}=0$, two H gates can be cancelled out. By choosing the values of $l'_{0,1}$ and $l'_{0,2}$ which maximize each cost and depth we achieve the following upper bounds.
\begin{lemma}\label{lem:mcsu2_lnn_cost_basic}
    Any multi-controlled $SU(2)$ gate with $n\geq 6$ controls can be implemented over $k>n$ qubits in LNN connectivity such that the target qubit is located at the bottom, using $8$ gates from $\{\rx,\rz\}$, in addition to no more than\\
CNOT cost $8k+12n-48$ and depth $8k+4n-8$,\\
T ~~~~~ cost $16n-32$ ~~~~~~ and depth $4n$,\\
H ~~~~~ cost $8n-10$ ~~~~~~~ and depth $4n-3$.
\end{lemma}
Finally, in order to account for any possible qubit permutation, we discuss the added cost in the case $t\not = q_{k}$. Applying a SWAP chain before and after the MCSU2 gate allows one to relocate the target qubit.
Noting that the LNN requirements hold if the order of $Q$ is reversed, the structure can be used if $t$ is either above or below all qubits in $C$. The maximal distance of $t$ from the nearest edge is $\floor{\frac{k-1}{2}}$ qubits, and therefore the maximal number of required SWAP gates is $2\floor{\frac{k-1}{2}}$, each equivalent to three CNOT gates.
We show in \apx{communication_overhead} that the maximal CNOT cost and depth of these SWAP gates can be reduced to $4\floor{\frac{k-1}{2}}$ and $2\floor{\frac{k-1}{2}}+4$, respectively. {\thm{lnn_su2_worst}} therefore follows from \lem{mcsu2_lnn_cost_basic}, {providing upper bounds of the Clifford+T cost/depth. These upper bounds include the added communication overhead in the worst case and therefore hold for any choice of the control and target qubits}.
\begin{theorem}\label{thm:lnn_su2_worst}
    Any multi-controlled $SU(2)$ gate with $n\geq 6$ controls can be implemented over $k>n$ qubits in LNN connectivity with any qubit permutation using $8$ gates from $\{\rx,\rz\}$, in addition to no more than\\
CNOT cost $10k+12n-50$ and depth $9k+4n-5$,\\
T ~~~~~ cost $16n-32$ ~~~~~~~ and depth $4n$,\\
H ~~~~~ cost $8n-10$ ~~~~~~~~ and depth $4n-3$.
\end{theorem}
We recall that in the ATA case, additional ancilla qubits are a useful resource that can help reduce gate counts (\apx{extra_dirty}). In the LNN case, ancilla qubits may be seen as an obstacle, as these increase the value of $k$ and therefore the number of CNOT gates as shown above.
However, in certain cases, the ancilla qubits can be used in order to reduce the gate count in LNN connectivity.
We discuss this, along with further cost reductions for specific cases in \apx{specific_cases}.

\subsection{Multi-controlled X}
Simiarly to the all-to-all case, the MCX structure is based on the MCSU2. We mainly address the difference in the CNOT overhead of assumming the location of a dirty ancilla, compared to a target. { In \apx{lnn_mcmtx} we also address the implementation of the multi-controlled multi-target X gate.}

Using the construction described in \qc{mcx_ata_macro} for the case of LNN restricted architecture, we assume that a dirty ancilla qubit $a$ is located below all qubits in $\{C,t\}$ and use \lem{mcsu2_lnn_cost_basic} in order to apply an $\MCO{\rv{2\pi}{\hat{y}}}{\{C,t\}}{a} = \MCO{Z}{C}{t}$ gate with $\{\rx,\rz\}$ gates replaced by four H gates. 
The location of the target does not matter, as it can be treated as one of the control qubits of the MCZ gate. 

The total CNOT cost and depth added in order to locate the ancilla qubit at the bottom are no larger than $4\floor{\frac{n+1}{2}}$ and $2\floor{\frac{n+1}{2}}+4$, respectively, similarly to the MCSU2 case, while now the ancilla closest to one of the edges can be labeled as $a$, and therefore the largest distance is $\floor{\frac{n+1}{2}}$.

All other costs and depths are achieved from the MCSU2 case as in \lem{mcsu2_lnn_cost_basic} with $n$ replaced by $n+1$. Finally, two Hadamard gates are applied to the MCX target qubit on both sides of the structure. This provides the following.
\begin{theorem}\label{thm:lnn_MCX_worst}
    Any multi-controlled $X$ gate with $n\geq 5$ controls can be implemented over $k\geq n+2$ qubits in any permutation in LNN connectivity using no more than\\
    CNOT cost $8k+14n-34$ and depth $8k+5n+1$,\\
T ~~~~~ cost $16n-16$ ~~~~~~ and depth $4n+4$,\\
H ~~~~~ cost $8n+4$ ~~~~~~~~ and depth $4n+3$.
\end{theorem}

{ As in the MCSU2 case, these upper bounds include the worst-case communication overhead and therefore hold for any location of the ancilla, target and the control qubits.}

\subsection{Multi-controlled U(2)} \label{sec:lnn_mcmtu2}

Similarly to the ATA case, we can use the MCMTSU2 structure { from \apx{MCMTSU2_lnn}} to implement an MCMTU2, using the clean ancilla as one of the targets, and apply a multi-controlled rotation about the $\hat{z}$ axis on it to account for the overall required phase{, using the trick shown in \lem{P_a_0}}. The MCMTU2 can therefore be implemented using \thm{lnn_mcmtsu2_worst}, with the adjustments mentioned in \thm{MCU2_costs}.

Noting that when the number of targets is small, it is beneficial to apply chains of partial swap gates in order to place all of the targets at the bottom.\\
An MCU2 gate requires one target and one clean ancilla qubit to be located at the bottom. In the worst case, the sum of the distances of these qubits from one of the edges is no larger than $k-2$. This can be realized by noting that if $d_1$ and $d_2$ represent the distances of each of the qubits from the bottom edge, then the total number of necessary swaps (pairs of SWAP gates) is $D=d_1+d_2-1$. Similarly, the distance from each of these qubits to the top edge is $k-d-1$, and the number of swaps will be $D'=2k-D-4$. As we can choose which edge to use, the required number of swaps is no larger than $min(D',D)$. The worst case is achieved when $D=D'$ and therefore $D=k-2$.
We note that if the qubit which is nearest to the chosen edge is relocated first, partial SWAP gates can be used in this case, as can be realized from \apx{communication_overhead}. The total number of CNOT gates required to place the target and the ancilla qubits at the bottom is no larger than $4k-8$, in depth $2k+8$.
The following can be achieved using \lem{mcsu2_lnn_cost_basic}.
\begin{theorem}\label{thm:lnn_u2_worst}
    Any multi-controlled $U(2)$ gate with $n\geq 6$ controls can be implemented over $k>n+1$ qubits with one ancilla qubit in state $\ket{0}$ in any permutation in LNN connectivity using $11$ gates from $\{\rx,\rz\}$, in addition to no more than\\
CNOT cost $12k+12n-56$ and depth $10k+4n$,\\
T ~~~~~ cost $16n-32$ ~~~~~~~ and depth $4n$,\\
H ~~~~~ cost $8n-8$ ~~~~~~~~~~ and depth $4n-1$.
\end{theorem}

Our upper bounds for LNN multi-controlled single target gates over any qubit permutation are summarized in \tab{LNN_costs}. { These results hold for any possible choice of controls, target, and ancillary qubits and show for the first time that the number of Clifford+T gates used to implement these MC gates over LNN connectivity scales linearly in the number of control qubits and the overall number of qubits in the 1D array.} The leading terms of the cost and depth of the CNOT and T gates are presented.

\begin{table}[H]
    \centering
    \resizebox{1.0\linewidth}{!}{
    \begin{tabular}{|p{1.2cm}|p{1.5cm}|p{1.7cm}|p{1.5cm}|p{1.5cm}|p{1cm}|p{1cm}|p{1cm}|p{1cm}|p{0.5cm}|p{0.5cm}|p{0.5cm}|p{0.5cm}|}
 \hline
 \multicolumn{3}{|c|}{Gate} & \multicolumn{2}{c|}{CNOT} & \multicolumn{2}{c|}{T} \\
  \hline
   Type & Ancilla & Source & Cost & Depth & Cost & Depth \\
 \hline
           MCSU2  & None  & \cellcolor{lightgray!60} \cite{cheng_mapping_2018,chakrabarti_nearest_2007,miller_elementary_2011,li_quantum_2023,tan_multi-strategy_2018}       & \cellcolor{lightgray!60} $O(nk)$    & \cellcolor{lightgray!60} $O(k+n)$  & \cellcolor{lightgray!60} $O(n)$    & \cellcolor{lightgray!60} $O(n)$ \\
        &     & Ours        & $10k+12n$    &  $9k+4n$  &  $16n$    &  $4n$ \\

 \hline
         MCX  & One dirty  & \cellcolor{lightgray!60} \cite{cheng_mapping_2018,chakrabarti_nearest_2007,miller_elementary_2011,li_quantum_2023,tan_multi-strategy_2018}       & \cellcolor{lightgray!60} $O(nk)$    & \cellcolor{lightgray!60} $O(k+n)$  & \cellcolor{lightgray!60} $O(n)$    & \cellcolor{lightgray!60} $O(n)$ \\
        &        & Ours  & $8k+14n$    &  $8k+5n$  &   $16n$   &  $4n$ \\

 \hline
          MCU2  & One clean  & \cellcolor{lightgray!60} \cite{cheng_mapping_2018,chakrabarti_nearest_2007,miller_elementary_2011,li_quantum_2023,tan_multi-strategy_2018}       & \cellcolor{lightgray!60} $O(nk)$    & \cellcolor{lightgray!60} $O(k+n)$  & \cellcolor{lightgray!60} $O(n)$    & \cellcolor{lightgray!60} $O(n)$ \\
        &    & Ours     & $12k+12n$   &   $10k+4n$  &   $16n$   &  $4n$ \\

 \hline
    \end{tabular}
    }
    \caption{A summary of our LNN upper bound results compared to atate of the art.}
    \label{tab:LNN_costs}
\end{table}

Previously known methods, such as the MCX decomposition described in \cite{cheng_mapping_2018}, initially assume a specific qubit ordering. The gates used in this implementation are from the NCV library. {Therefore, \tab{LNN_costs} only includes asymptotic scalings, as Clifford+T gate counts are not available.} By converting their results into fault-tolerant gates, these could be compared with \thm{lnn_su2_row} and \thm{lnn_su2_rowg} in \apx{specific_cases}. When implemented over a large circuit with arbitrary qubit permutation, the CNOT cost added due to communication overhead, which is required to reorder the qubits, scales quadratically as $O(nk)$. In contrast, our methods maintain a linear $O(k+n)$ scaling over any qubit ordering, even for the multi-controlled multi-target versions of our implementations.

\section{Conclusion}

In this paper, we presented new methods for the decomposition of various types of multi-controlled {(MC)} gates. {These gates are a crucial ingredient of many quantum algorithms \cite{arrazola_universal_2022,de_carvalho_parametrized_2024,ni_progressive_2024,iten_quantum_2016,malvetti_quantum_2021,grinko_efficient_2023,tanasescu_distribution_2022,park_circuit-based_2019,grover_quantum_1997,ali_function_2018,shafaei_cofactor_2014}.}
We have focused on the implementation of MCSU2, MCX, and MCU2, along with their multi-target versions, in both ATA and LNN connectivities. {For each type of gate, we use the minimal ancilla requirement which is known to allow or an ATA implementation with a linear Clifford+T gate count and a constant number of arbitrary rotations. Specifically,}
the MCSU2 gate is implemented without an ancilla, while MCX and MCU2 use one dirty ancilla and one clean ancilla, respectively. The underlying techniques used to decompose these gates are closely related. 

The decomposition of each type of gate requires a significantly lower number of basic Clifford+T gates, compared to previous methods which use the same ancilla requirements. In the ATA case, our cost reductions for the CNOT and T gates are $25\%-62.5\%$, and up to $50\%$, respectively \cite{barenco_elementary_1995,maslov_advantages_2016,iten_quantum_2016,vale_decomposition_2023}. { In terms of Toffoli gate count, we have achieved a $50\%$ reduction compared to the best known implementation of MCX with one dirty ancilla which was first introduced in 1995 \cite{barenco_elementary_1995}.} More strikingly, 
for our LNN implementations, the cost of each type of gate scales linearly with respect to the number of qubits in the circuit, regardless of the qubit ordering. To our knowledge, previous methods require a CNOT cost that scales quadratically for arbitrary qubit ordering \cite{cheng_mapping_2018,chakrabarti_nearest_2007,miller_elementary_2011,li_quantum_2023,tan_multi-strategy_2018} { as a result of a large overhead of SWAP gates added to reorder the qubits}. 
Moreover, we have shown in \apx{extra_dirty} and \apx{specific_cases} that our results can be further improved if additional ancilla qubits are available in some cases{, and for ATA connectivity we have managed to bridge between methods that use a constant number of dirty ancilla and those that use a linear number. Starting from a constant number of ancilla, we utilize each additional available dirty ancilla in order to reduce the cost until we reach the best known cost achievable with a linear number of dirty ancilla \cite{maslov_advantages_2016}, as demonstrated in \fig{fig_ancil_CNOT_cost}}.

As many quantum algorithms are designed using multi-controlled gates, our results directly improve the implementations of these algorithms in terms of significantly fewer basic gates, which in turn will should automatically provide a corresponding reduction in error rates.

\section*{Acknowledgements}
This work was supported by the Engineering and Physical Sciences Research Council [grant numbers EP/R513143/1, EP/T517793/1].
\section*{Competing Interests}
BZ and SB declare a relevant patent application:
United Kingdom Patent Application No. 2400486.3

\onecolumngrid

\newpage
\appendix

\section{ATA MCSU2 depth reductions}\label{apx:deph_reduce}
We provide a description of the depth reductions reported in \thm{MCSU2_costs}.
The relative phase Toffoli can be expressed as  $\MCO{X_\Delta}{\{q_1,q_2\}}{q_3} = \sqO{O_3}{\{q_1,q_2,q_3\}}\sqO{O_2}{\{q_1,q_3\}}$,
such that $\sqO{O_2}{\{q_1,q_3\}}$ and $\sqO{O_3}{\{q_1,q_2,q_3\}}$ are defined as the gates on the left of the barrier and on the right of the barrier in \qc{rf_toffoli} respectively.
The expression in \lem{Vchain_fullstruct} can be rewritten using these notations, noting that
$\prod_{j=3}^{n} \MCO{X_\Delta}{\{c_j,d_{j-1}\}}{d_{j-2}} = (\prod_{j=3}^{n}\sqO{O_3}{\{c_j,d_{j-1},d_{j-2}\}})(\prod_{j=3}^{n}\sqO{O_2}{\{c_j,d_{j-2}\}})$. This commutation is achieved because each $\sqO{O_2}{\{q_1,q_3\}}$ gate is the first to access the qubits $q_1,q_3$, as can be seen in \qc{v_chain_z_rel_phase}. We define the following notations:
\[
\Lambda_{C,D}^{O_3}= 
(\prod_{j=3}^{n}\sqO{O_3}{\{c_j,d_{j-1},d_{j-2}\}})^\dagger
\MCO{iZ}{\{c_1,c_2\}}{d_1}
(\prod_{j=3}^{n}\sqO{O_3}{\{c_j,d_{j-1},d_{j-2}\}})
\]
and
\[
\Xi_{C,D}^{O_2} = 
\prod_{j=3}^{n}\sqO{O_2}{\{c_j,d_{j-2}\}}
\]
and we can rewrite \lem{Vchain_fullstruct} as
\begin{lemma}\label{lem:Vchain_fullstruct_ivi}
$\{ Z\}^{n}_{C,D}\MCO{S}{c_1}{c_2} = 
(\Xi_{C,D}^{O_2})^\dagger
\Lambda_{C,D}^{O_3}
(\Xi_{C,D}^{O_2})$ 
for $n\geq 3$, $\abs{C}= n$, and $\abs{D}= n-1$.
\end{lemma}
Therefore, any $\{ Z\}^{n}_{C,D}\MCO{S}{c_1}{c_2}$ gate can be decomposed using one $[iZ]$ gate, $2(n-2)$ $\sqO{O_3}{\{q_1,q_2,q_3\}}$ gates, and $2(n-2)$ $\sqO{O_2}{\{q_1,q_3\}}$ gates. Since $\sqO{O_2}{\{q_1,q_3\}}$ gates can be applied simultaneously in each $\Xi$ operator, the depth is reduced to one $[iZ]$ gate, $2(n-2)$ $\sqO{O_3}{\{q_1,q_2,q_3\}}$ gates and two $\sqO{O_2}{\{q_1,q_3\}}$ gates, as presented in \lem{vchainz_cost_depth}.

We define the following control subsets:
\begin{table}[H]
     \centering
    \scalebox{0.9}{
    \begin{tabular}{ p{6cm} p{4cm}  }
 $C_1 = C[1:n_1]$ & $C_2 = C[n_1+1:n_1+n_2]$\\
 \\
  $C'_1 =
\begin{cases}
    C_1[3:n_2] & n_2=n_1\\
    \{C_1[3:n_1],C_1[2]\} & n_2\not =n_1\\
\end{cases}$  & $C'_2=C_2[3:n_1]$ \\
    \end{tabular}
}
\end{table}
Such that $n_1=\floor{\frac{n}{2}}$, $n_2=n-n_1=\ceil{\frac{n}{2}}$, $\abs{C'_1}=n_2-2$, and $\abs{C'_2}=n_1-2$.
We choose $\sqO{\Delta_1}{C}=\MCO{S}{C_1[1:2]}{}\{Z\}^{n_1-1}_{C_1,C_2'}$ and $\sqO{\Delta_2}{C}=\MCO{S}{C_2[1:2]}{}\{Z\}^{n_2-1}_{C_2,C_1'}$ as relative phase gates applied on the set $C$. For this choice we get
\[
\sqO{\Delta_1}{C}\MCO{Z}{C_1}{t}=\MCO{S}{C_1[1:2]}{}\{Z\}^{n_1}_{C_1,\{C_2',t\}}
\]
and
\[
\sqO{\Delta_2}{C}\MCO{Z}{C_2}{t}=\MCO{S}{C_2[1:2]}{}\{Z\}^{n_2}_{C_2,\{C_1',t\}}
\]
The inverted versions of these gates are implemented in the same way. These structures can be used in order to decompose \qc{mcrpi_base}.
Our choice of the sets $C_1,C_2,C_1',C_2'$ allows for a slight reduction of T depth,
considering the full structure.
Since the $\Xi$ operators commute with single qubit gates applied on the target qubit $t$, there are three copies of  $(\Xi_{C_1,C'_2}^{O_2})(\Xi_{C_2,C'_1}^{O_2})^\dagger$ or its inverse formed in layers between consecutive $\Lambda$ operators. We note that $C'_2[j-2]=C_2[j]$, and $C'_1[j-2]=C_1[j]$ for $3\leq j\leq n_1$.
Therefore,
\[
(\Xi_{C_1,C'_2}^{O_2})(\Xi_{C_2,C'_1}^{O_2})^\dagger =
\begin{cases}
    \prod_{j=3}^{n_1}(\sqO{O_2}{\{C_1[j],C_2[j]\}}\sqO{O_2^\dagger}{\{C_2[j],C_1[j]\}}) & n_2=n_1\\
    \prod_{j=3}^{n_1}(\sqO{O_2}{\{C_1[j],C_2[j]\}}\sqO{O_2^\dagger}{\{C_2[j],C_1[j]\}})\sqO{O_2^\dagger}{\{C_2[n_2],C_1[2]\}} & n_2\not =n_1\\
\end{cases}
\]
In both cases there are $n_1-2$ parallel copies of the following circuit: 
\[
\scalebox{1.0}{
\Qcircuit @C=1.0em @R=0.2em @!R { 
	 	\nghost{} & \lstick{} & \gate{\mathrm{T^\dagger}} & \targ & \gate{\mathrm{T}} & \gate{\mathrm{H}} \barrier[0em]{1} & \qw & \qw & \qw & \ctrl{1} & \qw & \qw & \qw\\
	 	\nghost{} & \lstick{} & \qw & \ctrl{-1} & \qw & \qw & \qw & \gate{\mathrm{H}} & \gate{\mathrm{T^\dagger}} & \targ & \gate{\mathrm{T}} & \qw & \qw\\
 }
  \hspace{3mm}\raisebox{-3mm}{=}\hspace{3mm}
\Qcircuit @C=1.0em @R=0.2em @!R { 
	 	 \lstick{} & \ctrl{1} & \gate{\mathrm{T^\dagger}} & \gate{\mathrm{H}} & \gate{\mathrm{T^\dagger}} & \targ & \qw & \qw\\
	 	 \lstick{} & \targ & \gate{\mathrm{T}} & \gate{\mathrm{H}} & \gate{\mathrm{T}} & \ctrl{-1} & \qw & \qw\\
 }
 }
 \qcref{two_cont_down}
\]
This allows for a reduction of H depth by 3. In addition, the T depth is reduced by 6 if $n$ is even and by 3 if $n$ is odd.
\section{ATA MCSU2 alternative}\label{apx:T_depth}

We provide an alternative implementation of the multi-controlled $SU(2)$ with reduced depth of T gates. This is achieved by using the following decompositions of $[X_\Delta]$ and $[iZ]$.

\[
\Qcircuit @C=1.0em @R=0.8em @!R { 
	 	\nghost{{q}_{1} :  } & \lstick{{q}_{1} :  } & \ctrlt{1} & \qw \\
	 	\nghost{{q_3} :  } & \lstick{{q_3} :  } & \targ & \qw \\
	 	\nghost{{q}_{2} :  } & \lstick{{q}_{2} :  } & \ctrl{-1} & \qw \\
 }
\hspace{5mm}\raisebox{-6mm}{=}\hspace{0mm}
\Qcircuit @C=1.0em @R=0.2em @!R { 
	 	\nghost{{q}_{1} :  } & \lstick{{q}_{1} :  } & \qw & \targ & \gate{\mathrm{T}} \barrier[0em]{2} & \qw & \qw & \targ & \gate{\mathrm{T^\dagger}} & \targ & \qw & \qw \\
	 	\nghost{{q_3} :  } & \lstick{{q_3} :  } & \gate{\mathrm{H}} & \ctrl{-1} & \gate{\mathrm{T^\dagger}} & \qw & \targ & \qw  & \gate{\mathrm{T}} & \ctrl{-1} & \gate{\mathrm{H}} & \qw \\
	 	\nghost{{q}_{2} :  } & \lstick{{q}_{2} :  } & \qw & \qw & \qw & \qw & \ctrl{-1} & \ctrl{-2}  & \qw & \qw & \qw & \qw \\
 }
 \qcref{toff_deco_clifft_t_depth}
\]
and

\[
\Qcircuit @C=1.0em @R=0.1em @!R { 
	 	\nghost{{q}_{1} :  } & \lstick{{q}_{1} :  } & \ctrl{1} & \qw \\
	 	\nghost{{q}_{2} :  } & \lstick{{q}_{2} :  } & \ctrl{1} & \qw \\
	 	\nghost{{q_3} :  } & \lstick{{q_3} :  } & \gate{\mathrm{iZ}} & \qw \\
 }
\hspace{5mm}\raisebox{-6mm}{=}\hspace{0mm}
\Qcircuit @C=1.0em @R=0.2em @!R { 
	 	\nghost{{q}_{1} :  } & \lstick{{q}_{1} :  } & \targ & \gate{\mathrm{T}} & \targ & \qw & \gate{\mathrm{T^\dagger}} & \targ & \qw & \qw \\
	 	\nghost{{q}_{2} :  } & \lstick{{q}_{2} :  } & \qw & \qw & \ctrl{-1} & \ctrl{1}  & \qw & \qw & \ctrl{1} & \qw \\
	 	\nghost{{q_3} :  } & \lstick{{q_3} :  } & \ctrl{-2} & \gate{\mathrm{T^\dagger}} & \qw & \targ &  \gate{\mathrm{T}} & \ctrl{-2} & \targ & \qw \\
 }
 \qcref{ccz_deco_clifft_t_depth}
\]

In this case, $\sqO{O_2}{\{q_1,q_3\}}$ and $\sqO{O_3}{\{q_1,q_2,q_3\}}$ are defined as the gates on the left of the barrier and on the right of the barrier, respectively, in \qc{toff_deco_clifft_t_depth}. Using these gates in \lem{Vchain_fullstruct_ivi} provides the following.
\begin{lemma}\label{lem:vchainz_cost_depth_T}
A $\{ Z\}^{n}_{C,D}\MCO{S}{c_1}{c_2}$ gate with $n\geq 3$, $\abs{C}= n$, and $\abs{D}= n-1$ can be implemented using\\
CNOT cost $8n-11$ and depth $6n-5$,\\
T ~~~~~ cost $8n-12$ and depth $2n$,\\
H ~~~~~ cost $4n-8$~ and depth $2n-2$.
\end{lemma}

When applied as part of \lem{ccrv_macro}, using the sets $C_1,C_2,C_1',C_2'$ as defined in \apx{deph_reduce}, some reductions in CNOT gate {\em count} can be achieved as $(\Xi_{C_1,C'_2}^{O_2})(\Xi_{C_2,C'_1}^{O_2})^\dagger$ are now comprised of $n_1-2$ parallel copies of the following circuit:
\[
\Qcircuit @C=1.0em @R=0.2em @!R { 
	 	\nghost{{}} & \lstick{{}} & \gate{\mathrm{T}} & \ctrl{1} & \gate{\mathrm{H}}\barrier[0em]{1} & \qw  & \qw & \targ & \gate{\mathrm{T}} & \qw \\
	 	\nghost{{}} & \lstick{{}} & \gate{\mathrm{T^\dagger}} & \targ & \qw  & \qw & \gate{\mathrm{H}} & \ctrl{-1} & \gate{\mathrm{T^\dagger}} & \qw \\
 }
 \hspace{3mm}\raisebox{-3mm}{=}\hspace{3mm}
 \Qcircuit @C=1.0em @R=0.2em @!R { 
	 	 \lstick{{}} & \gate{\mathrm{T}} & \gate{\mathrm{H}} & \gate{\mathrm{T}} & \qw\\
	 	 \lstick{{}} & \gate{\mathrm{T^\dagger}} & \gate{\mathrm{H}} & \gate{\mathrm{T^\dagger}} & \qw \\
 }
 \qcref{mid_layer_reduction}
\]
These cancellations can be applied in three layers and reduce the CNOT count by
$6(n_1-2)$. In addition, the H depth is reduced by $3$, and the CNOT depth is reduced by $6$ if $n$ is even, or by $3$ if $n$ is odd.\\

We now show that the CNOT cost can be further reduced by applying cancellations for the first $n_1-2$ [$X_\Delta$] gates of the full MCSU2 structure. The following gate is controlled by the state $\ket{0}$, and therefore commutes with the MCSU2 gate if its control qubit is in the set $C$.

\[
\scalebox{1}{
\Qcircuit @C=0.5em @R=0.em @!R { 
	 	\nghost{q_1} & \lstick{q_1} & \ctrlo{1} & \qw \\
	 	\nghost{q_2} & \lstick{q_2} & \rvgate{\frac{\pi}{2}}{\hat{x}} & \qw \\
 }
   \hspace{5mm}\raisebox{-4mm}{=}\hspace{0mm}
 \Qcircuit @C=0.5em @R=0.4em @!R { 
	 	\nghost{q_1} & \lstick{q_1} & \qw & \targ & \gate{\mathrm{T}} & \targ & \qw & \qw\\
	 	\nghost{q_2} & \lstick{q_2} & \gate{\mathrm{H}} & \ctrl{-1} & \gate{\mathrm{T}} & \ctrl{-1} & \gate{\mathrm{H}} & \qw \\
 }
 }
 \qcref{zero_controlled_rx}
\]
This gate and its inverse can be applied on both sides of the MCSU2 structure with $q_1=C_1[j]$ and $q_2=C_2[j]$ for each value $3 \leq j\leq n_1$. 
The two following circuits describe the gates which can replace the first $n_1-2$ [$X_\Delta$] gates, and their counterparts at the end of the MCSU2 structure.
Setting $q_3=\{C'_2,t\}[j-1]$ and $q_4=\{C'_1,t\}[j-1]$ for $3 \leq j\leq n_1$, we get the following circuits.
\[
\Qcircuit @C=1.0em @R=0em @!R { 
	 	\nghost{{q}_{1} :  } & \lstick{{q}_{1} :  } & \ctrlo{1} & \ctrlt{1} & \qw \\
	 	\nghost{{q_2} :  } & \lstick{{q_2} :  } & \gate{\mathrm{R_{\hat{x}}(\frac{\pi}{2})^\dagger}} & \targ & \qw \\
	 	\nghost{{q}_{3} :  } & \lstick{{q}_{3} :  } & \qw & \ctrl{-1} & \qw \\
 }
\hspace{5mm}\raisebox{-6mm}{=}\hspace{0mm}
\Qcircuit @C=0.5em @R=0.2em @!R { 
	 	\nghost{{q}_{1} :  } & \lstick{{q}_{1} :  } & \qw  & \qw \barrier[0em]{2} & \qw & \qw & \targ & \gate{\mathrm{T^\dagger}} & \targ & \qw & \qw \\
	 	\nghost{{q_2} :  } & \lstick{{q_2} :  } & \gate{\mathrm{H}} & \gate{\mathrm{S^\dagger}} & \qw & \targ & \ctrl{-1} & \gate{\mathrm{T}} & \ctrl{-1} & \gate{\mathrm{H}} & \qw \\
	 	\nghost{{q}_{3} :  } & \lstick{{q}_{3} :  } & \qw  & \qw & \qw & \ctrl{-1} & \qw & \qw & \qw & \qw & \qw  \\
 }
 \qcref{toff_deco_clifft_t_depth_left_cancel}
\]

\[
\Qcircuit @C=1.0em @R=0em @!R { 
\nghost{{q_1} :  } & \lstick{{q_1} :  } & \targ & \ctrlo{1} & \qw \\
	 	\nghost{{q}_{2} :  } & \lstick{{q}_{2} :  } & \ctrlt{-1} & \rvgate{\frac{\pi}{2}}{\hat{x}} & \qw \\ 	
	 	\nghost{{q}_{4} :  } & \lstick{{q}_{4} :  } & \ctrl{-1} & \qw & \qw \\
 }
\hspace{5mm}\raisebox{-6mm}{=}\hspace{0mm}
\Qcircuit @C=0.5em @R=0.2em @!R { 
\nghost{{q_1} :  } & \lstick{{q_1} :  } & \gate{\mathrm{H}} & \ctrl{1} & \gate{\mathrm{T^\dagger}} & \targ & \qw\barrier[0em]{2} & \qw & \gate{\mathrm{T}}  & \gate{\mathrm{H}} & \gate{\mathrm{T}} & \targ & \qw & \qw \\
	 	\nghost{{q}_{2} :  } & \lstick{{q}_{2} :  } & \qw & \targ & \gate{\mathrm{T}} & \qw & \targ & \qw & \gate{\mathrm{T^\dagger}} & \gate{\mathrm{H}} & \gate{\mathrm{T}} &\ctrl{-1} & \gate{\mathrm{H}} & \qw \\	 	
	 	\nghost{{q}_{4} :  } & \lstick{{q}_{4} :  } & \qw & \qw & \qw & \ctrl{-2} & \ctrl{-1} & \qw & \qw  & \qw & \qw & \qw & \qw & \qw\\
 }
 \qcref{toff_deco_clifft_t_depth_right_cancel}
\]

Each of these circuits is used to replace $n_1-2$ copies of \qc{toff_deco_clifft_t_depth}.
This allows for a reduction of $n_1-2$ CNOT gates in depth $1$, and an increase of $2(n_1-2)$ H gates in depth $1$, along with the introduction of $n_1-2$ S gates in depth $1$. 

The total reduction of CNOT gates is $7(n_1-2)$, which is equivalent to $3.5n-14$ if $n$ is even, and $3.5n-17.5$ if $n$ is odd.
The following theorem is achieved from \lem{ccrv_macro}, \lem{vchainz_cost_depth_T} and the above CNOT cancellations.
\begin{theorem}\label{thm:MCSU2_costs_tdepth}
    Any $SU(2)$ operator controlled by $n\geq 6$ qubits can be implemented without ancilla using eight gates from $\{\rx,\rz\}$ in addition to\\
CNOT cost $12.5n-30$ ($12.5n-26.5$ for odd $n$) and depth $12n-27$ ($12n-24$ for odd $n$),\\
T ~~~~~ cost $16n-48$ ~~~~~~~~~~~~~~~~~~~~~~~~~~~~~~~~ and depth $4n$,\\
H ~~~~~ cost $9n-36$ ($9n-37$ for odd $n$) ~~~~~~~~ and depth $4n-10$,\\
S ~~~~~ cost $0.5n-2$ ($0.5n-2.5$ for odd $n$) ~~~~~ and depth $1$.
\end{theorem}
We note that it is possible to use the structure which produces \thm{MCSU2_costs}, and use \qc{toff_deco_clifft_t_depth} to replace only the $6(n_1-2)$ $[X_\Delta]$ gates which allow CNOT cancellations according to \qc{mid_layer_reduction}. This approach maintains all of the costs described in \thm{MCSU2_costs}, while reducing the T depth by $3n+O(1)$, and increasing the CNOT depth by the same. This provides a T depth of $5n+O(1)$, with no increase to the CNOT cost.

The results mentioned in \thm{MCSU2_costs_tdepth} can be used, along with \thm{MCX_costs_short} and \thm{MCU2_costs} in order to obtain the reduced T depth implementations of MCX and MCU2 gates.

We summarize the results for multi-controlled single target gates with reduced T depth in \tab{tcosts_table_ATA}. We use the same comparisons as used in \tab{costs_table_ATA}.
The depth reduction for T gates is now between $75\%-87.5\%$, while maintaining low costs.

\begin{table}[H]
    \centering
    \begin{tabular}{|p{1.2cm}|p{1.5cm}|p{1.2cm}|p{1cm}|p{1cm}|p{1cm}|p{1cm}|p{1cm}|p{1cm}|p{0.5cm}|p{0.5cm}|p{0.5cm}|p{0.5cm}|}
 \hline
 \multicolumn{3}{|c|}{Gate} & \multicolumn{2}{c|}{CNOT} & \multicolumn{2}{c|}{T} \\
  \hline
   Type & Ancilla & Source & Cost & Depth & Cost & Depth \\
 \hline
    MCSU2    &  None    & \cellcolor{lightgray!60} \cite{vale_decomposition_2023}       & \cellcolor{lightgray!60} $20n$    & \cellcolor{lightgray!60} $20n$  & \cellcolor{lightgray!60} $20n$    & \cellcolor{lightgray!60} $20n$ \\
        &      & Ours & $12.5n$  &  $12n$  &   $16n$   &  $4n$ \\
 \hline
     MCX   &   One dirty     & \cellcolor{lightgray!60} \cite{iten_quantum_2016}       & \cellcolor{lightgray!60} $16n$    & \cellcolor{lightgray!60} $16n$  &  \cellcolor{lightgray!60} $16n$   & \cellcolor{lightgray!60} $16n$ \\
        &        & Ours & $12.5n$  &  $12n$  &   $16n$   &  $4n$ \\
 \hline
    MCU2    &   One clean    & \cellcolor{lightgray!60} \cite{barenco_elementary_1995},\cite{iten_quantum_2016}       & \cellcolor{lightgray!60} $32n$   &  \cellcolor{lightgray!60} $32n$  & \cellcolor{lightgray!60}  $32n$   & \cellcolor{lightgray!60} $32n$ \\
        &       & Ours & $12.5n$  &  $12n$  &   $16n$   &  $4n$  \\
 \hline
    \end{tabular}
    \caption{A summary of our results compared to previous methods.}
    \label{tab:tcosts_table_ATA}
\end{table}

{ \section{ATA Multi-controlled multi-target SU(2)} \label{apx:MCMTSU2}

We present a method to extend the multi-controlled $SU(2)$ structure to implement a multi-control multi-target $SU(2)$ gate (MCMTSU2), in which any $SU(2)$ operator can be applied to each target. We use the following structure.
\[
\scalebox{0.9}{
\Qcircuit @C=1.0em @R=0.4em @!R { 
	 	\nghost{{C_1} :  } & \lstick{{C_1} :  } & \qw {/^{n_1}} & \ctrl{1} & \qw \\
	 	\nghost{{C}_{2} :  } & \lstick{{C}_{2} :  } & \qw {/^{n_2}} & \ctrl{1} & \qw \\
\nghost{{t}_{1} :  } & \lstick{{t}_{1} :  } & \qw & \gate{\mathrm{W_1}}\qwx[1] & \qw \\
	 	\nghost{{t}_{2} :  } & \lstick{{t}_{2} :  } & \qw & \gate{\mathrm{W_2}}\qwx[1] & \qw \\
	 	\nghost{{t}_{3} :  } & \nghost{{t}_{3} :  } & &  \ar @{.} [1,0]  \\
        \nghost{{t}_{3} :  } & \nghost{{t}_{3} :  } & &    \\
	   \nghost{{t}_{m} :  } & \lstick{{t}_{m} :  } & \qw  & \gate{\mathrm{W_m}}\qwx[-1] & \qw \\
 }
\hspace{5mm}\raisebox{-19.5mm}{=}\hspace{0mm}
\Qcircuit @C=1.0em @R=0.2em @!R { 
	 	\nghost{{C}_{1} :  } & \lstick{{C}_{1} :  } & \qw {/^{n_1}}  & \ctrl{2} & \multigate{1}{\mathrm{\Delta_1}}  & \qw & \multigate{1}{\mathrm{\Delta_2}} & \ctrl{2} & \multigate{1}{\mathrm{\Delta_1^\dagger}} & \qw & \multigate{1}{\mathrm{\Delta_2^\dagger}} & \qw \\
	 	\nghost{{C}_{2} :  } & \lstick{{C}_{2} :  } & \qw {/^{n_2}} & \qw & \ghost{\mathrm{\Delta_1}}  & \ctrl{1} & \ghost{\mathrm{\Delta_2}} & \qw & \ghost{\mathrm{\Delta_1^\dagger}} & \ctrl{1} & \ghost{\mathrm{\Delta_2^\dagger}} & \qw \\
	 	\nghost{{t_1} :  } & \lstick{{t_1} :  }  & \gate{\mathrm{A^1_4}} & \gate{\mathrm{Z}}\ar @{-} [1,0] & \gate{\mathrm{A_2^1}} & \gate{\mathrm{Z}}\ar @{-} [1,0] & \gate{\mathrm{A_3^1}} & \gate{\mathrm{Z}}\ar @{-} [1,0] & \gate{\mathrm{A_2^1}} & \gate{\mathrm{Z}}\ar @{-} [1,0] & \gate{\mathrm{A_1^1}} & \qw \\
   \nghost{{t_2} :  } & \lstick{{t_2} :  }  & \gate{\mathrm{A_4^2}} & \gate{\mathrm{Z}}\ar @{-} [1,0] & \gate{\mathrm{A_2^2}} & \gate{\mathrm{Z}}\ar @{-} [1,0] & \gate{\mathrm{A_3^2}} & \gate{\mathrm{Z}}\ar @{-} [1,0] & \gate{\mathrm{A_2^2}} & \gate{\mathrm{Z}}\ar @{-} [1,0] & \gate{\mathrm{A_1^2}} & \qw \\
   \nghost{{t}_{3} :  } & \nghost{{t}_{3} :  } & &  \ar @{.} [1,0] & &  \ar @{.} [1,0] & &  \ar @{.} [1,0] & &  \ar @{.} [1,0] & &  \\
        \nghost{{t}_{3} :  } & \nghost{{t}_{3} :  } & & & & & & & & & &   \\
 \nghost{{t}_{m} :  } & \lstick{{t}_{m} :  }  & \gate{\mathrm{A_4^m}} & \gate{\mathrm{Z}}\ar @{-} [-1,0] & \gate{\mathrm{A_2^m}} & \gate{\mathrm{Z}}\ar @{-} [-1,0] & \gate{\mathrm{A_3^m}} & \gate{\mathrm{Z}}\ar @{-} [-1,0] & \gate{\mathrm{A_2^m}} & \gate{\mathrm{Z}}\ar @{-} [-1,0] & \gate{\mathrm{A_1^m}} & \qw \\
 }
 }
 \qcref{mcmtsu2_base}
\]

\begin{lemma}\label{lem:mcmtsu2_struct}
    Any $\prod_{j=1}^{m}\MCO{\rv{\lambda_j}
{\hat{v}_j}}{C}{t_j}$ gate with $\rv{\lambda_j}
{\hat{v}_j}\in SU(2)$ can be implemented using $8m$ gates from $\{\rx,\rz\}$, two $\prod_{j=1}^{m}\MCO{Z}{C_1}{t_j}$ and two $\prod_{j=1}^{m}\MCO{Z}{C_2}{t_j}$ gates as \qc{mcmtsu2_base} such that $A_2^j=\rv{-\frac{\lambda_j}{4}}{\hat{x}}$, $A_3^j=\rv{\frac{\lambda_j}{4}}{\hat{x}}$, $A_4^j=\rv{\theta_2^j}{\hat{z}}\rv{\theta_1^j}{\hat{x}}$, $A_1^j={A_4^j}^\dagger A_3^j$, and $C_1\cup C_2 = C$. The  $\sqO{\Delta}{C}$ gates apply a relative phase.
\end{lemma}
\begin{proof}
 This can be realised by applying the structure described in \lem{ccrv_macro} for each $\MCO{\rv{\lambda_j}
{\hat{v}_j}}{C}{t_j}$ gate separately. The gates can be reordered to achieve \qc{mcmtsu2_base} since MCZ gates commute with each other, and single-qubit gates operating on different qubits commute as well.
\end{proof}
The following identity can be derived from
\lem{basic_struct}.
\[
\Qcircuit @C=1.0em @R=0.7em @!R { 
\nghost{{C_1} :  } & \lstick{{C_1} :  } & \qw {/^{n_1}}  & \ctrl{2} & \multigate{1}{\mathrm{\Delta}} & \qw  \\	 	\nghost{{C_2} :  } & \lstick{{C_2} :  } & \qw {/^{n_2}}  & \qw & \ghost{\mathrm{\Delta}} & \qw   \\
	 	\nghost{{t}_{1} :  } & \lstick{{t}_{1} :  }  & \targ & \control\qw & \targ & \qw  \\
	 	\nghost{{t}_{2} :  } & \lstick{{t}_{2} :  }  & \ctrl{-1} & \qw & \ctrl{-1} & \qw  \\
 }
\hspace{5mm}\raisebox{-8mm}{=}\hspace{0mm}
\Qcircuit @C=1.0em @R=0.7em @!R { 
\nghost{{C_1} :  } & \lstick{{C_1} :  } & \qw {/^{n_1}}  & \ctrl{2} & \ctrl{3} & \multigate{1}{\mathrm{\Delta}} & \qw  \\	 	\nghost{{C_2} :  } & \lstick{{C_2} :  } & \qw {/^{n_2}}  & \qw & \qw & \ghost{\mathrm{\Delta}} & \qw   \\
	 	\nghost{{t}_{1} :  } & \lstick{{t}_{1} :  }  & \qw & \ctrl{0} & \qw & \qw & \qw  \\
	 	\nghost{{t}_{2} :  } & \lstick{{t}_{2} :  }  & \qw & \qw & \ctrl{0} & \qw & \qw  \\
 }
\hspace{5mm}\raisebox{-8mm}{=}\hspace{0mm}
\Qcircuit @C=1.0em @R=0.3em @!R { 
\nghost{{C_1} :  } & \lstick{{C_1} :  } & \qw {/^{n_1}}  & \ctrl{2} & \multigate{1}{\mathrm{\Delta}} & \qw  \\	 	\nghost{{C_2} :  } & \lstick{{C_2} :  } & \qw {/^{n_2}}  & \qw & \ghost{\mathrm{\Delta}} & \qw   \\
	 	\nghost{{t}_{1} :  } & \lstick{{t}_{1} :  }  & \qw & \gate{\mathrm{Z}}\ar @{-} [1,0] & \qw & \qw  \\
	 	\nghost{{t}_{2} :  } & \lstick{{t}_{2} :  }  & \qw & \gate{\mathrm{Z}} & \qw & \qw  \\
 } \qcref{add_mcz_new_target}
\]
\lem{mcmtz_struct_log} describes the structure used to extend any MCZ gates to a MCMTZ gate with $m$ targets, using $2(m-1)$ CNOT gates in depth $2\ceil{log_2(m)}$.
\begin{lemma}\label{lem:mcmtz_struct_log}
$\sqO{\Delta'}{C}\prod_{j=1}^{m}\MCO{Z}{C'}{t_j} =
(\prod_{k=1}^{\ceil{log_{2}(m)}}\prod_{l=1}^{f'}\MCO{X}{t_{l+f}}{t_{l}})\sqO{\Delta'}{C}
\MCO{Z}{C'}{t_1}(\prod_{k=1}^{\ceil{log_{2}(m)}}\prod_{l=1}^{f'}\MCO{X}{t_{l+f}}{t_{l}})^\dagger$ with $C'\in C$, and $t_j\not\in C$, such that $f=2^{k-1}$ and $f'=min(f,m-f)$. 
\end{lemma}
\begin{proof}
    As demonstrated in \qc{add_mcz_new_target}, it follows from \lem{basic_struct} that
    \[
    \sqO{\Delta'}{C}\MCO{Z}{C'}{t_j}\MCO{Z}{C'}{t_j+f}
    =
    \MCO{X}{t_{j+f}}{t_{j}}\sqO{\Delta'}{C}\MCO{Z}{C'}{t_j}\MCO{X}{t_{j+f}}{t_{j}}
    \]
The number of targets can therefor be doubled by applying a CNOT transformation in depth 2, as follows:
\[
   \sqO{\Delta'}{C}\prod_{j=1}^{2f}\MCO{Z}{C'}{t_j}=(\prod_{l=1}^{f}\MCO{X}{t_{l+f}}{t_{l}})\sqO{\Delta'}{C}\prod_{j=1}^{f}\MCO{Z}{C'}{t_j}(\prod_{l=1}^{f}\MCO{X}{t_{l+f}}{t_{l}})
\]
This can be repeatedly applied until $log_{2}(f)=\floor{log_{2}(m)}$. The remaining $m-f$ targets can be added similarly by applying $\prod_{l=1}^{m-f}\MCO{X}{t_{l+f}}{t_{l}}$ at each side. 

\end{proof}
\qc{mcmtsu2_a2a_targets} describes the structure for m=8.
\[
\Qcircuit @C=0.4em @R=0.59em @!R { 
	 	\nghost{{C_1} :  } & \lstick{{C_1} :  } & \qw {/^{n_1}}  & \qw & \qw & \ctrl{2} & \multigate{1}{\mathrm{\Delta}} & \qw & \qw & \qw \\
   \nghost{{C_2} :  } & \lstick{{C_2} :  } & \qw {/^{n_2}}  & \qw & \qw & \qw & \ghost{\mathrm{\Delta}}  & \qw & \qw & \qw \\
	 	\nghost{{t}_{1} :  } & \lstick{{t}_{1} :  }  & \targ & \targ & \targ & \control\qw & \targ & \targ & \targ & \qw \\
	 	\nghost{{t}_{5} :  } & \lstick{{t}_{5} :  }  & \ctrl{-1} & \qw & \qw & \qw & \qw & \qw & \ctrl{-1}  & \qw\\
	 	\nghost{{t}_{3} :  } & \lstick{{t}_{3} :  }  & \targ & \ctrl{-2} & \qw & \qw & \qw & \ctrl{-2} & \targ  & \qw\\
	 	\nghost{{t}_{7} :  } & \lstick{{t}_{7} :  }  & \ctrl{-1} & \qw & \qw & \qw & \qw & \qw & \ctrl{-1}  & \qw\\
	 	\nghost{{t}_{2} :  } & \lstick{{t}_{2} :  }  & \targ & \targ & \ctrl{-4} & \qw & \ctrl{-4} & \targ & \targ & \qw \\
	 	\nghost{{t}_{6} :  } & \lstick{{t}_{6} :  }  & \ctrl{-1} & \qw & \qw & \qw & \qw & \qw & \ctrl{-1}  & \qw\\
	 	\nghost{{t}_{4} :  } & \lstick{{t}_{4} :  }  & \targ & \ctrl{-2} & \qw & \qw & \qw & \ctrl{-2} & \targ  & \qw\\
	 	\nghost{{t}_{8} :  } & \lstick{{t}_{8} :  }  & \ctrl{-1} & \qw & \qw & \qw & \qw & \qw & \ctrl{-1} & \qw \\
 }
\hspace{5mm}\raisebox{-23.5mm}{=}\hspace{0mm}
\Qcircuit @C=0.5em @R=0.2em @!R { 
\nghost{{C_1} :  } & \lstick{{C_1} :  } & \qw {/^{n_1}}  & \ctrl{2} & \multigate{1}{\mathrm{\Delta}} & \qw  \\	 	\nghost{{C_2} :  } & \lstick{{C_2} :  } & \qw {/^{n_2}}  & \qw & \ghost{\mathrm{\Delta}} & \qw   \\
	 	\nghost{{t}_{1} :  } & \lstick{{t}_{1} :  }  & \qw & \gate{\mathrm{Z}}\ar @{-} [1,0] & \qw & \qw  \\
	 	\nghost{{t}_{5} :  } & \lstick{{t}_{5} :  }  & \qw & \gate{\mathrm{Z}}\ar @{-} [1,0] & \qw & \qw  \\
   \nghost{{t}_{3} :  } & \lstick{{t}_{3} :  }  & \qw & \gate{\mathrm{Z}}\ar @{-} [1,0] & \qw & \qw  \\
   \nghost{{t}_{7} :  } & \lstick{{t}_{7} :  }  & \qw & \gate{\mathrm{Z}}\ar @{-} [1,0] & \qw & \qw  \\
   \nghost{{t}_{2} :  } & \lstick{{t}_{2} :  }  & \qw & \gate{\mathrm{Z}}\ar @{-} [1,0] & \qw & \qw  \\
   \nghost{{t}_{6} :  } & \lstick{{t}_{6} :  }  & \qw & \gate{\mathrm{Z}}\ar @{-} [1,0] & \qw & \qw  \\
   \nghost{{t}_{4} :  } & \lstick{{t}_{4} :  }  & \qw & \gate{\mathrm{Z}}\ar @{-} [1,0] & \qw & \qw  \\
   \nghost{{t}_{8} :  } & \lstick{{t}_{8} :  }  & \qw & \gate{\mathrm{Z}} & \qw & \qw  \\
 }
 \qcref{mcmtsu2_a2a_targets}
\]
The following can be obtained using  
\lem{mcmtsu2_struct} and \lem{mcmtz_struct_log}.
\begin{theorem}\label{thm:MCMTSU2_costs}
    Any multi-controlled multi-target $SU(2)$ with $n\geq 6$ controls, and $m\geq 1$ targets can be implemented without ancilla in the same costs and depths of an MCSU2 gate implemented using \lem{ccrv_macro}, with an increase of $8(m-1)$ to both $\{\rx/\rz\}$ and CNOT count, and $8\ceil{log_{2}(m)}$ to CNOT depth.
\end{theorem}

The final costs and depths of our implementation can therefore be achieved by applying these adjustments to the costs and depths listed in \thm{MCSU2_costs}.
}

{  \section{ATA MCMTX}\label{apx:ata_MCMTX}
We extend the multi-controlled $X$ structure to a multi-control multi-target $X$ gate (MCMTX).
First we demonstrate that if the number of target qubits is larger than one, an ancilla qubit is not required.
The following can be easily deduced from \lem{basic_struct} (and has been shown in \cite{maslov_toffoli_2005}). 
\[
\Qcircuit @C=1.0em @R=0.8em @!R { 
  \nghost{{C}_{1} :  } & \lstick{{C}_{1} :  } & \qw {/^{n_1}} & \ctrl{1} & \qw \\
	 	\nghost{{C}_{2} :  } & \lstick{{C}_{2} :  } & \qw {/^{n_2}} & \ctrl{1} & \qw \\
	 	\nghost{{t_1} :  } & \lstick{{t_1} :  } & \qw & \targ\ar @{-} [1,0] & \qw \\
	 	\nghost{{t_2} :  } & \lstick{{t_2} :  } & \qw & \targ & \qw \\
 }
 \hspace{5mm}\raisebox{-9mm}{=}\hspace{0mm}
\Qcircuit @C=1.0em @R=0.8em @!R { 
  \nghost{{C}_{1} :  } & \lstick{{C}_{1} :  } & \qw {/^{n_1}} & \ctrl{1} & \qw & \qw \\
	 	\nghost{{C}_{2} :  } & \lstick{{C}_{2} :  } & \qw {/^{n_2}} & \ctrl{1} & \qw & \qw \\
	 	\nghost{{t_1} :  } & \lstick{{t_1} :  } & \ctrl{1} & \targ & \ctrl{1} & \qw \\
	 	\nghost{{t_2} :  } & \lstick{{t_2} :  } & \targ & \qw & \targ & \qw \\
 }
 \qcref{add_targ_mcx}
\]

Since $t_2$ is not used while the MCX gate is applied, it can be utilized as the required dirty ancilla previously refered to as qubit $a$. Similarly, \lem{mcmtz_struct_log}  can be used in order to add any number of targets as follows.

\[
\scalebox{0.7}{
\Qcircuit @C=0.5em @R=0.6em @!R { 
\nghost{{C_1} :  } & \lstick{{C_1} :  } & \qw {/^{n_1}}  & \ctrl{1} & \qw & \qw  \\	 	\nghost{{C_2} :  } & \lstick{{C_2} :  } & \qw {/^{n_2}}  & \ctrl{1} & \qw & \qw   \\
	 	\nghost{{t}_{1} :  } & \lstick{{t}_{1} :  }  & \qw & \targ\ar @{-} [1,0] & \qw & \qw  \\
	 	\nghost{{t}_{5} :  } & \lstick{{t}_{5} :  }  & \qw & \targ\ar @{-} [1,0] & \qw & \qw  \\
   \nghost{{t}_{3} :  } & \lstick{{t}_{3} :  }  & \qw & \targ\ar @{-} [1,0] & \qw & \qw  \\
   \nghost{{t}_{7} :  } & \lstick{{t}_{7} :  }  & \qw & \targ\ar @{-} [1,0] & \qw & \qw  \\
   \nghost{{t}_{2} :  } & \lstick{{t}_{2} :  }  & \qw & \targ\ar @{-} [1,0] & \qw & \qw  \\
   \nghost{{t}_{6} :  } & \lstick{{t}_{6} :  }  & \qw & \targ\ar @{-} [1,0] & \qw & \qw  \\
   \nghost{{t}_{4} :  } & \lstick{{t}_{4} :  }  & \qw & \targ\ar @{-} [1,0] & \qw & \qw  \\
   \nghost{{t}_{8} :  } & \lstick{{t}_{8} :  }  & \qw & \targ & \qw & \qw  \\
 }
\hspace{5mm}\raisebox{-23.5mm}{=}\hspace{0mm}
\Qcircuit @C=0.5em @R=0.2em @!R { 
\nghost{{C_1} :  } & \lstick{{C_1} :  } & \qw {/^{n_1}}  & \ctrl{1} & \qw & \qw  \\	 	\nghost{{C_2} :  } & \lstick{{C_2} :  } & \qw {/^{n_2}}  & \ctrl{1} & \qw & \qw   \\
	 	\nghost{{t}_{1} :  } & \lstick{{t}_{1} :  }  & \gate{\mathrm{H}} & \gate{\mathrm{Z}}\ar @{-} [1,0] & \gate{\mathrm{H}}  & \qw  \\
	 	\nghost{{t}_{5} :  } & \lstick{{t}_{5} :  }  & \gate{\mathrm{H}} & \gate{\mathrm{Z}}\ar @{-} [1,0] & \gate{\mathrm{H}}  & \qw  \\
   \nghost{{t}_{3} :  } & \lstick{{t}_{3} :  }  & \gate{\mathrm{H}} & \gate{\mathrm{Z}}\ar @{-} [1,0] & \gate{\mathrm{H}}  & \qw  \\
   \nghost{{t}_{7} :  } & \lstick{{t}_{7} :  }  & \gate{\mathrm{H}} & \gate{\mathrm{Z}}\ar @{-} [1,0] & \gate{\mathrm{H}}  & \qw  \\
   \nghost{{t}_{2} :  } & \lstick{{t}_{2} :  }  & \gate{\mathrm{H}} & \gate{\mathrm{Z}}\ar @{-} [1,0] & \gate{\mathrm{H}}  & \qw  \\
   \nghost{{t}_{6} :  } & \lstick{{t}_{6} :  }  & \gate{\mathrm{H}} & \gate{\mathrm{Z}}\ar @{-} [1,0] & \gate{\mathrm{H}}  & \qw  \\
   \nghost{{t}_{4} :  } & \lstick{{t}_{4} :  }  & \gate{\mathrm{H}} & \gate{\mathrm{Z}}\ar @{-} [1,0] & \gate{\mathrm{H}}  & \qw  \\
   \nghost{{t}_{8} :  } & \lstick{{t}_{8} :  }  & \gate{\mathrm{H}} & \gate{\mathrm{Z}} & \gate{\mathrm{H}}  & \qw  \\
 }
\hspace{5mm}\raisebox{-23.5mm}{=}\hspace{0mm}
 \Qcircuit @C=0.4em @R=0.2em @!R { 
	 	\nghost{{C_1} :  } & \lstick{{C_1} :  } & \qw {/^{n_1}}  & \qw & \qw & \qw & \ctrl{1} & \qw & \qw & \qw & \qw & \qw  \\
   \nghost{{C_2} :  } & \lstick{{C_2} :  } & \qw {/^{n_2}} & \qw  & \qw & \qw & \ctrl{1} & \qw  & \qw & \qw & \qw & \qw  \\
	 	\nghost{{t}_{1} :  } & \lstick{{t}_{1} :  } & \gate{\mathrm{H}}  & \targ & \targ & \targ & \control\qw & \targ & \targ & \targ & \gate{\mathrm{H}} & \qw \\
	 	\nghost{{t}_{5} :  } & \lstick{{t}_{5} :  } & \gate{\mathrm{H}} & \ctrl{-1} & \qw & \qw & \qw & \qw & \qw & \ctrl{-1} & \gate{\mathrm{H}}  & \qw\\
	 	\nghost{{t}_{3} :  } & \lstick{{t}_{3} :  } & \gate{\mathrm{H}} & \targ & \ctrl{-2} & \qw & \qw & \qw & \ctrl{-2} & \targ & \gate{\mathrm{H}}  & \qw\\
	 	\nghost{{t}_{7} :  } & \lstick{{t}_{7} :  } & \gate{\mathrm{H}} & \ctrl{-1} & \qw & \qw & \qw & \qw & \qw & \ctrl{-1} & \gate{\mathrm{H}}  & \qw\\
	 	\nghost{{t}_{2} :  } & \lstick{{t}_{2} :  } & \gate{\mathrm{H}} & \targ & \targ & \ctrl{-4} & \qw & \ctrl{-4} & \targ & \targ & \gate{\mathrm{H}} & \qw \\
	 	\nghost{{t}_{6} :  } & \lstick{{t}_{6} :  } & \gate{\mathrm{H}} & \ctrl{-1} & \qw & \qw & \qw & \qw & \qw & \ctrl{-1} & \gate{\mathrm{H}}  & \qw\\
	 	\nghost{{t}_{4} :  } & \lstick{{t}_{4} :  } & \gate{\mathrm{H}} & \targ & \ctrl{-2} & \qw & \qw & \qw & \ctrl{-2} & \targ & \gate{\mathrm{H}}  & \qw\\
	 	\nghost{{t}_{8} :  } & \lstick{{t}_{8} :  } & \gate{\mathrm{H}} & \ctrl{-1} & \qw & \qw & \qw & \qw & \qw & \ctrl{-1} & \gate{\mathrm{H}} & \qw \\
 }
 \hspace{5mm}\raisebox{-23.5mm}{=}\hspace{0mm}
\Qcircuit @C=0.4em @R=0.59em @!R { 
	 	\nghost{{C_1} :  } & \lstick{{C_1} :  } & \qw {/^{n_1}}  & \qw & \qw & \ctrl{1} & \qw & \qw & \qw & \qw \\
   \nghost{{C_2} :  } & \lstick{{C_2} :  } & \qw {/^{n_2}}  & \qw & \qw & \ctrl{1} & \qw  & \qw & \qw & \qw \\
	 	\nghost{{t}_{1} :  } & \lstick{{t}_{1} :  }  & \control\qw & \control\qw & \control\qw & \targ & \control\qw & \control\qw & \control\qw & \qw \\
	 	\nghost{{t}_{5} :  } & \lstick{{t}_{5} :  }  & \targ\ar @{-} [-1,0] & \qw & \qw & \qw & \qw & \qw & \targ\ar @{-} [-1,0]  & \qw\\
	 	\nghost{{t}_{3} :  } & \lstick{{t}_{3} :  }  & \ctrl{0} & \targ\ar @{-}[-2,0] & \qw & \qw & \qw & \targ\ar @{-}[-2,0] & \ctrl{0} & \qw\\
	 	\nghost{{t}_{7} :  } & \lstick{{t}_{7} :  }  & \targ\ar @{-}[-1,0] & \qw & \qw & \qw & \qw & \qw & \targ\ar @{-}[-1,0]  & \qw\\
	 	\nghost{{t}_{2} :  } & \lstick{{t}_{2} :  }  & \ctrl{0} & \ctrl{0} & \targ\ar @{-}[-4,0] & \qw & \targ\ar @{-}[-4,0] & \ctrl{0} & \ctrl{0} & \qw \\
	 	\nghost{{t}_{6} :  } & \lstick{{t}_{6} :  }  & \targ\ar @{-}[-1,0] & \qw & \qw & \qw & \qw & \qw & \targ\ar @{-}[-1,0]  & \qw\\
	 	\nghost{{t}_{4} :  } & \lstick{{t}_{4} :  }  & \ctrl{0} & \targ\ar @{-}[-2,0] & \qw & \qw & \qw & \targ\ar @{-}[-2,0] & \ctrl{0}  & \qw\\
	 	\nghost{{t}_{8} :  } & \lstick{{t}_{8} :  }  & \targ\ar @{-}[-1,0] & \qw & \qw & \qw & \qw & \qw & \targ\ar @{-}[-1,0] & \qw \\
 } 
 }
 \qcref{mcmtX_a2a_targets}
\]
The last step is achieved by commuting and cancelling out the H gates, which results in transforming controls into targets and vice versa, as been mentioned in [\cite{maslov_depth_2022}, {\em Proposition 1}].
\thm{MCMTX_costs} simply follows.
\begin{theorem}\label{thm:MCMTX_costs}
    Any multi-controlled multi-target $X$ gate with $n\geq 5$ controls, and $m\geq 2$ targets can be implemented without ancilla in the same costs and depths of an MCX gate, with an increase of $2(m-1)$  to CNOT count, and $2\ceil{log_{2}(m)}$ to CNOT depth.
\end{theorem}
The costs and depths of our MCMTX structure can be achieved by merging \thm{MCMTX_costs} and \thm{MCX_costs_short}.
 As can be seen in \qc{mcmtX_a2a_targets}, $m-1$ targets can be used as dirty ancilla qubits for the implementation of the MCX gate, while one acts as the required ancilla qubit $a$, and the rest $m-2$ can be used in order to reduce the cost of the MCX gate. We discuss these cost reductions in \apx{extra_dirty}.

 The results of this section can be used to implement any multi-controlled multi-target $\Pi$ gate as well, simply by using \lem{MC_transform} in order to change the axes of rotations. This requires no more than four arbitrary $\{\rx,\rz\}$ gates per target.
}

\section{ATA MCMTSU2 with dirty ancillae}\label{apx:extra_dirty}

We provide a method which reduces the cost of each type of Clifford+T gate used in our decompositions if an additional set $\chi$ of $n_\chi \leq \floor{\frac{n-6}{2}}$ dirty ancilla qubits is available. Since any unused qubit in the circuit can be utilized as a dirty ancilla, this reduction can be applied in any scenario in which the multi controlled gate does not operate on all available qubits.

We note that the $[O_2]$ gates which commute out of the structure corresponding to $C_1$ to form $\Xi$ layers, as mentioned in \apx{deph_reduce}, can be cancelled out if those commute with the structure corresponding with $C_2$. Such commutations can be achieved by using available dirty ancilla qubits as part of $C_2'$ as previously defined. Each available dirty ancilla can be used in that way to reduce the cost by four $[O_2]$ gates, until the number of ancilla reaches $n+O(1)$, and all of these gates have been cancelled out.
Here we present a method which allows to double the number of gates which can be cancelled out due to each added ancilla.\\

The following structure can be used to implement any MCSU2 gate according to \lem{ccrv_macro}, with $C_1^\chi \cup C_1^2 = C_1$.

\[
\scalebox{1.0}{
\Qcircuit @C=1.0em @R=0.2em @!R { 
	 \nghost{{\chi} :  } & \lstick{{\chi} :  } & \qw {/^{n_\chi}} & \qw & \qw \\
     	 	\nghost{{C_1^\chi} :  } & \lstick{{C_1^\chi} :  } & \qw {/^{2n_\chi}} & \ctrl{1} & \qw \\
  \nghost{{C_1^2} :  } & \lstick{{C_1^2} :  } & \qw {/^{n_1'}} & \ctrl{1} & \qw \\
	 	\nghost{{C}_{2} :  } & \lstick{{C}_{2} :  } & \qw {/^{n_2}} & \ctrl{1} & \qw \\
	 	\nghost{{t} :  } & \lstick{{t} :  } & \qw & \rvgate{\lambda}{\hat{v}} & \qw \\
 }
\hspace{5mm}\raisebox{-15mm}{=}\hspace{0mm}
\Qcircuit @C=1.0em @R=0.4em @!R { 
	 \nghost{{\chi} :  } & \lstick{{\chi} :  } & \qw {/^{n_\chi}} & \qw & \multigate{3}{\mathrm{\Delta_1}} & \qw & \qw & \qw & \multigate{3}{\mathrm{\Delta_1^\dagger}} & \qw & \qw & \qw\\
  	 \nghost{{C_1^\chi} :  } & \lstick{{C_1^\chi} :  } & \qw {/^{2n_\chi}} & \ctrl{1} & \ghost{\mathrm{\Delta_1}} & \qw & \qw & \ctrl{1} & \ghost{\mathrm{\Delta_1}} & \qw & \qw & \qw\\
  \nghost{{C}_{1}^2 :  } & \lstick{{C}_{1}^2 :  } & \qw {/^{n_1'}}  & \ctrl{2} & \ghost{\mathrm{\Delta_1}}  & \qw & \multigate{1}{\mathrm{\Delta_2}} & \ctrl{2} & \ghost{\mathrm{\Delta_1 }} & \qw & \multigate{1}{\mathrm{\Delta_2^\dagger}} & \qw \\
	 	\nghost{{C}_{2} :  } & \lstick{{C}_{2} :  } & \qw {/^{n_2}} & \qw & \ghost{\mathrm{\Delta_1}}  & \ctrl{1} & \ghost{\mathrm{\Delta_2}} & \qw & \ghost{\mathrm{\Delta_1^\dagger}} & \ctrl{1} & \ghost{\mathrm{\Delta_2^\dagger}} & \qw \\
	 	\nghost{{t} :  } & \lstick{{t} :  }  & \gate{\mathrm{A_4}} & \ctrl{0} & \gate{\mathrm{A_2}} & \ctrl{0} & \gate{\mathrm{A_3}} & \ctrl{0} & \gate{\mathrm{A_2}} & \ctrl{0} & \gate{\mathrm{A_1}} & \qw \\
 }
 }
 \qcref{mcrpi_base_extra_ancil}
\]
The subsets are defined as follows:
\begin{table}[H]
     \centering
    \scalebox{0.8}{
    \begin{tabular}{ p{6cm} p{6cm} p{6cm} }
 $C_1^\chi=C[1:2n_\chi]$ & $C_1^2=C[2n_\chi+1:2n_\chi+n_1']$ &
 $C_2=C[2n_\chi+n_1'+1:n]$\\
 \\
  $C'_1 =
\begin{cases}
    C_1^2[3:n_2] & n_2=n_1'\\
    \{C_1^2[3:n_1'],C_1^2[2]\} & n_2\not =n_1'\\
\end{cases}$  & $C_2'=\{\chi,C_2[3:n_1']\}$ & $C_1=\{C_1^2[1:2],C_1^\chi,C_1^2[3:n_1']\}$\\
    \end{tabular}
}
\end{table}

The set marked as $C_1^\chi$ holds the first $2n_\chi$ qubits of the control set $C$. The remaining $n-2n_\chi$ control qubits are divided between the sets $C_1^2$ and $C_2$ of size $n_1'=\floor{\frac{n-2n_\chi}{2}}$ and $n_2=\ceil{\frac{n-2n_\chi}{2}}$ respectively. The set $C_1$ holds all $n_1=2n_\chi+\floor{\frac{n-2n_\chi}{2}}$ qubits of sets $C_1^\chi$ and $C_1^2$.

The gate $\sqO{\Delta_2}{\{C_2,C_1^2\}}\MCO{Z}{C_2}{t}=\MCO{S}{C_2[1:2]}{}\{Z\}^{n_2}_{C_2,\{C_1',t\}}$ is decomposed using \lem{Vchain_fullstruct_ivi} as shown in the previous sections. The gate $\sqO{\Delta_1}{\{C,\chi\}}\MCO{Z}{C_1}{t}$ is also decomposed using \lem{Vchain_fullstruct_ivi}, while the qubits in $C_1^\chi$ are arranged as $n_\chi$ pairs, and each pair is treated as a single control bit. This adjustment is simply achieved by adding a control bit to each $[X_\Delta]$ gate which is controlled by such a pair. We use the following notations for a three-controlled relative phase Toffoli and its inverse.

\[
\scalebox{1.0}{
\Qcircuit @C=1.0em @R=0.6em @!R { 
	 	\lstick{} & \ctrlzr{1}  & \qw\\
	 	\lstick{} & \ctrlt{1}  & \qw\\
	 	\lstick{} & \targ & \qw\\
	 	\lstick{} & \ctrl{-1}  & \qw\\
 }
 \hspace{5mm}\raisebox{-8mm}{=}\hspace{3mm}
\Qcircuit @C=1.0em @R=0.14em @!R { 
	 	\lstick{} & \ctrl{1} \barrier[0em]{3} & \qw & \qw & \ctrl{1} & \qw & \qw \\
	 	\lstick{} & \ctrl{1} & \qw & \ctrl{1} & \ctrl{2} & \ctrl{2} & \qw \\
	 	\lstick{} & \targ & \qw & \control\qw & \qw & \qw & \qw \\
	 	\lstick{} & \ctrl{-1} & \qw & \ctrl{-1} & \gate{\mathrm{S}} & \gate{\mathrm{S}} & \qw \\
 }
 \hspace{5mm}\raisebox{-15mm}{,}\hspace{3mm}
 \Qcircuit @C=1.0em @R=0.6em @!R { 
	 	\lstick{} & \ctrlzl{1}  & \qw\\
	 	\lstick{} & \ctrlt{1}  & \qw\\
	 	\lstick{} & \targ & \qw\\
	 	\lstick{} & \ctrl{-1}  & \qw\\
 }
\hspace{5mm}\raisebox{-8mm}{=}\hspace{3mm}
 \Qcircuit @C=1.0em @R=0. em @!R { 
	 	\lstick{} & \qw & \ctrl{1} & \qw \barrier[0em]{3} & \qw & \ctrl{1} & \qw & \qw\\
	 	\lstick{} & \ctrl{2} & \ctrl{2} & \ctrl{1} & \qw & \ctrl{1} & \qw & \qw\\
	 	\lstick{} & \qw & \qw & \control\qw & \qw & \targ & \qw & \qw\\
	 	\lstick{} & \gate{\mathrm{S^\dagger}} & \gate{\mathrm{S^\dagger}} & \ctrl{-1} & \qw & \ctrl{-1} & \qw & \qw\\
 }
 }
 \qcref{three_con_rf_toff_first}
\]

Using \lem{basic_struct}, we get the following version of \qc{add_mcz}.

\[
\Qcircuit @C=1.0em @R=0.2em @!R { 
	 	\nghost{} & \lstick{} & \qw {/^n}  & \ctrl{3} & \qw & \qw \\
	 	\nghost{} & \lstick{}  & \ctrlzr{1} & \qw & \ctrlzl{1} & \qw \\
   	 	\nghost{} & \lstick{}  & \ctrlt{1} & \qw & \ctrlt{1} & \qw \\
	 	\nghost{} & \lstick{} & \targ & \control\qw & \targ & \qw \\
	 	\nghost{} & \lstick{}  & \ctrl{-1} & \qw & \ctrl{-1} & \qw \\
 }
\hspace{5mm}\raisebox{-8mm}{=}\hspace{5mm}
 \Qcircuit @C=1.0em @R=0.3em @!R { 
	 	\lstick{} & \qw {/^n}  & \ctrl{3} & \qw & \qw \\
	 	 \lstick{}  & \ctrl{1} & \qw & \ctrl{1} & \qw \\
   	 	 \lstick{}  & \ctrl{1} & \qw & \ctrl{1} & \qw \\
	 	\lstick{} & \targ & \control\qw & \targ & \qw \\
	 	 \lstick{}  & \ctrl{-1} & \qw & \ctrl{-1} & \qw \\
 }
 \hspace{5mm}\raisebox{-8mm}{=}\hspace{5mm}
 \Qcircuit @C=1.0em @R=0.87em @!R { 
	 	 \lstick{} & \qw {/^n} & \ctrl{1} & \ctrl{3} & \qw \\
	 	 \lstick{} & \qw & \ctrl{1} & \qw & \qw \\
   	 	 \lstick{} & \qw & \ctrl{2} & \qw & \qw \\
	 	 \lstick{} & \qw & \qw & \control\qw & \qw \\
	 	 \lstick{} & \qw & \control\qw & \qw & \qw \\
 }
 \qcref{add_mcz_two}
\]

The following circuit presents the structure used to decompose $\sqO{\Delta_1}{\{C,\chi\}}\MCO{Z}{C_1}{t}=\MCO{S}{C_1[1:2]}{}\{Z\}^{n_1}_{C_1,\{C_2',t\}}$ according to \lem{Vchain_fullstruct} while pairs of qubits from $C_1^\chi$ are treated as a single qubit.

\[
\scalebox{0.7}{
\Qcircuit @C=1.0em @R=0.55em @!R { 
	 	 \lstick{} & \qw & \multigate{16}{\mathrm{\Delta}} & \qw\\
	 	\lstick{} & \qw & \ghost{{\mathrm{\Delta}}} & \qw\\
	 	\lstick{} & \qw & \ghost{{\mathrm{\Delta}}} & \qw\\
	 	\lstick{} & \ctrl{1} & \ghost{{\mathrm{\Delta}}} & \qw\\
	 	\lstick{} & \ctrl{1} & \ghost{{\mathrm{\Delta}}} & \qw\\
	 	\lstick{} & \ctrl{1} & \ghost{{\mathrm{\Delta}}} & \qw\\
	 	\lstick{} & \ctrl{1} & \ghost{{\mathrm{\Delta}}} & \qw\\
	 	\lstick{} & \ctrl{1} & \ghost{{\mathrm{\Delta}}} & \qw\\
	 	\lstick{} & \ctrl{1} & \ghost{{\mathrm{\Delta}}} & \qw\\
	 	\lstick{} & \ctrl{1} & \ghost{{\mathrm{\Delta}}} & \qw\\
	 	\lstick{} & \ctrl{1} & \ghost{{\mathrm{\Delta}}} & \qw\\
	 	\lstick{} & \ctrl{1} & \ghost{{\mathrm{\Delta}}} & \qw\\
	 	\lstick{} & \ctrl{1} & \ghost{{\mathrm{\Delta}}} & \qw\\
	 	\lstick{} & \ctrl{4} & \ghost{{\mathrm{\Delta}}} & \qw\\
	 	\lstick{} & \qw & \ghost{{\mathrm{\Delta}}} & \qw\\
	 	\lstick{} & \qw & \ghost{{\mathrm{\Delta}}} & \qw\\
	 	\lstick{} & \qw & \ghost{{\mathrm{\Delta}}} & \qw\\
	 	\lstick{t} & \control\qw & \qw & \qw \inputgroupv{1}{3}{1.3em}{1.5em}{\chi} \inputgroupv{4}{5}{1.3em}{0.8em}{C_1^2} \inputgroupv{6}{11}{1.3em}{3.5em}{C_1^\chi}
        \inputgroupv{12}{14}{1.3em}{1.5em}{C_1^2}
        \inputgroupv{15}{17}{1.3em}{1.5em}{C_2'}\\
 }
 \hspace{5mm}\raisebox{-40mm}{=}\hspace{10mm}
\Qcircuit @C=1.0em @R=0.17em @!R { 
	 	\lstick{} & \qw \barrier[0em]{17} & \qw & \qw & \qw & \qw & \qw & \qw & \control\qw & \qw & \qw & \qw\\
	 	\lstick{} & \qw & \qw & \qw & \qw & \qw & \qw & \ctrl{2} & \qw & \qw & \qw & \qw\\
	 	\lstick{} & \qw & \qw & \qw & \qw & \qw & \ctrl{1} & \qw & \qw & \qw & \qw & \qw\\
	 	\lstick{} & \ctrl{1} & \qw & \ctrl{1} & \ctrl{1} & \ctrl{1} & \ctrl{1} & \ctrl{1} & \ctrl{-3} & \ctrl{1} & \qw & \qw\\
	 	\lstick{} & \ctrl{1} & \qw & \ctrl{1} & \ctrl{1} & \ctrl{1} & \ctrl{1} & \ctrl{1} & \ctrl{-1} & \gate{\mathrm{S}} & \qw & \qw\\
	 	\lstick{} & \ctrl{1} & \qw & \ctrl{1} & \ctrl{1} & \ctrl{1} & \ctrl{1} & \ctrl{1} & \qw & \qw & \qw & \qw\\
	 	\lstick{} & \ctrl{1} & \qw & \ctrl{1} & \ctrl{1} & \ctrl{1} & \ctrl{1} & \control\qw & \qw & \qw & \qw & \qw\\
	 	\lstick{} & \ctrl{1} & \qw & \ctrl{1} & \ctrl{1} & \ctrl{1} & \ctrl{1} & \qw & \qw & \qw & \qw & \qw\\
	 	\lstick{} & \ctrl{1} & \qw & \ctrl{1} & \ctrl{1} & \ctrl{1} & \control\qw & \qw & \qw & \qw & \qw & \qw\\
	 	\lstick{} & \ctrl{1} & \qw & \ctrl{1} & \ctrl{1} & \ctrl{1} & \qw & \qw & \qw & \qw & \qw & \qw\\
	 	\lstick{} & \ctrl{1} & \qw & \ctrl{1} & \ctrl{1} & \ctrl{4} & \qw & \qw & \qw & \qw & \qw & \qw\\
	 	\lstick{} & \ctrl{1} & \qw & \ctrl{1} & \ctrl{4} & \qw & \qw & \qw & \qw & \qw & \qw & \qw\\
	 	\lstick{} & \ctrl{1} & \qw & \ctrl{4} & \qw & \qw & \qw & \qw & \qw & \qw & \qw & \qw\\
	 	\lstick{} & \ctrl{4} & \qw & \qw & \qw & \qw & \qw & \qw & \qw & \qw & \qw & \qw\\
	 	\lstick{} & \qw & \qw & \qw & \qw & \control\qw & \qw & \qw & \qw & \qw & \qw & \qw\\
	 	\lstick{} & \qw & \qw & \qw & \control\qw & \qw & \qw & \qw & \qw & \qw & \qw & \qw\\
	 	\lstick{} & \qw & \qw & \control\qw & \qw & \qw & \qw & \qw & \qw & \qw & \qw & \qw\\
	 	\lstick{t} & \control\qw & \qw & \qw & \qw & \qw & \qw & \qw & \qw & \qw & \qw & \qw \inputgroupv{1}{3}{1.3em}{1.5em}{\chi} \inputgroupv{4}{5}{1.3em}{0.8em}{C_1^2} \inputgroupv{6}{11}{1.3em}{3.5em}{C_1^\chi}
        \inputgroupv{12}{14}{1.3em}{1.5em}{C_1^2}
        \inputgroupv{15}{17}{1.3em}{1.5em}{C_2'}\\
 }
  \hspace{5mm}\raisebox{-40mm}{=}\hspace{10mm}
\Qcircuit @C=0.5em @R=0.17em @!R { 
	 	\lstick{} & \qw & \qw & \qw & \qw & \qw & \targ & \gate{\mathrm{iZ}} & \targ & \qw & \qw & \qw & \qw & \qw & \qw & \qw\\
	 	\lstick{} & \qw & \qw & \qw & \qw & \targ & \ctrl{-1} & \qw & \ctrl{-1} & \targ & \qw & \qw & \qw & \qw & \qw & \qw\\
	 	\lstick{} & \qw & \qw & \qw & \targ & \ctrl{-1} & \qw & \qw & \qw & \ctrl{-1} & \targ & \qw & \qw & \qw & \qw & \qw\\
	 	\lstick{} & \qw & \qw & \qw & \qw & \qw & \qw & \ctrl{-3} & \qw & \qw & \qw & \qw & \qw & \qw & \qw & \qw\\
	 	\lstick{} & \qw & \qw & \qw & \qw & \qw & \qw & \ctrl{-1} & \qw & \qw & \qw & \qw & \qw & \qw & \qw & \qw\\
	 	\lstick{} & \qw & \qw & \qw & \qw & \qw & \ctrlzr{-4} & \qw & \ctrlzl{-4} & \qw & \qw & \qw & \qw & \qw & \qw & \qw\\
	 	\lstick{} & \qw & \qw & \qw & \qw & \qw & \ctrlt{-1} & \qw & \ctrlt{-1} & \qw & \qw & \qw & \qw & \qw & \qw & \qw\\
	 	\lstick{} & \qw & \qw & \qw & \qw & \ctrlzr{-5} & \qw & \qw & \qw & \ctrlzl{-5} & \qw & \qw & \qw & \qw & \qw & \qw\\
	 	\lstick{} & \qw & \qw & \qw & \qw & \ctrlt{-1} & \qw & \qw & \qw & \ctrlt{-1} & \qw & \qw & \qw & \qw & \qw & \qw\\
	 	\lstick{} & \qw & \qw & \qw & \ctrlzr{-7} & \qw & \qw & \qw & \qw & \qw & \ctrlzl{-7} & \qw & \qw & \qw & \qw & \qw\\
	 	\lstick{} & \qw & \qw & \qw & \ctrlt{-1} & \qw & \qw & \qw & \qw & \qw & \ctrlt{-1} & \qw & \qw & \qw & \qw & \qw\\
	 	\lstick{} & \qw & \qw & \ctrlt{3} & \qw & \qw & \qw & \qw & \qw & \qw & \qw & \ctrlt{3} & \qw & \qw & \qw & \qw\\
	 	\lstick{} & \qw & \ctrlt{3} & \qw & \qw & \qw & \qw & \qw & \qw & \qw & \qw & \qw & \ctrlt{3} & \qw & \qw & \qw\\
	 	\lstick{} & \ctrlt{3} & \qw & \qw & \qw & \qw & \qw & \qw & \qw & \qw & \qw & \qw & \qw & \ctrlt{3} & \qw & \qw\\
	 	\lstick{} & \qw & \qw & \targ & \ctrl{-4} & \qw & \qw & \qw & \qw & \qw & \ctrl{-4} & \targ & \qw & \qw & \qw & \qw\\
	 	\lstick{} & \qw & \targ & \ctrl{-1} & \qw & \qw & \qw & \qw & \qw & \qw & \qw & \ctrl{-1} & \targ & \qw & \qw & \qw\\
	 	\lstick{} & \targ & \ctrl{-1} & \qw & \qw & \qw & \qw & \qw & \qw & \qw & \qw & \qw & \ctrl{-1} & \targ & \qw & \qw\\
	 	\lstick{t} & \ctrl{-1} & \qw & \qw & \qw & \qw & \qw & \qw & \qw & \qw & \qw & \qw & \qw & \ctrl{-1} & \qw & \qw \inputgroupv{1}{3}{1.3em}{1.5em}{\chi} \inputgroupv{4}{5}{1.3em}{0.8em}{C_1^2} \inputgroupv{6}{11}{1.3em}{3.5em}{C_1^\chi}
        \inputgroupv{12}{14}{1.3em}{1.5em}{C_1^2}
        \inputgroupv{15}{17}{1.3em}{1.5em}{C_2'}\\
 }}
 \qcref{spaceship_struct}
\]
In order to decompose this structure in terms of Clifford+T gates, we use the following implementation, as provided in \cite{maslov_advantages_2016}.

\[
\scalebox{1.0}{
\Qcircuit @C=1.0em @R=0.6em @!R { 
	 	\lstick{} & \ctrlzr{1}  & \qw\\
	 	\lstick{} & \ctrlt{1}  & \qw\\
	 	\lstick{} & \targ & \qw\\
	 	\lstick{} & \ctrl{-1}  & \qw\\
 }
 \hspace{5mm}\raisebox{-8mm}{=}\hspace{5mm}
 \Qcircuit @C=0.5em @R=0.em @!R { 
	 	\lstick{} & \qw & \qw & \ctrl{2} & \qw & \qw & \qw & \qw & \qw \barrier[0em]{3} & \qw & \qw & \qw & \qw & \qw & \qw & \qw & \qw & \ctrl{2} & \qw & \qw & \qw & \qw\\
	 	\lstick{} & \qw & \qw & \qw & \qw & \qw & \qw & \ctrl{1} & \qw & \qw & \qw & \qw & \ctrl{1} & \qw & \qw & \qw & \qw & \qw & \qw & \qw & \qw & \qw\\
	 	\lstick{} & \gate{\mathrm{H}} & \gate{\mathrm{T^\dagger}} & \targ & \gate{\mathrm{T}} & \gate{\mathrm{H}} & \gate{\mathrm{T^\dagger}} & \targ & \gate{\mathrm{T}} & \qw & \targ & \gate{\mathrm{T^\dagger}} & \targ & \gate{\mathrm{T}} & \targ & \gate{\mathrm{H}} & \gate{\mathrm{T^\dagger}} & \targ & \gate{\mathrm{T}} & \gate{\mathrm{H}} & \qw & \qw\\
	 	\lstick{} & \qw & \qw & \qw & \qw & \qw & \qw & \qw & \qw & \qw & \ctrl{-1} & \qw & \qw & \qw & \ctrl{-1} & \qw & \qw & \qw & \qw & \qw & \qw & \qw\\
 }
 }
 \qcref{cccxd_cnot}
\]

As can be seen, when this circuit is used in order to decompose \qc{spaceship_struct}, all gates on the left to the barrier can commute with any gate which is applied on the left hand side, and for the inverse gates, a similar set of gate can commute to the right hand side. Moreover, as can be seen from \qc{mcrpi_base_extra_ancil}, these gates commute with the structures controlled by $C_2$, and thus gates which commute from the right hand side of the first structure controlled by $C_1$ cancel out with their counterparts from the left of the second structure controlled by $C_1$. 

The rest of these commuted gates, which commute to the right and to the left of the full MCSU2 structure, cancel out as well since all of these gates operate on a qubit from $\chi$, either as a single qubit gate, or as a gate controlled by a qubit from $C$. Therefore, a similar sequence of these gates and their inverse can be applied (analogically to \qc{commute_sides}) on both sides of the MCSU2 gate shown in \qc{mcrpi_base_extra_ancil} without changing the operator, and cancel out the gates which commute towards the edges.
The rest of the $[X_\Delta]$ gates in \qc{spaceship_struct} are decomposed in the same ways as \qc{rf_toffoli}, creating two $\Xi$ operators which are applied on $C_2' \cup C_1^2$.\\

The gate $\sqO{\Delta_1}{\{C,\chi\}}\MCO{Z}{C_1}{t}$ can be implemented using one $[iZ]$ gate, $2n_1'-4$ $[X_\Delta]$ gates, and $2n_\chi$ three-controlled relative phase Toffoli gates. Noting that due to our choice of $C_1,C_2,C_1',C_2'$, the same depth reductions or CNOT cancellations between pairs of $\Xi$ operators can be achieved $n_1'-2$ times per layer for gates applied on $C_2' \cup C_1^2$.  
The total costs and depths of the MCSU2 structure can be simply achieved from \thm{MCSU2_costs}, with $n$ replaced by $n-2n_\chi$, in addition to the costs and depths of $4n_\chi$ reduced three-controlled Toffoli gates, each contributing cost and depth of 4 CNOT, 4 T and 2 H gates.
We note that the steps used to produce \thm{MCMTSU2_costs} can be used in this case as well, as those are agnostic to the decomposition of the large relative phase gates. \thm{MCSU2_costs_ancil} simply follows.
\begin{theorem}\label{thm:MCSU2_costs_ancil}
    Any multi-controlled multi-target $SU(2)$ with $n\geq 6$ controls, and $m\geq 1$ targets can be implemented with $0\leq n_\chi \leq \floor{\frac{n-6}{2}}$ dirty ancilla qubits using $8m$ gates from $\{\rx,\rz\}$ in depth $8$, in addition to\\
CNOT cost $12n + 8m-8n_\chi-40$ and depth $8n+8\ceil{log_2(m)}-8$,\\
T ~~~~~ cost $16n-16n_\chi-48$ ~~~~~ and depth $8n-6$ ($8n-3$ for odd $n$),\\
H ~~~~~ cost $8n- 8n_\chi-32$ ~~~~~~~~ and depth $4n-11$.
\end{theorem}

For a minimized T depth decomposition, the same approach is used, with $[X_\Delta]$ and $[iZ]$ gates implemented as \qc{toff_deco_clifft_t_depth} and \qc{ccz_deco_clifft_t_depth}, and the three-controlled relative phase Toffoli implemented as the following circuit.

\[
\scalebox{1.0}{
\Qcircuit @C=1.0em @R=0.5em @!R { 
	 	\lstick{} & \ctrlzr{1}  & \qw\\
	 	\lstick{} & \ctrlt{1}  & \qw\\
	 	\lstick{} & \targ & \qw\\
	 	\lstick{} & \ctrl{-1}  & \qw\\
 }
 \hspace{5mm}\raisebox{-8mm}{=}\hspace{5mm}
 }
 \Qcircuit @C=0.5em @R=0.em @!R { 
	 	\lstick{} & \qw & \qw & \ctrl{2} & \qw & \ctrl{2} & \qw & \qw & \qw & \qw & \qw \barrier[0em]{3} & \qw & \qw & \qw & \qw & \qw & \qw & \qw & \targ & \gate{\mathrm{T^\dagger}} & \targ & \qw & \qw & \qw\\
	 	\lstick{} & \qw & \qw & \qw & \qw & \qw & \qw & \qw & \ctrl{1} & \qw & \ctrl{1} & \qw & \qw & \targ & \gate{\mathrm{T^\dagger}} & \targ & \qw & \qw & \qw & \qw & \qw & \qw & \qw & \qw\\
	 	\lstick{} & \gate{\mathrm{H}} & \gate{\mathrm{T^\dagger}} & \targ & \gate{\mathrm{T}} & \targ & \gate{\mathrm{H}} & \gate{\mathrm{T^\dagger}} & \targ & \gate{\mathrm{T}} & \targ & \qw & \targ & \ctrl{-1} & \gate{\mathrm{T}} & \ctrl{-1} & \targ & \gate{\mathrm{H}} & \ctrl{-2} & \gate{\mathrm{T}} & \ctrl{-2} & \gate{\mathrm{H}} & \qw & \qw\\
	 	\lstick{} & \qw & \qw & \qw & \qw & \qw & \qw & \qw & \qw & \qw & \qw & \qw & \ctrl{-1} & \qw & \qw & \qw & \ctrl{-1} & \qw & \qw & \qw & \qw & \qw & \qw & \qw\\
 }
 \qcref{cccxd_tdepth}
\]

The gates on the left of the barrier are removed in the same manner as described above, and the CNOT cancellations which provide \thm{MCSU2_costs_tdepth} can be applied for gates operating on $C_2' \cup C_1^2$. 
The total costs and depths of the MCSU2 structure can be simply achieved from \thm{MCSU2_costs_tdepth}, with $n$ replaced by $n-2n_\chi$, and in addition to the costs and depths of $4n_\chi$ reduced three-controlled Toffoli gates, each contributing 4 T gates in depth 2, in addition to cost and depth of 6 CNOT and 2 H gates.

\begin{theorem}\label{thm:MCSU2_costs_tdepth_ancil}
    Any multi-controlled multi-target $SU(2)$ with $n\geq 6$ controls, and $m\geq 1$ targets can be implemented with $0\leq n_\chi \leq \floor{\frac{n-6}{2}}$ dirty ancilla qubits using $8m$ gates from $\{\rx,\rz\}$ in depth $8$, in addition to\\
CNOT cost $12.5n+8m-n_\chi-38$ ($+3.5$ for odd $n$) and depth $12n+8\ceil{log_2(m)}-27$ ($+3$ for odd $n$),\\
T ~~~~~ cost $16n-16n_\chi-48$ ~~~~~~~~~~~~~~~~~~~~~~~~~~~ and depth $4n$,\\
H ~~~~~ cost $9n-10n_\chi-36$ ($-1$ for odd $n$) ~~~~~~~~~ and depth $4n-10$,\\
S ~~~~~ cost $0.5n-n_\chi-2$ ($-0.5$ for odd $n$) ~~~~~~~~~ and depth $1$.
\end{theorem}

We note that the number of dirty ancilla qubits can be increased up to $\ceil{\frac{n-2}{2}}$. This is achieved by allowing the implementation of the second and fourth MCZ and $[\Delta]$ gates in \qc{mcrpi_base_extra_ancil} to be implemented as $[iZ]$ or CZ gates when $n_2$ is reduced to 2 or 1. Finally, if $n_2=0$, these controlled gates are replaced by Z gates, which can commute and cancel out. This final stage might increase the Clifford+T count slightly, however it reduces the number of single qubit arbitrary rotations, and halves the cost/depth of CNOT gates which are added for for $m>1$ targets. Note that if this is used, and $n$ is odd, the last control qubit will be added using a $[X_\Delta]$ gate with reduced cost, since the gates which can commute outwards cancel out similarly to the three-controlled Toffoli gates.

The results mentioned in \thm{MCSU2_costs_ancil} and \thm{MCSU2_costs_tdepth_ancil} can be used, together with \thm{MCX_costs_short}-\thm{MCU2_costs} in order to obtain the implementations of the MCX, MCMTX, MCU2 and MCMTU2 gates with additional dirty ancilla qubits, noting that for MCMTX, $m-2$ targets can be used as dirty ancilla if for MCMTX with $m > 2$.

We present our CNOT and T costs of MCSU2, MCX and MCU2 as a function of the number of available dirty ancilla qubits in \fig{fig_ancil_CNOT_cost}. For comparison, we use the known results from \tab{costs_table_ATA}, noting that additional ancilla qubits do not improve the costs in those methods.
When around $n/2$ dirty ancilla qubits are available, our results for MCSU2, MCX an MCU2 are close to those provided by \cite{maslov_advantages_2016} for MCX with a similar number of ancilla. When the number of ancilla is lower, our methods still benefit from this resource, while other methods use at most one of the available ancilla qubits.

\begin{figure}[H]
\centering

\includegraphics[scale=0.8]{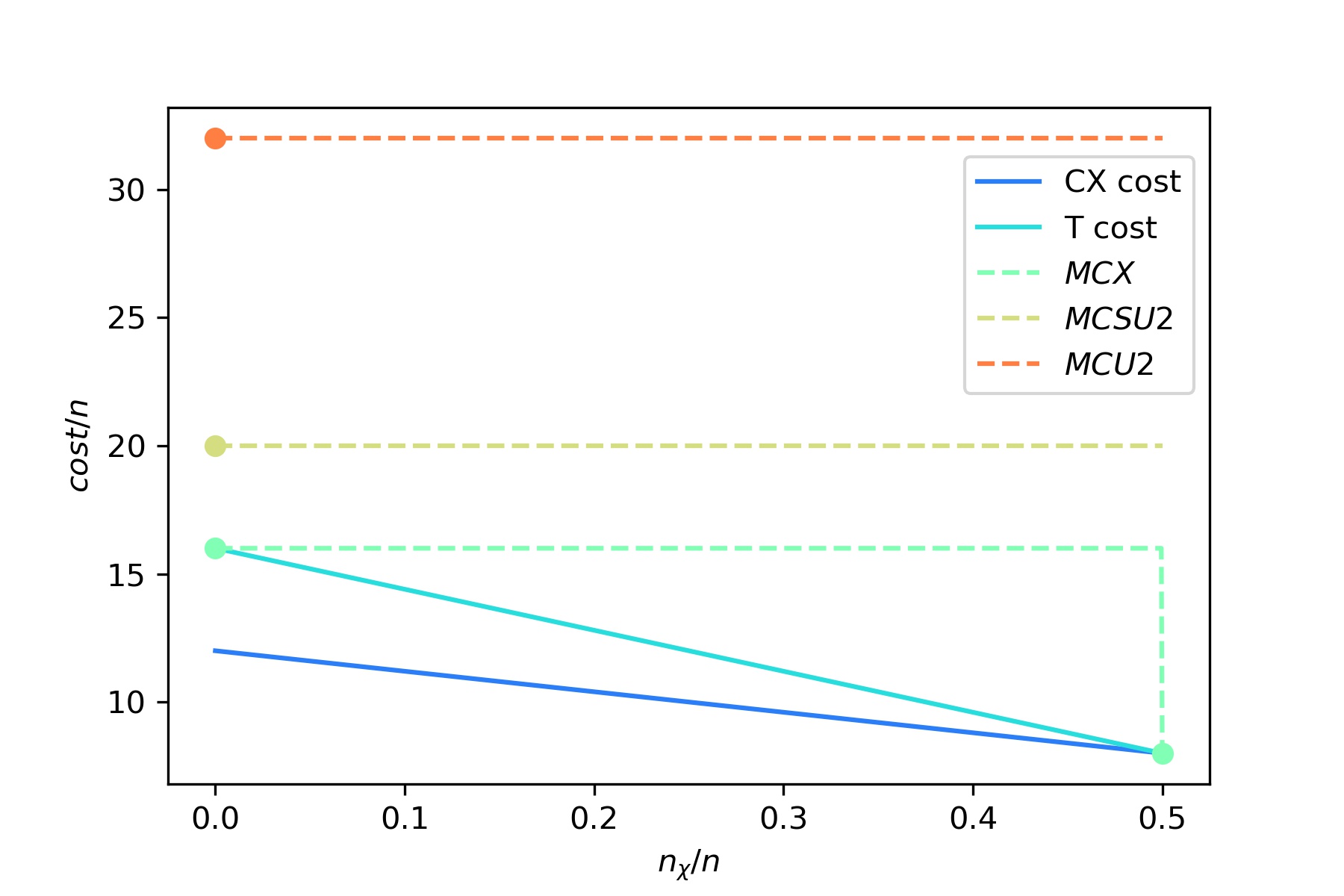}
\caption{Presenting the leading coefficient of CNOT and T costs for a single target, as a function of available dirty ancilla qubit. In our implementations, the leading terms are the same for MCX, MCU2, and MCSU2, and are marked as CX or T cost. Each of the best known methods that we use for comparison (\tab{costs_table_ATA}) provides the same leading term for CNOT and T cost. The results provided for these methods are marked as dots, which are continued as a dashed line if there is no improvement due to additional ancilla. For MCX, the results from \cite{maslov_advantages_2016} are shown for $\frac{n}{2}+O(1)$ ancilla.}  

\label{fig:fig_ancil_CNOT_cost}
\end{figure}

\renewcommand{\controlz}{*!<-0.28em,.0em>-=-<.2em>{\blacktriangleright}}

\renewcommand{\controlzl}{*!<0.28em,.0em>-=-<.2em>{\blacktriangleleft}}

\section{LNN communication overhead}\label{apx:communication_overhead}
In \sec{MCSU2_decomposition}, SWAP chains are used in order to relocate the target qubit $t$ of an LNN MCSU2.
We present the method used to reduce the CNOT cost required to implement the SWAP chains in \thm{lnn_su2_worst}.
Each SWAP gate can be replaced with a partial SWAP that can be implemented using two CNOT gates instead of three. We use the following gate notations.

\[
 \Qcircuit @C=1.0em @R=1.7em @!R { 	 	 \lstick{{}} & \qswap & \qw\\
	 	 \lstick{{}} & \ctrlzr{-1}  & \qw \\
 }
 \hspace{3mm}\raisebox{-3mm}{=}\hspace{3mm}
  \Qcircuit @C=1.0em @R=1.2em @!R { 
	 	 \lstick{{}} & \ctrl{1} & \targ & \qw\\
	 	 \lstick{{}} & \targ & \ctrl{-1}  & \qw \\
 }
 \hspace{3mm}\raisebox{-5mm}{,}\hspace{3mm}
  \Qcircuit @C=1.0em @R=1.7em @!R { 
	 	 \lstick{{}} & \qswap & \qw\\
	 	 \lstick{{}} & \ctrlzl{-1}  & \qw \\
 }
 \hspace{3mm}\raisebox{-3mm}{=}\hspace{3mm}
  \Qcircuit @C=1.0em @R=1.2em @!R { 
	 	 \lstick{{}} & \targ & \ctrl{1}  & \qw\\
	 	 \lstick{{}} & \ctrl{-1} & \targ   & \qw \\
 }
 \qcref{p_swapsies}
\]
such that 
\[
  \Qcircuit @C=1.0em @R=2em @!R { 
	 	 \lstick{{}} & \qswap & \qw\\
	 	 \lstick{{}} & \qswap\qwx{-1}  & \qw \\
 }
 \hspace{3mm}\raisebox{-3mm}{=}\hspace{3mm}
\Qcircuit @C=1.0em @R=1.2em @!R { 
	 	 \lstick{{}} & \qswap &\ctrl{1} & \qw\\
	 	 \lstick{{}} & \ctrlzr{-1} & \targ  & \qw \\
 }
  \hspace{3mm}\raisebox{-3mm}{=}\hspace{3mm}
    \Qcircuit @C=1.0em @R=1.2em @!R { 
	 	 \lstick{{}} & \ctrl{1} & \qswap & \qw\\
	 	 \lstick{{}} &\targ & \ctrlzl{-1}  & \qw \\
 }
 \qcref{p_swapsieses}
\]

Recalling from \lem{ccrv_macro} that $\MCO{\rv{\lambda}{\hat{v}}}{C}{t} = \sqO{A_1}{t}\sqO{A_3^\dagger}{t}\MCO{\rv{\lambda}{\hat{x}}}{C}{t}\sqO{A_4}{t}$, the SWAP chains can be applied as follows.

\[
\scalebox{0.7}{
\Qcircuit @C=1.0em @R=0.2em @!R { 
	 	\nghost{} & \lstick{} & \rvgate{\lambda}{\hat{v}} & \qw & \qw\\
	 	\nghost{} & \lstick{} & \ctrl{-1} & \qw & \qw\\
	 	\nghost{} & \lstick{} & \qw & \qw & \qw\\
	 	\nghost{} & \lstick{} & \ctrl{-2} & \qw & \qw\\
	 	\nghost{} & \lstick{} & \ctrl{-1} & \qw & \qw\\
	 	\nghost{} & \lstick{} & \qw & \qw & \qw\\
	 	\nghost{} & \lstick{} & \ctrl{-2} & \qw & \qw\\
 }
  \hspace{5mm}\raisebox{-20mm}{=}\hspace{0mm}
  \Qcircuit @C=0.6em @R=0.em @!R { 
	 	\nghost{} & \lstick{} & \gate{\mathrm{A_4}} & \qswap & \qw & \qw & \qw & \qw & \qw & \ctrl{2} & \qw & \qw & \qw & \qw & \qw & \qw & \qswap & \gate{\mathrm{A_1}} & \qw & \qw\\
	 	\nghost{} & \lstick{} & \qw & \qswap \qwx[-1] & \qswap & \qw & \qw & \qw & \qw & \qw & \qw & \qw & \qw & \qw & \qw & \qswap & \qswap \qwx[-1] & \qw & \qw & \qw\\
	 	\nghost{} & \lstick{} & \qw & \qw & \qswap \qwx[-1] & \qswap & \qw & \qw & \qw & \ctrl{1} & \qw & \qw & \qw & \qw & \qswap & \qswap \qwx[-1] & \qw & \qw & \qw & \qw\\
	 	\nghost{} & \lstick{} & \qw & \qw & \qw & \qswap \qwx[-1] & \qswap & \qw & \qw & \ctrl{2} & \qw & \qw & \qw & \qswap & \qswap \qwx[-1] & \qw & \qw & \qw & \qw & \qw\\
	 	\nghost{} & \lstick{} & \qw & \qw & \qw & \qw & \qswap \qwx[-1] & \qswap & \qw & \qw & \qw & \qw & \qswap & \qswap \qwx[-1] & \qw & \qw & \qw & \qw & \qw & \qw\\
	 	\nghost{} & \lstick{} & \qw & \qw & \qw & \qw & \qw & \qswap \qwx[-1] & \qswap & \ctrl{1} & \qw & \qswap & \qswap \qwx[-1] & \qw & \qw & \qw & \qw & \qw & \qw & \qw\\
	 	\nghost{} & \lstick{} & \qw & \qw & \qw & \qw & \qw & \qw & \qswap \qwx[-1] & \rvgate{\lambda}{\hat{x}} & \gate{\mathrm{A_3^\dagger}} & \qswap \qwx[-1] & \qw & \qw & \qw & \qw & \qw & \qw & \qw & \qw\\
 }
 \hspace{5mm}\raisebox{-20mm}{=}\hspace{0mm}
  \Qcircuit @C=0.6em @R=0.em @!R { 
	 	\nghost{} & \lstick{} & \gate{\mathrm{A_4}} & \qswap & \qw & \qw & \qw & \qw & \qw & \ctrl{2} & \qw & \qw & \qw & \qw & \qw & \qw & \qswap & \gate{\mathrm{A_1}} & \qw & \qw\\
	 	\nghost{} & \lstick{} & \qw & \ctrlzr{-1} & \qswap & \qw & \qw & \qw & \qw & \qw & \qw & \qw & \qw & \qw & \qw & \qswap & \ctrlzl{-1} & \qw & \qw & \qw\\
	 	\nghost{} & \lstick{} & \qw & \qw & \ctrlzr{-1} & \qswap & \qw & \qw & \qw & \ctrl{1} & \qw & \qw & \qw & \qw & \qswap & \ctrlzl{-1} & \qw & \qw & \qw & \qw\\
	 	\nghost{} & \lstick{} & \qw & \qw & \qw & \ctrlzr{-1} & \qswap & \qw & \qw & \ctrl{2} & \qw & \qw & \qw & \qswap & \ctrlzl{-1} & \qw & \qw & \qw & \qw & \qw\\
	 	\nghost{} & \lstick{} & \qw & \qw & \qw & \qw & \ctrlzr{-1} & \qswap & \qw & \qw & \qw & \qw & \qswap & \ctrlzl{-1} & \qw & \qw & \qw & \qw & \qw & \qw\\
	 	\nghost{} & \lstick{} & \qw & \qw & \qw & \qw & \qw & \ctrlzr{-1} & \qswap & \ctrl{1} & \qw & \qswap & \ctrlzl{-1} & \qw & \qw & \qw & \qw & \qw & \qw & \qw\\
	 	\nghost{} & \lstick{} & \qw & \qw & \qw & \qw & \qw & \qw & \ctrlzr{-1} & \rvgate{\lambda}{\hat{x}} & \gate{\mathrm{A_3^\dagger}} & \ctrlzl{-1} & \qw & \qw & \qw & \qw & \qw & \qw & \qw & \qw\\
 }
  }
  \qcref{swap_chain_t}
\]
with the final circuit achieved by applying CNOT cancellations between the SWAP chains. CNOT gates are first "extracted" from each side as follows.
\[
\scalebox{0.78}{
\Qcircuit @C=0.6em @R=0.4em @!R { 
	 	\nghost{} & \lstick{} & \qswap & \qw & \qw & \qw & \qw & \qw & \qw & \qw\\
	 	\nghost{} & \lstick{}  & \qswap \qwx[-1] & \qswap & \qw & \qw & \qw & \qw & \qw & \qw\\
	 	\nghost{} & \lstick{}  & \qw & \qswap \qwx[-1] & \qswap & \qw & \qw & \qw & \qw & \qw\\
	 	\nghost{} & \lstick{}  & \qw & \qw & \qswap \qwx[-1] & \qswap & \qw & \qw & \qw & \qw\\
	 	\nghost{} & \lstick{}  & \qw & \qw & \qw & \qswap \qwx[-1] & \qswap & \qw & \qw & \qw\\
	 	\nghost{} & \lstick{}  & \qw & \qw & \qw & \qw & \qswap \qwx[-1] & \qswap & \qw & \qw\\
	 	\nghost{} & \lstick{}  & \qw & \qw & \qw & \qw & \qw & \qswap \qwx[-1] & \qw & \qw\\
 }
  \hspace{5mm}\raisebox{-14mm}{=}\hspace{0mm}
  \Qcircuit @C=0.6em @R=0.4em @!R { 
	 	\nghost{} & \lstick{}  & \qswap & \qw & \qw & \qw & \qw & \qw & \ctrl{6} & \qw & \qw\\
	 	\nghost{} & \lstick{}  & \ctrlzr{-1} & \qswap & \qw & \qw & \qw & \qw & \qw & \qw & \qw\\
	 	\nghost{} & \lstick{}  & \qw & \qswap \qwx[-1] & \qswap & \qw & \qw & \qw & \qw & \qw & \qw\\
	 	\nghost{} & \lstick{}  & \qw & \qw & \qswap \qwx[-1] & \qswap & \qw & \qw & \qw & \qw & \qw\\
	 	\nghost{} & \lstick{}  & \qw & \qw & \qw & \qswap \qwx[-1] & \qswap & \qw & \qw & \qw & \qw\\
	 	\nghost{} & \lstick{}  & \qw & \qw & \qw & \qw & \qswap \qwx[-1] & \qswap & \qw & \qw & \qw\\
	 	\nghost{} & \lstick{}  & \qw & \qw & \qw & \qw & \qw & \qswap \qwx[-1] & \targ & \qw & \qw\\
 }
   \hspace{5mm}\raisebox{-14mm}{=}\hspace{0mm}
   \Qcircuit @C=0.5em @R=0.4em @!R { 
	 	\nghost{} & \lstick{}  & \qswap & \qw & \qw & \qw & \qw & \qw & \qw & \qw & \qw & \qw & \qw & \ctrl{6} & \qw & \qw\\
	 	\nghost{} & \lstick{}  & \ctrlzr{-1} & \qswap & \qw & \qw & \qw & \qw & \qw & \qw & \qw & \qw & \ctrl{5} & \qw & \qw & \qw\\
	 	\nghost{} & \lstick{}  & \qw & \ctrlzr{-1} & \qswap & \qw & \qw & \qw & \qw & \qw & \qw & \ctrl{4} & \qw & \qw & \qw & \qw\\
	 	\nghost{} & \lstick{}  & \qw & \qw & \ctrlzr{-1} & \qswap & \qw & \qw & \qw & \qw & \ctrl{3} & \qw & \qw & \qw & \qw & \qw\\
	 	\nghost{} & \lstick{}  & \qw & \qw & \qw & \ctrlzr{-1} & \qswap & \qw & \qw & \ctrl{2} & \qw & \qw & \qw & \qw & \qw & \qw\\
	 	\nghost{} & \lstick{}  & \qw & \qw & \qw & \qw & \ctrlzr{-1} & \qswap & \ctrl{1} & \qw & \qw & \qw & \qw & \qw & \qw & \qw\\
	 	\nghost{} & \lstick{}  & \qw & \qw & \qw & \qw & \qw & \ctrlzr{-1} & \targ & \targ & \targ & \targ & \targ & \targ & \qw & \qw\\
 }}
 \qcref{cnot_less_1}
\]
 The extracted CNOT gates can be removed since each such gate commutes with $\sqO{A_3^\dagger}{t}\MCO{\rv{\lambda}{\hat{x}}}{C}{t}$, as both operators apply a rotation about the $\hat{x}$ axis on the bottom qubit. In addition, the CNOT gates are controlled by qubits which are either controls of the MCSU2 gate, or are not affected by it.

Since each partial SWAP is decomposed as two CNOT gates, the maximal number of CNOT gates required for this reordering is $4d$, with $d$ as the "distance" of the swap, i.e. the number of qubits below the location of $t$. The CNOT depth is no larger than $2(d+2)$, as can be realised from the following. 
\[
\scalebox{0.8}{
 \Qcircuit @C=0.6em @R=0.4em @!R { 
	 	\nghost{} & \lstick{} & \qswap & \qw & \qw & \qw & \qw & \qw & \qw\\
	 	\nghost{} & \lstick{}  & \ctrlzr{-1} & \qswap & \qw & \qw & \qw & \qw & \qw \\
	 	\nghost{} & \lstick{}  & \qw & \ctrlzr{-1} & \qswap & \qw & \qw & \qw  & \qw\\
	 	\nghost{} & \lstick{}  & \qw & \qw & \ctrlzr{-1} & \qswap & \qw & \qw  & \qw\\
	 	\nghost{} & \lstick{}  & \qw & \qw & \qw & \ctrlzr{-1} & \qswap & \qw & \qw \\
	 	\nghost{} & \lstick{}  & \qw & \qw & \qw & \qw & \ctrlzr{-1} & \qswap & \qw\\
	 	\nghost{} & \lstick{}  & \qw & \qw & \qw & \qw & \qw & \ctrlzr{-1}  & \qw\\
 }
   \hspace{5mm}\raisebox{-14mm}{=}\hspace{0mm}
 \Qcircuit @C=0.2em @R=0.4em @!R { 
	 	\nghost{} & \lstick{} & \ctrl{1} & \targ & \qw & \qw & \qw & \qw & \qw & \qw & \qw & \qw & \qw & \qw & \qw & \qw\\
	 	\nghost{} & \lstick{} & \targ & \ctrl{-1} & \ctrl{1} & \targ & \qw & \qw & \qw & \qw & \qw & \qw & \qw & \qw & \qw & \qw\\
	 	\nghost{} & \lstick{} & \qw & \qw & \targ & \ctrl{-1} & \ctrl{1} & \targ & \qw & \qw & \qw & \qw & \qw & \qw & \qw & \qw\\
	 	\nghost{} & \lstick{} & \qw & \qw & \qw & \qw & \targ & \ctrl{-1} & \ctrl{1} & \targ & \qw & \qw & \qw & \qw & \qw & \qw\\
	 	\nghost{} & \lstick{} & \qw & \qw & \qw & \qw & \qw & \qw & \targ & \ctrl{-1} & \ctrl{1} & \targ & \qw & \qw & \qw & \qw\\
	 	\nghost{} & \lstick{} & \qw & \qw & \qw & \qw & \qw & \qw & \qw & \qw & \targ & \ctrl{-1} & \ctrl{1} & \targ & \qw & \qw\\
	 	\nghost{} & \lstick{} & \qw & \qw & \qw & \qw & \qw & \qw & \qw & \qw & \qw & \qw & \targ & \ctrl{-1} & \qw & \qw\\
 }
  \hspace{5mm}\raisebox{-14mm}{=}\hspace{0mm}
 \Qcircuit @C=0.2em @R=0.4em @!R { 
	 	\nghost{} & \lstick{} & \ctrl{1} & \qw & \targ & \qw & \qw & \qw & \qw & \qw & \qw\\
	 	\nghost{} & \lstick{}  & \targ & \ctrl{1} & \ctrl{-1} & \targ & \qw & \qw & \qw & \qw & \qw\\
	 	\nghost{} & \lstick{}  & \qw & \targ & \ctrl{1} & \ctrl{-1} & \targ  & \qw  & \qw & \qw & \qw\\
	 	\nghost{} & \lstick{}  & \qw & \qw & \targ & \ctrl{1} & \ctrl{-1} & \targ  & \qw & \qw & \qw\\
	 	\nghost{} & \lstick{}  & \qw & \qw & \qw & \targ & \ctrl{1} & \ctrl{-1} & \targ & \qw & \qw\\
	 	\nghost{} & \lstick{}  & \qw & \qw & \qw & \qw & \targ & \ctrl{1} & \ctrl{-1} & \targ & \qw \\
	 	\nghost{} & \lstick{}  & \qw & \qw & \qw & \qw & \qw & \targ  & \qw & \ctrl{-1} & \qw\\
 }
 }
 \qcref{cnot_less_2}
\]
In the MCSU2 structure, two additional CNOT gates can be cancelled out, noting that the structure between the partial SWAP chains in \qc{swap_chain_t} must either begin or end with a $\MCO{X}{q_{k}}{q_{k-1}}$ gate. This is since $q_{k-1}$ cannot be in both $C_1$ and $C_2$.\\

A similar analysis can be used for multi-controlled multi-target gates in order to place all of the target qubits at the bottom, however, the method described in \apx{MCMTSU2_lnn} can be used instead if the number of targets is large. 
The MCU2 is implemented as an MCMTSU2 with two target qubits. In this case it is beneficial to place the targets below all control qubits, as can be realised from the cost analysis in
\sec{lnn_mcmtu2}.
 The SWAP gates can be replaced by partial SWAP gates in this case as well, by first applying the swap chains for the target qubit which is nearest to the bottom as in \qc{swap_chain_t}, and then repeat a similar process for the second qubit. This can be realised from the following example.

\[
\scalebox{0.6}{
\Qcircuit @C=1.0em @R=0.2em @!R { 
	 	\nghost{} & \lstick{} & \rvgate{-2\psi}{\hat{z}} & \qw & \qw\\
	 	\nghost{} & \lstick{} & \ctrl{-1} & \qw & \qw\\
	 	\nghost{} & \lstick{} & \rvgate{\lambda}{\hat{v}} \ar @{-} [-1,0] & \qw & \qw\\
	 	\nghost{} & \lstick{} & \ctrl{-1} & \qw & \qw\\
	 	\nghost{} & \lstick{} & \qw & \qw & \qw\\
	 	\nghost{} & \lstick{} & \ctrl{-2} & \qw & \qw\\
	 	\nghost{} & \lstick{} & \ctrl{-1} & \qw & \qw\\
	 	\nghost{} & \lstick{} & \qw & \qw & \qw\\
	 	\nghost{} & \lstick{} & \ctrl{-2} & \qw & \qw\\
 }
 \hspace{5mm}\raisebox{-23mm}{=}\hspace{0mm}
 \Qcircuit @C=0.6em @R=0em @!R { 
	 	\nghost{} & \lstick{} & \gate{\mathrm{B_{2}H}} & \qswap & \qw & \qw & \qw & \qw & \qw & \qw & \qw & \ctrl{1} & \qw & \qw & \qw & \qw & \qw & \qw & \qswap & \gate{\mathrm{H}} & \qw & \qw\\
	 	\nghost{} & \lstick{} & \qw & \qswap \qwx[-1] & \qswap & \qw & \qw & \qw & \qw & \qw & \qw & \ctrl{2} & \qw & \qw & \qw & \qw & \qw & \qswap & \qswap \qwx[-1] & \qw & \qw & \qw\\
	 	\nghost{} & \lstick{} & \gate{\mathrm{A_4}} & \qswap & \qswap \qwx[-1] & \qswap & \qw & \qw & \qw & \qw & \qw & \qw & \qw & \qw & \qw & \qw & \qswap & \qswap \qwx[-1] & \qswap & \gate{\mathrm{A_1}} & \qw & \qw\\
	 	\nghost{} & \lstick{} & \qw & \qswap \qwx[-1] & \qswap & \qswap \qwx[-1] & \qswap & \qw & \qw & \qw & \qw & \ctrl{1} & \qw & \qw & \qw & \qswap & \qswap \qwx[-1] & \qswap & \qswap \qwx[-1] & \qw & \qw & \qw\\
	 	\nghost{} & \lstick{} & \qw & \qw & \qswap \qwx[-1] & \qswap & \qswap \qwx[-1] & \qswap & \qw & \qw & \qw & \ctrl{2} & \qw & \qw & \qswap & \qswap \qwx[-1] & \qswap & \qswap \qwx[-1] & \qw & \qw & \qw & \qw\\
	 	\nghost{} & \lstick{} & \qw & \qw & \qw & \qswap \qwx[-1] & \qswap & \qswap \qwx[-1] & \qswap & \qw & \qw & \qw & \qw & \qswap & \qswap \qwx[-1] & \qswap & \qswap \qwx[-1] & \qw & \qw & \qw & \qw & \qw\\
	 	\nghost{} & \lstick{} & \qw & \qw & \qw & \qw & \qswap \qwx[-1] & \qswap & \qswap \qwx[-1] & \qswap & \qw & \ctrl{1} & \qswap & \qswap \qwx[-1] & \qswap & \qswap \qwx[-1] & \qw & \qw & \qw & \qw & \qw & \qw\\
	 	\nghost{} & \lstick{} & \qw & \qw & \qw & \qw & \qw & \qswap \qwx[-1] & \qswap & \qswap \qwx[-1] & \gate{\mathrm{B_2^\dagger}} & \rvgate{-2\psi}{\hat{x}}\ar @{-} [1,0] & \qswap \qwx[-1] & \qswap & \qswap \qwx[-1] & \qw & \qw & \qw & \qw & \qw & \qw & \qw\\
	 	\nghost{} & \lstick{} & \qw & \qw & \qw & \qw & \qw & \qw & \qswap \qwx[-1] & \qw & \qw & \rvgate{\lambda}{\hat{x}} & \gate{\mathrm{A_3^\dagger}} & \qswap \qwx[-1] & \qw & \qw & \qw & \qw & \qw & \qw & \qw & \qw\\
 }
  \hspace{5mm}\raisebox{-23mm}{=}\hspace{0mm}
 \Qcircuit @C=0.6em @R=0em @!R { 
	 	\nghost{} & \lstick{} & \gate{\mathrm{B_{2}H}} & \qswap & \qw & \qw & \qw & \qw & \qw & \qw & \qw & \ctrl{1} & \qw & \qw & \qw & \qw & \qw & \qw & \qswap & \gate{\mathrm{H}} & \qw & \qw\\
	 	\nghost{} & \lstick{} & \qw & \ctrlzr{-1} & \qswap & \qw & \qw & \qw & \qw & \qw & \qw & \ctrl{2} & \qw & \qw & \qw & \qw & \qw & \qswap & \ctrlzl{-1} & \qw & \qw & \qw\\
	 	\nghost{} & \lstick{} & \gate{\mathrm{A_4}} & \qswap & \ctrlzr{-1} & \qswap & \qw & \qw & \qw & \qw & \qw & \qw & \qw & \qw & \qw & \qw & \qswap & \ctrlzl{-1} & \qswap & \gate{\mathrm{A_1}} & \qw & \qw\\
	 	\nghost{} & \lstick{} & \qw & \ctrlzr{-1} & \qswap & \ctrlzr{-1} & \qswap & \qw & \qw & \qw & \qw & \ctrl{1} & \qw & \qw & \qw & \qswap & \ctrlzl{-1} & \qswap & \ctrlzl{-1} & \qw & \qw & \qw\\
	 	\nghost{} & \lstick{} & \qw & \qw & \ctrlzr{-1} & \qswap & \ctrlzr{-1} & \qswap & \qw & \qw & \qw & \ctrl{2} & \qw & \qw & \qswap & \ctrlzl{-1} & \qswap & \ctrlzl{-1} & \qw & \qw & \qw & \qw\\
	 	\nghost{} & \lstick{} & \qw & \qw & \qw & \ctrlzr{-1} & \qswap & \ctrlzr{-1} & \qswap & \qw & \qw & \qw & \qw & \qswap & \ctrlzl{-1} & \qswap & \ctrlzl{-1} & \qw & \qw & \qw & \qw & \qw\\
	 	\nghost{} & \lstick{} & \qw & \qw & \qw & \qw & \ctrlzr{-1} & \qswap & \ctrlzr{-1} & \qswap & \qw & \ctrl{1} & \qswap & \ctrlzl{-1} & \qswap & \ctrlzl{-1} & \qw & \qw & \qw & \qw & \qw & \qw\\
	 	\nghost{} & \lstick{} & \qw & \qw & \qw & \qw & \qw & \ctrlzr{-1} & \qswap & \ctrlzr{-1} & \gate{\mathrm{B_2^\dagger}} & \rvgate{-2\psi}{\hat{x}}\ar @{-} [1,0] & \ctrlzl{-1} & \qswap & \ctrlzl{-1}& \qw & \qw & \qw & \qw & \qw & \qw & \qw\\
	 	\nghost{} & \lstick{} & \qw & \qw & \qw & \qw & \qw & \qw & \ctrlzr{-1} & \qw & \qw & \rvgate{\lambda}{\hat{x}} & \gate{\mathrm{A_3^\dagger}} & \ctrlzl{-1} & \qw & \qw & \qw & \qw & \qw & \qw & \qw & \qw\\
 }
}
\qcref{cnot_less_3}
\]

The depth of the swap chains vary slightly, depending on the distance between the target qubits. From \qc{cnot_less_4}, it can be noted that the CNOT depth is no larger than $2(k+4)$.

\[
\scalebox{0.8}{
\Qcircuit @C=0.6em @R=0.4em @!R { 
	 	\nghost{} & \lstick{} & \qw & \qswap & \qw & \qw & \qw & \qw & \qw & \qw & \qw\\
	 	\nghost{} & \lstick{} & \qswap & \ctrlzr{-1} & \qswap & \qw & \qw & \qw & \qw & \qw & \qw\\
	 	\nghost{} & \lstick{} & \ctrlzr{-1} & \qswap & \ctrlzr{-1} & \qswap & \qw & \qw & \qw & \qw & \qw\\
	 	\nghost{} & \lstick{} & \qw & \ctrlzr{-1} & \qswap & \ctrlzr{-1} & \qswap & \qw & \qw & \qw & \qw\\
	 	\nghost{} & \lstick{} & \qw & \qw & \ctrlzr{-1} & \qswap & \ctrlzr{-1} & \qswap & \qw & \qw & \qw\\
	 	\nghost{} & \lstick{} & \qw & \qw & \qw & \ctrlzr{-1} & \qswap & \ctrlzr{-1} & \qswap & \qw & \qw\\
	 	\nghost{} & \lstick{} & \qw & \qw & \qw & \qw & \ctrlzr{-1} & \qswap & \ctrlzr{-1} & \qw & \qw\\
	 	\nghost{} & \lstick{} & \qw & \qw & \qw & \qw & \qw & \ctrlzr{-1} & \qw & \qw & \qw\\
 }
 \hspace{5mm}\raisebox{-18mm}{=}\hspace{0mm}
 \Qcircuit @C=0.2em @R=0.4em @!R { 
	 	\nghost{} & \lstick{} & \qw & \qw & \ctrl{1} & \targ & \qw & \qw & \qw & \qw & \qw & \qw & \qw & \qw & \qw & \qw & \qw & \qw\\
	 	\nghost{} & \lstick{} & \ctrl{1} & \targ & \targ & \ctrl{-1} & \ctrl{1} & \targ & \qw & \qw & \qw & \qw & \qw & \qw & \qw & \qw & \qw & \qw\\
	 	\nghost{} & \lstick{} & \targ & \ctrl{-1} & \ctrl{1} & \targ & \targ & \ctrl{-1} & \ctrl{1} & \targ & \qw & \qw & \qw & \qw & \qw & \qw & \qw & \qw\\
	 	\nghost{} & \lstick{} & \qw & \qw & \targ & \ctrl{-1} & \ctrl{1} & \targ & \targ & \ctrl{-1} & \ctrl{1} & \targ & \qw & \qw & \qw & \qw & \qw & \qw\\
	 	\nghost{} & \lstick{} & \qw & \qw & \qw & \qw & \targ & \ctrl{-1} & \ctrl{1} & \targ & \targ & \ctrl{-1} & \ctrl{1} & \targ & \qw & \qw & \qw & \qw\\
	 	\nghost{} & \lstick{} & \qw & \qw & \qw & \qw & \qw & \qw & \targ & \ctrl{-1} & \ctrl{1} & \targ & \targ & \ctrl{-1} & \ctrl{1} & \targ & \qw & \qw\\
	 	\nghost{} & \lstick{} & \qw & \qw & \qw & \qw & \qw & \qw & \qw & \qw & \targ & \ctrl{-1} & \ctrl{1} & \targ & \targ & \ctrl{-1} & \qw & \qw\\
	 	\nghost{} & \lstick{} & \qw & \qw & \qw & \qw & \qw & \qw & \qw & \qw & \qw & \qw & \targ & \ctrl{-1} & \qw & \qw & \qw & \qw\\
 }
\hspace{5mm}\raisebox{-18mm}{=}\hspace{0mm}
 \Qcircuit @C=0.2em @R=0.4em @!R { 
	 	\nghost{} & \lstick{}  & \qw & \qw & \qw & \ctrl{1} & \qw & \targ & \qw & \qw & \qw & \qw & \qw & \qw & \qw\\
	 	\nghost{} & \lstick{} & \ctrl{1} & \qw & \targ & \targ & \ctrl{1} & \ctrl{-1} & \targ & \qw & \qw & \qw & \qw & \qw & \qw\\
	 	\nghost{} & \lstick{} & \targ & \ctrl{1} & \ctrl{-1} & \targ & \targ & \ctrl{1} & \ctrl{-1} & \targ & \qw & \qw & \qw & \qw & \qw\\
	 	\nghost{} & \lstick{} & \qw & \targ & \ctrl{1} & \ctrl{-1} & \targ & \targ & \ctrl{1} & \ctrl{-1} & \targ & \qw & \qw & \qw & \qw\\
	 	\nghost{} & \lstick{} & \qw & \qw & \targ & \ctrl{1} & \ctrl{-1} & \targ & \targ & \ctrl{1} & \ctrl{-1} & \targ & \qw & \qw & \qw\\
	 	\nghost{} & \lstick{} & \qw & \qw & \qw & \targ & \ctrl{1} & \ctrl{-1} & \targ & \targ & \ctrl{1} & \ctrl{-1} & \targ & \qw & \qw\\
	 	\nghost{} & \lstick{} & \qw & \qw & \qw & \qw & \targ & \ctrl{1} & \ctrl{-1} & \targ & \targ & \qw & \ctrl{-1} & \qw & \qw\\
	 	\nghost{} & \lstick{} & \qw & \qw & \qw & \qw & \qw & \targ & \qw & \ctrl{-1} & \qw & \qw & \qw & \qw & \qw\\
 }
}
\qcref{cnot_less_4}
\]

\section{LNN specific cases}\label{apx:specific_cases}

We recall that in the ATA case, additional ancilla qubits are a useful resource, which may assist in reducing gate counts. In the LNN case, ancilla qubits may be seen as an obstacle, as these increase the value of $k$ and therefore the number of CNOT gates as shown in \sec{MCSU2_decomposition}.
However, in certain cases, the ancilla qubits can be used in order to reduce the gate count in LNN connectivity. 

As can be seen in our example in \qc{mcsu2_struct_lnn}, some of the control qubits of $C_1$ are located above all $C_2$ qubits. In this case, the gates which commute out of the $C_1$ structures to form $\Xi$ layers (the gates on the left of the barrier in \qc{rf_toffoli_LNN}) may be cancelled out, if those commute with the structure between them, which corresponds to $C_2$. The identities described in \qc{remove_qubit_1} and \qc{remove_qubit_2} may be used in some cases to increase the number of gates which can be commuted and cancelled out in this way. In addition, similarly to the ATA case, some of the gates commuting towards the edges of the full MCSU2 structure can be removed by applying gates which commute with the MCSU2 on both sides as follows.
\[
\scalebox{0.7}{
\Qcircuit @C=1.0em @R=0.2em @!R { 
	 	\nghost{} & \lstick{} & \ctrl{2} & \qw \\
	 	\nghost{} & \lstick{} & \qw & \qw & \\
	 	\nghost{} & \lstick{} & \ctrl{3} & \qw \\
	 	\nghost{} & \lstick{} & \qw & \qw \\
	 	\nghost{} & \lstick{} & \qw & \qw \\
	 	\nghost{} & \lstick{} & \ctrl{2} & \qw \\
	 	\nghost{} & \lstick{} & \qw & \qw \\
	 	\nghost{} & \lstick{} & \ctrl{1} & \qw \\
	 	\nghost{{}   } & \lstick{{}   } & \ar @{.} [1,0] & \\
	 	\nghost{{}  } & \lstick{{}  } &  &  \\
	 	\nghost{} & \lstick{} & \rvgate{\lambda}{\hat{v}}\ar @{-} [-1,0] & \qw \\
 }
 \hspace{5mm}\raisebox{-30mm}{=}\hspace{0mm}
\Qcircuit @C=1.0em @R=0.2em @!R { 
	 	\nghost{} & \lstick{} & \qw & \qw & \qw & \qw & \qw \barrier[0em]{10} & \qw & \ctrl{2} \barrier[0em]{10} & \qw & \qw & \qw & \qw & \qw & \qw & \qw \\
	 	\nghost{} & \lstick{} & \targ & \gate{\mathrm{T^\dagger}} & \targ & \gate{\mathrm{T}} & \gate{\mathrm{H}} & \qw & \qw & \qw & \gate{\mathrm{H}} & \gate{\mathrm{T^\dagger}} & \targ & \gate{\mathrm{T}} & \targ & \qw \\
	 	\nghost{} & \lstick{} & \ctrl{-1} & \qw & \ctrl{-1} & \qw & \qw & \qw & \ctrl{3} & \qw & \qw & \qw & \ctrl{-1} & \qw & \ctrl{-1} & \qw \\
	 	\nghost{} & \lstick{} & \qw & \qw & \qw & \qw & \qw & \qw & \qw & \qw & \qw & \qw & \qw & \qw & \qw & \qw \\
	 	\nghost{} & \lstick{} & \targ & \gate{\mathrm{T^\dagger}} & \targ & \gate{\mathrm{T}} & \gate{\mathrm{H}} & \qw & \qw & \qw & \gate{\mathrm{H}} & \gate{\mathrm{T^\dagger}} & \targ & \gate{\mathrm{T}} & \targ & \qw \\
	 	\nghost{} & \lstick{} & \ctrl{-1} & \qw & \ctrl{-1} & \qw & \qw & \qw & \ctrl{2} & \qw & \qw & \qw & \ctrl{-1} & \qw & \ctrl{-1} & \qw \\
	 	\nghost{} & \lstick{} & \targ & \gate{\mathrm{T^\dagger}} & \targ & \gate{\mathrm{T}} & \gate{\mathrm{H}} & \qw & \qw & \qw & \gate{\mathrm{H}} & \gate{\mathrm{T^\dagger}} & \targ & \gate{\mathrm{T}} & \targ & \qw \\
	 	\nghost{} & \lstick{} & \ctrl{-1} & \qw & \ctrl{-1} & \qw & \qw & \qw & \ctrl{1} & \qw & \qw & \qw & \ctrl{-1} & \qw & \ctrl{-1} & \qw \\
	 	\nghost{} & \lstick{} &  & & &  &  &  & \ar @{.} [1,0] &  &  &  &  &  &  & \\
	 	\nghost{} & \lstick{} &  &  &  &  &  & &  &  &  &  &  &  &  &  \\
	 	\nghost{} & \lstick{} & \qw & \qw & \qw & \qw & \qw & \qw & \rvgate{\lambda}{\hat{v}}\ar @{-} [-1,0] & \qw & \qw & \qw & \qw & \qw & \qw & \qw \\
 }
 }
 \qcref{commute_sides}
\]

This allows for the following cancellations between the added gates, and the ones commuted to the left from \qc{rf_toffoli_LNN} (similarly for the right side of the MCSU2).

\[
\scalebox{0.8}{
  \Qcircuit @C=0.5em @R=0.2em @!R { 
	 	\nghost{} & \lstick{} & \targ & \gate{\mathrm{T^\dagger}} & \targ & \gate{\mathrm{T}} & \gate{\mathrm{H}} \barrier[0em]{1} & \qw & \gate{\mathrm{H}} & \ctrl{1} & \gate{\mathrm{T^\dagger}} & \targ & \qw \\
	 	\nghost{} & \lstick{} & \ctrl{-1} & \qw & \ctrl{-1} & \qw & \qw & \qw & \qw & \targ & \gate{\mathrm{T}} & \ctrl{-1} & \qw \\
 }
   \hspace{5mm}\raisebox{-3mm}{=}\hspace{0mm}
 \Qcircuit @C=0.5em @R=0.2em @!R { 
	 	\nghost{} & \lstick{} & \targ & \gate{\mathrm{T^\dagger}} & \qswap & \qw & \qw \\
	 	\nghost{} & \lstick{} & \ctrl{-1} & \qw & \qswap \qwx[-1] & \gate{\mathrm{T}} & \qw \\
 }
   \hspace{5mm}\raisebox{-3mm}{=}\hspace{0mm}
   \Qcircuit @C=0.5em @R=0.8em @!R { 
	 	\nghost{} & \lstick{} & \ctrl{1} & \targ & \qw \\
	 	\nghost{} & \lstick{} & \targ & \ctrl{-1} & \qw \\
 }
    \hspace{5mm}\raisebox{-3mm}{=}\hspace{0mm}
   \Qcircuit @C=0.5em @R=0.8em @!R { 
	 	\nghost{} & \lstick{} & \ctrlzl{1} &  \qw \\
	 	\nghost{} & \lstick{} & \qswap &  \qw \\
 }
}
\qcref{lnn_gate_reductions}
\]
Therefore, each "useful" ancilla allows to remove $8$ T gates, $4$ CNOT and $4$ H gates. These cancellations apply when the ancilla qubit is located directly above the control qubit. However, if the ancilla was originally located directly below the control, it can be swapped as follows.
\[
\Qcircuit @C=1.0em @R=0.2em @!R { 
	 	\nghost{} & \lstick{} & \ctrl{1} & \qw & \qw\\
	 	\nghost{} & \lstick{} & \ctrl{2} & \qw & \qw\\
	 	\nghost{} & \lstick{} & \qw & \qw & \qw\\
	 	\nghost{} & \lstick{} & \ctrl{2} & \qw & \qw\\
	 	\nghost{} & \lstick{} & \qw & \qw & \qw\\
	 	\nghost{} & \lstick{} & \ctrl{1} & \qw & \qw\\
	 	\nghost{} & \lstick{} & \ctrl{3} & \qw & \qw\\
	 	\nghost{} & \lstick{} & \qw & \qw & \qw\\
	 	\nghost{} & \lstick{} & \qw & \qw & \qw\\
	 	\nghost{} & \lstick{} & \ctrl{1} & \qw & \qw\\
	 	\nghost{} & \lstick{} & \ar @{.} [1,0] &  & \\
	 	\nghost{} & \lstick{} &  & & \\
 }
  \hspace{5mm}\raisebox{-20mm}{=}\hspace{0mm}
\Qcircuit @C=1.0em @R=0.2em @!R { 
	 	\nghost{} & \lstick{} & \qw & \ctrl{1} & \qw & \qw & \qw\\
	 	\nghost{} & \lstick{} & \qw & \ctrl{2} & \qw & \qw & \qw\\
	 	\nghost{} & \lstick{} & \qw & \qw & \qw & \qw & \qw\\
	 	\nghost{} & \lstick{} & \qw & \ctrl{2} & \qw & \qw & \qw\\
	 	\nghost{} & \lstick{} & \qw & \qw & \qw & \qw & \qw\\
	 	\nghost{} & \lstick{} & \qw & \ctrl{2} & \qw & \qw & \qw\\
	 	\nghost{} & \lstick{} & \ctrlzr{0} & \qw & \ctrlzl{0} & \qw & \qw\\
	 	\nghost{} & \lstick{} & \qswap \qwx[-1] & \ctrl{2} & \qswap \qwx[-1] & \qw & \qw\\
	 	\nghost{} & \lstick{} & \qw & \qw & \qw & \qw & \qw\\
	 	\nghost{} & \lstick{} & \qw & \ctrl{1} & \qw & \qw & \qw\\
	 	\nghost{} & \lstick{} &  & \ar @{.} [1,0] &  &  & \\
	 	\nghost{} & \lstick{} &  &  & & & \\
 }
 \qcref{that_one_cir}
\]
This has two benefits. First, as can be seen, the two remaining CNOT gates in \qc{lnn_gate_reductions} can be cancelled out, and thus reducing the CNOT cost by additional $4$. Moreover, this swap may increase the number of $C_1$ qubits located above all $C_2$ ones, thus allowing for more such cancellations by making more ancilla qubits "useful", and at the same time reduce $k_2$, thus reducing the number of CNOT gate even further. 

From this arises an important question - in which cases is it beneficial to apply swap gates? As can be seen, if one wishes to minimize the number of T gates, and is willing to "pay" with a quadratic increase of CNOT gates, then many SWAP chains can be applied in order to choose the a qubit ordering which maximizes the number of useful ancilla.
Otherwise, one might compare the costs of each possible qubit reordering and select the most efficient result, considering all Clifford+T gates. As this might be expensive in terms of classical computation, we identify some cases in which it is always beneficial to reorder the qubits. 

Initially, both $k_1$ and $k_2$ may be reduced by relocating the bottom $\{C,t\}$ qubit upward towards the nearest $\{C,t\}$ one above it, using two partial SWAP chains. Similarly, $k_1$ can be reduced by relocating the top qubit downwards, unless this results in an increase of $k_2$. 

Finally, one can easily identify the first pair of neighboring control qubits which do not include $C[1]$, and the nearest ancilla qubit below this pair. As mentioned above, if the ancilla is directly below this pair, it is always beneficial to apply a swap, and iteratively continue to identify the next neighboring pair of controls. In case the ancilla is not directly below, it is beneficial to relocate it upward if the number of added CNOT gates applied for this relocation is no larger than the number of CNOT gates which this relocation allows to reduce.

We consider the specific control permutation $\Omega :=[1,1,1,0,1,0,1,..,0,1,0,0]$ such that the target qubit $t$ is at the end. After one iteration of the mentioned procedure, a SWAP($q_3,q_4$) will be applied, and the new order will be $\Omega'=[1,1,0,1,1,0,1,..,0,1,0,0]$, then SWAP($q_5,q_6$) will be applied and so on. Finally, the permutation will be $\Omega^{(n-2)}=[1,1,0,1,0,1,0,..,1,0,1,0]$.  This case assumes $k=2n-1$, i.e. $n-2$ available dirty ancilla qubits, and allows for all $[X_\Delta]$ gates to be implemented as the gates to the right of the barrier in \qc{rf_toffoli_LNN}, or their inverse. 

Furthermore, all applied SWAP gates are automatically cancelled during the process. In this case, $k_2 = n_2= 0$, and $l'_{0,1}=1$. The complete structure will require two SWAP-$[iZ]$ gates and $4(n-2)$ reduced $[X_\Delta]$ gates in addition to five $\{\rx,\rz\}$ gates. We note that since $n_2=0$, in this case the following applies for $n\geq 3$.

\begin{theorem}\label{thm:lnn_su2_row}
    Any multi-controlled $SU(2)$ gate with $n\geq 3$ controls can be implemented over $k=2n-1$ qubits, ordered as $\Omega$ or its reversed string, in LNN connectivity such that the target qubit $t$ is at the top or bottom, using five gates from $\{\rx,\rz\}$, in addition to\\
CNOT cost $12n-12$ and depth $12n-12$,\\
T ~~~~~ cost $8n-8$ ~~ and depth $4n-4$,\\
H ~~~~~ cost $4n-8$ ~~ and depth $4n-8$.
\end{theorem}
This can be considered as the best-case scenario; however, slightly different permutations may result in a reduction of some costs/depths while increasing others.\\
We cover an additional specific case in which all qubits in the circuit are used, i.e. $k=n+1$, and the target is located at the top or bottom. For this case, $C=[1,1,1,..,1,1,0]$. It is clear that $l'_{0,1}=l'_{0,2}=1$, $k_1=k$, and $k_2=k-2$. The following result can be achieved by following the steps presented in \sec{MCSU2_decomposition}.

\begin{theorem}\label{thm:lnn_su2_rowg}
    Any multi-controlled $SU(2)$ gate with $n\geq 6$ controls can be implemented over $k=n+1$ qubits in LNN connectivity such that the target qubit $t$ is at the top or bottom, using $8$ gates from $\{\rx,\rz\}$, in addition to\\
CNOT cost $20n-48$ and depth $12n$,\\
T ~~~~~ cost $16n-48$ and depth $4n$,\\
H ~~~~~ cost $8n-32$~ and depth $4n-11$.
\end{theorem}

The costs and depths of MCX and MCU2 with one dirty/clean ancilla can be obtained in a similar way. As the structures used for these gates are based on the MCSU2 implementation, it is clear that in these cases the leading terms are the same as mentioned above.

{\section{LNN Multi-controlled multi-target SU(2)}\label{apx:MCMTSU2_lnn}

We provide an implementation for any LNN restricted MCMTSU2 gate with a set $C=\{c_1,c_2,..,c_n\}$ of $n$ controls and $\tau=\{t_1,t_2,..,t_m\}$ of $m$ targets implemented over a set $Q'=\{q_1,q_2,..,q_{k'}\}$ of LNN qubits. The set $Q\in Q'$ of size $k$ is defined in this case as the smallest set of LNN qubits which satisfies $C\cup \tau \in Q$ and $q_k\not\in C$.

At first we assume that each target qubit in $\tau$ is located below all qubits in $C$, i.e. $C^j=C$ for any $q_j \in \tau$.
In this case, recalling the ATA MCMTSU2 case in \lem{mcmtsu2_struct}, the structure provided for a single target LNN MCSU2 can be directly used, along with $8(m-1)$ additional single qubit rotations from $\{\rx,\rz\}$ in order to implement the LNN MCMTSU2.
This is achieved since $\{ \overline{ Z }\}^{k_1}_{C_1,Q_1}$ and $\{ \overline{ Z }\}^{k_2}_{C_2,Q_2}$
include $\prod_{t\in\tau}\MCO{Z}{C_1}{t}$ and $\prod_{t\in\tau}\MCO{Z}{C_2}{t}$ respectively, and no other MCZ gate is applied on qubits from $\tau$.  

However, in order to guarantee that all target qubits are at the bottom of the circuit, many SWAP chains must be applied and in a random qubit permutation, this may result in a quadratic overhead of $O(mk)$ CNOT gates.
In order to maintain a linear ($O(k)$) cost  for each basic gate type, we provide a method which can be applied for any permutation in which $q_k\not \in C$. In case $q_k\in C$, two partial SWAP chains can be used in order to relocate one qubit as before. Since in this case, any non-control qubit can be chosen, the maximal distance from the nearest edge is $\floor{\frac{n}{2}}$ qubits. The CNOT cost overhead in the worst case is $4\floor{\frac{n}{2}}$ in depth $2\floor{\frac{n}{2}}+4$, as shown in \apx{communication_overhead}.

The macro structure in this case is similar to \qc{mcmtsu2_base}, with $[\Delta]$ gates applied on subsets of $Q\setminus \tau$, and the subsets $C_1$ and $C_2$ are defined similarly to the LNN MCSU2 case with one additional condition -
for each of the subsets $C_1,C_2$ it must hold that any consecutive subset of target qubits from $\tau$ may not have two neighboring control qubits.

A similar procedure to the one defined for MCSU2 can be used while treating the qubit above and below consecutive targets as nearest neighbors. In addition, in order to simplify this case, we are not allowing for any nearest neighboring controls in either subset, which in turn removes the $[iZ]$ gate from this structure.

Starting with $C_1=\emptyset$, $C_2=\emptyset$, the iteration index $l=1$, and an additional index $l'=1$. At each iteration, $q_l$ is added to $C_1$ if
\[
q_l\in C\land q_{l'}\not\in C_1.
\]
Such that the qubit $q_{l'}$ holds the previous non-target qubit if $l>1$.
If this condition is not met and $q_l\in C$, the qubit $q_l$ is added to $C_2$. If $q_l\not\in \tau$, the index $l'$ is then set to the value of $l$. Finally, $l$ is increased by 1 for the next iteration, stopping when $C_1 \cup C_2$ holds all of the control qubits in $C$, i.e. $n_1+n_2=n$. 

 The following circuit provides an example of the macro structure achieved for an arbitrary choice of $C,\tau$ and $Q$.

\[
\scalebox{0.65}{
\Qcircuit @C=1.0em @R=0.4em @!R { 
	 	\nghost{{q}_{1} :  } & \lstick{{q}_{1} :  } & \ctrl{2} & \qw & \qw\\
	 	\nghost{{q}_{2} :  } & \lstick{{q}_{2} :  } & \qw & \qw & \qw\\
	 	\nghost{{q}_{3} :  } & \lstick{{q}_{3} :  } & \ctrl{1} & \qw & \qw\\
	 	\nghost{{q}_{4} :  } & \lstick{{q}_{4} :  } & \gate{\mathrm{W_1}}\ar @{-} [2,0] & \qw & \qw\\
	 	\nghost{{q}_{5} :  } & \lstick{{q}_{5} :  } & \qw & \qw & \qw\\
	 	\nghost{{q}_{6} :  } & \lstick{{q}_{6} :  } & \ctrl{2} & \qw & \qw\\
	 	\nghost{{q}_{7} :  } & \lstick{{q}_{7} :  } & \qw & \qw & \qw\\
	 	\nghost{{q}_{8} :  } & \lstick{{q}_{8} :  } & \ctrl{1} & \qw & \qw\\
	 	\nghost{{q}_{9} :  } & \lstick{{q}_{9} :  } & \ctrl{1} & \qw & \qw\\
	 	\nghost{{q}_{10} :  } & \lstick{{q}_{10} :  } & \gate{\mathrm{W_2}}\ar @{-} [1,0] & \qw & \qw\\
	 	\nghost{{q}_{11} :  } & \lstick{{q}_{11} :  } & \ctrl{2} & \qw & \qw\\
	 	\nghost{{q}_{12} :  } & \lstick{{q}_{12} :  } & \qw & \qw & \qw\\
	 	\nghost{{q}_{13} :  } & \lstick{{q}_{13} :  } & \ctrl{1} & \qw & \qw\\
	 	\nghost{{q}_{14} :  } & \lstick{{q}_{14} :  } & \ctrl{1} & \qw & \qw\\
	 	\nghost{{q}_{15} :  } & \lstick{{q}_{15} :  } & \ctrl{2} & \qw & \qw\\
	 	\nghost{{q}_{16} :  } & \lstick{{q}_{16} :  } & \qw & \qw & \qw\\
	 	\nghost{{q}_{17} :  } & \lstick{{q}_{17} :  } & \ctrl{1} & \qw & \qw\\
	 	\nghost{{q}_{18} :  } & \lstick{{q}_{18} :  } & \ctrl{1} & \qw & \qw\\
	 	\nghost{{q}_{19} :  } & \lstick{{q}_{19} :  } & \gate{\mathrm{W_3}}\ar @{-} [2,0] & \qw & \qw\\
	 	\nghost{{q}_{20} :  } & \lstick{{q}_{20} :  } & \qw & \qw & \qw\\
	 	\nghost{{q}_{21} :  } & \lstick{{q}_{21} :  } & \gate{\mathrm{W_4}} & \qw & \qw\\
 }
 \hspace{5mm}\raisebox{-70mm}{=}\hspace{0mm}
\Qcircuit @C=1.0em @R=0.em @!R { 
	 	\nghost{{q}_{1} :  } & \lstick{{q}_{1} :  } & \qw & \ctrl{2} & \multigate{2}{\mathrm{\Delta_1}} & \qw & \qw & \qw & \qw & \ctrl{2} & \multigate{2}{\mathrm{\Delta_1^\dagger}} & \qw & \qw & \qw & \qw & \qw & \qw\\
	 	\nghost{{q}_{2} :  } & \lstick{{q}_{2} :  } & \qw & \qw & \ghost{\mathrm{\Delta_1}} & \qw & \qw & \qw & \qw & \qw & \ghost{\mathrm{\Delta_1}} & \qw & \qw & \qw & \qw & \qw & \qw\\
	 	\nghost{{q}_{3} :  } & \lstick{{q}_{3} :  } & \qw & \ctrl{1} & \ghost{\mathrm{\Delta_1}}\ar @{-} [2,0] & \qw & \qw & \qw & \qw & \ctrl{1} & \ghost{\mathrm{\Delta_1}}\ar @{-} [2,0] & \qw & \qw & \qw & \qw & \qw & \qw\\
	 	\nghost{{q}_{4} :  } & \lstick{{q}_{4} :  } & \gate{\mathrm{A_4^1}} & \gate{\mathrm{Z}}\ar @{-} [2,0] & \qw & \gate{\mathrm{A_2^1}} & \gate{\mathrm{Z}}\ar @{-} [5,0] & \qw & \gate{\mathrm{A_3^1}} & \gate{\mathrm{Z}}\ar @{-} [2,0] & \qw & \gate{\mathrm{A_2^1}} & \gate{\mathrm{Z}}\ar @{-} [5,0] & \qw & \gate{\mathrm{A_1^1}} & \qw & \qw\\
	 	\nghost{{q}_{5} :  } & \lstick{{q}_{5} :  } & \qw & \qw & \multigate{4}{\mathrm{\Delta_1}} & \qw & \qw & \multigate{4}{\mathrm{\Delta_2}} & \qw & \qw & \multigate{4}{\mathrm{\Delta_1^\dagger}} & \qw & \qw & \multigate{4}{\mathrm{\Delta_2^\dagger}} & \qw & \qw & \qw\\
	 	\nghost{{q}_{6} :  } & \lstick{{q}_{6} :  } & \qw & \ctrl{2} & \ghost{\mathrm{\Delta_1}} & \qw & \qw & \ghost{\mathrm{\Delta_1}} & \qw & \ctrl{2} & \ghost{\mathrm{\Delta_1}}& \qw & \qw & \ghost{\mathrm{\Delta_1}} & \qw & \qw & \qw\\
	 	\nghost{{q}_{7} :  } & \lstick{{q}_{7} :  } & \qw & \qw & \ghost{\mathrm{\Delta_1}} & \qw & \qw & \ghost{\mathrm{\Delta_1}} & \qw & \qw & \ghost{\mathrm{\Delta_1}} & \qw & \qw & \ghost{\mathrm{\Delta_1}} & \qw & \qw & \qw\\
	 	\nghost{{q}_{8} :  } & \lstick{{q}_{8} :  } & \qw & \ctrl{2} & \ghost{\mathrm{\Delta_1}} & \qw & \qw & \ghost{\mathrm{\Delta_1}} & \qw & \ctrl{2} & \ghost{\mathrm{\Delta_1}} & \qw & \qw & \ghost{\mathrm{\Delta_1}} & \qw & \qw & \qw\\
	 	\nghost{{q}_{9} :  } & \lstick{{q}_{9} :  } & \qw & \qw & \ghost{\mathrm{\Delta_1}}\ar @{-} [2,0] & \qw & \ctrl{1} & \ghost{\mathrm{\Delta_1}}\ar @{-} [2,0] & \qw & \qw & \ghost{\mathrm{\Delta_1}}\ar @{-} [2,0] & \qw & \ctrl{1} & \ghost{\mathrm{\Delta_1}}\ar @{-} [2,0] & \qw & \qw & \qw\\
	 	\nghost{{q}_{10} :  } & \lstick{{q}_{10} :  } & \gate{\mathrm{A_4^2}} & \gate{\mathrm{Z}}\ar @{-} [1,0] & \qw & \gate{\mathrm{A_2^2}} & \gate{\mathrm{Z}}\ar @{-} [4,0] & \qw & \gate{\mathrm{A_3^2}} & \gate{\mathrm{Z}}\ar @{-} [1,0] & \qw & \gate{\mathrm{A_2^2}} & \gate{\mathrm{Z}}\ar @{-} [4,0] & \qw & \gate{\mathrm{A_1^2}} & \qw & \qw\\
	 	\nghost{{q}_{11} :  } & \lstick{{q}_{11} :  } & \qw & \ctrl{2} & \multigate{7}{\mathrm{\Delta_1}} & \qw & \qw & \multigate{7}{\mathrm{\Delta_2}} & \qw & \ctrl{2} & \multigate{7}{\mathrm{\Delta_1^\dagger}} & \qw & \qw & \multigate{7}{\mathrm{\Delta_2^\dagger}} & \qw & \qw & \qw\\
	 	\nghost{{q}_{12} :  } & \lstick{{q}_{12} :  } & \qw & \qw & \ghost{\mathrm{\Delta_1}} & \qw & \qw & \ghost{\mathrm{\Delta_1}} & \qw & \qw & \ghost{\mathrm{\Delta_1}} & \qw & \qw & \ghost{\mathrm{\Delta_1}} & \qw & \qw & \qw\\
	 	\nghost{{q}_{13} :  } & \lstick{{q}_{13} :  } & \qw & \ctrl{2} & \ghost{\mathrm{\Delta_1}} & \qw & \qw & \ghost{\mathrm{\Delta_1}} & \qw & \ctrl{2} & \ghost{\mathrm{\Delta_1}} & \qw & \qw & \ghost{\mathrm{\Delta_1}} & \qw & \qw & \qw\\
	 	\nghost{{q}_{14} :  } & \lstick{{q}_{14} :  } & \qw & \qw & \ghost{\mathrm{\Delta_1}} & \qw & \ctrl{4} & \ghost{\mathrm{\Delta_1}} & \qw & \qw & \ghost{\mathrm{\Delta_1}} & \qw & \ctrl{4} & \ghost{\mathrm{\Delta_1}} & \qw & \qw & \qw\\
	 	\nghost{{q}_{15} :  } & \lstick{{q}_{15} :  } & \qw & \ctrl{2} & \ghost{\mathrm{\Delta_1}} & \qw & \qw & \ghost{\mathrm{\Delta_1}} & \qw & \ctrl{2} & \ghost{\mathrm{\Delta_1}} & \qw & \qw & \ghost{\mathrm{\Delta_1}} & \qw & \qw & \qw\\
	 	\nghost{{q}_{16} :  } & \lstick{{q}_{16} :  } & \qw & \qw & \ghost{\mathrm{\Delta_1}} & \qw & \qw & \ghost{\mathrm{\Delta_1}} & \qw & \qw & \ghost{\mathrm{\Delta_1}} & \qw & \qw & \ghost{\mathrm{\Delta_1}} & \qw & \qw & \qw\\
	 	\nghost{{q}_{17} :  } & \lstick{{q}_{17} :  } & \qw & \ctrl{2} & \ghost{\mathrm{\Delta_1}} & \qw & \qw & \ghost{\mathrm{\Delta_1}} & \qw & \ctrl{2} & \ghost{\mathrm{\Delta_1}} & \qw & \qw & \ghost{\mathrm{\Delta_1}} & \qw & \qw & \qw\\
	 	\nghost{{q}_{18} :  } & \lstick{{q}_{18} :  } & \qw & \qw & \ghost{\mathrm{\Delta_1}}\ar @{-} [2,0] & \qw & \ctrl{1} & \ghost{\mathrm{\Delta_1}}\ar @{-} [2,0] & \qw & \qw & \ghost{\mathrm{\Delta_1}}\ar @{-} [2,0] & \qw & \ctrl{1} & \ghost{\mathrm{\Delta_1}}\ar @{-} [2,0] & \qw & \qw & \qw\\
	 	\nghost{{q}_{19} :  } & \lstick{{q}_{19} :  } & \gate{\mathrm{A_4^3}} & \gate{\mathrm{Z}}\ar @{-} [2,0] & \qw & \gate{\mathrm{A_2^3}} & \gate{\mathrm{Z}}\ar @{-} [2,0] & \qw & \gate{\mathrm{A_3^3}} & \gate{\mathrm{Z}}\ar @{-} [2,0] & \qw & \gate{\mathrm{A_2^3}} & \gate{\mathrm{Z}}\ar @{-} [2,0] & \qw & \gate{\mathrm{A_1^3}} & \qw & \qw\\
	 	\nghost{{q}_{20} :  } & \lstick{{q}_{20} :  } & \qw & \qw & \gate{\mathrm{\Delta_1}} & \qw & \qw & \gate{\mathrm{\Delta_2}} & \qw & \qw & \gate{\mathrm{\Delta_1^\dagger}} & \qw & \qw & \gate{\mathrm{\Delta_2^\dagger}} & \qw & \qw & \qw\\
	 	\nghost{{q}_{21} :  } & \lstick{{q}_{21} :  } & \gate{\mathrm{A_4^4}} & \gate{\mathrm{Z}} & \qw & \gate{\mathrm{A_2^4}} & \gate{\mathrm{Z}} & \qw & \gate{\mathrm{A_3^4}} & \gate{\mathrm{Z}} & \qw & \gate{\mathrm{A_2^4}} & \gate{\mathrm{Z}} & \qw & \gate{\mathrm{A_1^4}} & \qw & \qw\\
 }
 }
 \qcref{LNN_MCMTSU2}
\]

We now wish to decompose the $\sqO{\Delta_1}{Q_1 \setminus \tau}\prod_{j=1}^{m}\MCO{Z}{C_1}{t_j}$ and $\sqO{\Delta_2}{Q_2 \setminus \tau}\prod_{j=1}^{m}\MCO{Z}{C_2}{t_j}$ gates. We use the $\{\overline{ Z }\}^{k_1}_{C_1,Q_1}$ and $\{\overline{ Z }\}^{k_2}_{C_2,Q_2}$ decompositions as in \lem{Vchain_fullstruct_lnn}, and apply CNOT transformations in two steps. 

First we remove any $\MCO{Z}{C_1^j}{q_j}$ and $\MCO{Z}{C_2^j}{q_j}$ for $q_j\in\tau$ if $C_1^j\not =C_1$ or $C_2^j\not =C_2$ respectively. This can be achieved due to our choice of $C_1$ and $C_2$ such that sequential target qubits cannot have more than one neighboring control qubit in either control subset.

In case the qubit below the target sequence {\em is not} a control, we can use the following.
\[
\scalebox{1.0}{
\Qcircuit @C=1.0em @R=0.2em @!R { 
	 	\nghost{{C} :  } & \lstick{{C} :  } & \qw{/^n} & \qw \barrier[0em]{5} & \qw & \ctrl{2} & \ctrl{2} & \ctrl{2} & \ctrl{1} \barrier[0em]{5} & \qw & \qw & \qw & \qw & \qw\\
	 	\nghost{{a} :  } & \lstick{{a} :  } & \qw & \qw & \qw & \qw & \qw & \qw & \control\qw & \qw & \qw & \qw & \qw & \qw\\
	 	\nghost{{c} :  } & \lstick{{c} :  } & \qw & \qw & \qw & \ctrl{3} & \ctrl{2} & \ctrl{1} & \qw & \qw & \qw & \qw & \qw & \qw\\
	 	\nghost{{t} :  } & \lstick{{t} :  } & \qw & \ctrl{1} & \qw & \qw & \qw & \control\qw & \qw & \qw & \ctrl{1} & \qw & \qw & \qw\\
	 	\nghost{{t} :  } & \lstick{{t} :  } & \ctrl{1} & \targ & \qw & \qw & \control\qw & \qw & \qw & \qw & \targ & \ctrl{1} & \qw & \qw\\
	 	\nghost{{a} :  } & \lstick{{a} :  } & \targ & \qw & \qw & \control\qw & \qw & \qw & \qw & \qw & \qw & \targ & \qw & \qw\\
 }
 \hspace{5mm}\raisebox{-10mm}{=}\hspace{0mm}
 \Qcircuit @C=1.0em @R=0.2em @!R { 
	 	\nghost{{C} :  } & \lstick{{C} :  } & \qw{/^n} & \ctrl{2} & \ctrl{1} & \qw & \qw\\
	 	\nghost{{a} :  } & \lstick{{a} :  } & \qw & \qw & \control\qw & \qw & \qw\\
	 	\nghost{{c} :  } & \lstick{{c} :  } & \qw & \ctrl{3} & \qw & \qw & \qw\\
	 	\nghost{{t} :  } & \lstick{{t} :  } & \qw & \qw & \qw & \qw & \qw\\
	 	\nghost{{t} :  } & \lstick{{t} :  } & \qw & \qw & \qw & \qw & \qw\\
	 	\nghost{{a} :  } & \lstick{{a} :  } & \qw & \control\qw & \qw & \qw & \qw\\
 }
 }
 \qcref{remove_qubit_1}
\]

And if the qubit below {\em is} a control (and the qubit above is not a control), we use the following.

\[
\scalebox{1.0}{
\Qcircuit @C=1.0em @R=0.2em @!R { 
	 	\nghost{{C} :  } & \lstick{{C} :  } & \qw{/^n} & \qw \barrier[0em]{5} & \qw & \ctrl{4} & \ctrl{3} & \ctrl{2} & \ctrl{1} \barrier[0em]{5} & \qw & \qw & \qw & \qw & \qw\\
	 	\nghost{{a} :  } & \lstick{{a} :  } & \targ & \qw & \qw & \qw & \qw & \qw & \control\qw & \qw & \qw & \targ & \qw & \qw\\
	 	\nghost{{t} :  } & \lstick{{t} :  } & \ctrl{-1} & \targ & \qw & \qw & \qw & \control\qw & \qw & \qw & \targ & \ctrl{-1} & \qw & \qw\\
	 	\nghost{{t} :  } & \lstick{{t} :  } & \qw & \ctrl{-1} & \qw & \qw & \control\qw & \qw & \qw & \qw & \ctrl{-1} & \qw & \qw & \qw\\
	 	\nghost{{c} :  } & \lstick{{c} :  } & \qw & \qw & \qw & \ctrl{1} & \qw & \qw & \qw & \qw & \qw & \qw & \qw & \qw\\
	 	\nghost{{a} :  } & \lstick{{a} :  } & \qw & \qw & \qw & \control\qw & \qw & \qw & \qw & \qw & \qw & \qw & \qw & \qw\\
 }
 \hspace{5mm}\raisebox{-10mm}{=}\hspace{0mm}
 \Qcircuit @C=1.0em @R=0.2em @!R { 
	 	\nghost{{C} :  } & \lstick{{C} :  } & \qw{/^n} & \ctrl{4} & \ctrl{1} & \qw & \qw\\
	 	\nghost{{a} :  } & \lstick{{a} :  } & \qw & \qw & \control\qw & \qw & \qw\\
	 	\nghost{{t} :  } & \lstick{{t} :  } & \qw & \qw & \qw & \qw & \qw\\
	 	\nghost{{t} :  } & \lstick{{t} :  } & \qw & \qw & \qw & \qw & \qw\\
	 	\nghost{{c} :  } & \lstick{{c} :  } & \qw & \ctrl{1} & \qw & \qw & \qw\\
	 	\nghost{{a} :  } & \lstick{{a} :  } & \qw & \control\qw & \qw & \qw & \qw\\
 }
 }
 \qcref{remove_qubit_2}
\]

In our example, $\{\overline{ Z }\}^{k_1}_{C_1,Q_1} = \sqO{\Delta'_1}{Q_1\setminus t_m}\MCO{Z}{C_1}{t_m}$ is converted to $\sqO{\Delta_1}{Q_1 \setminus \tau}\MCO{Z}{C_1}{t_m}\MCO{Z}{C_1}{t_{m-1}}$ as follows.

\[
\scalebox{0.65}{
\Qcircuit @C=1.0em @R=0.em @!R { 
	 	\nghost{{q}_{1} :  } & \lstick{{q}_{1} :  } & \ctrl{2} & \multigate{2}{\mathrm{\Delta_1}} & \qw\\
	 	\nghost{{q}_{2} :  } & \lstick{{q}_{2} :  } & \qw & \ghost{\mathrm{\Delta_1}} & \qw\\
	 	\nghost{{q}_{3} :  } & \lstick{{q}_{3} :  } & \ctrl{3} & \ghost{\mathrm{\Delta_1}}\ar @{-} [2,0]& \qw\\
	 	\nghost{{q}_{4} :  } & \lstick{{q}_{4} :  } & \qw & \qw & \qw\\
	 	\nghost{{q}_{5} :  } & \lstick{{q}_{5} :  } & \qw & \multigate{4}{\mathrm{\Delta_1}} & \qw\\
	 	\nghost{{q}_{6} :  } & \lstick{{q}_{6} :  } & \ctrl{2} & \ghost{\mathrm{\Delta_1}} & \qw\\
	 	\nghost{{q}_{7} :  } & \lstick{{q}_{7} :  } & \qw & \ghost{\mathrm{\Delta_1}} & \qw\\
	 	\nghost{{q}_{8} :  } & \lstick{{q}_{8} :  } & \ctrl{3} & \ghost{\mathrm{\Delta_1}} & \qw\\
	 	\nghost{{q}_{9} :  } & \lstick{{q}_{9} :  } & \qw & \ghost{\mathrm{\Delta_1}}\ar @{-} [2,0] & \qw\\
	 	\nghost{{q}_{10} :  } & \lstick{{q}_{10} :  } & \qw & \qw & \qw\\
	 	\nghost{{q}_{11} :  } & \lstick{{q}_{11} :  } & \ctrl{2} & \multigate{7}{\mathrm{\Delta_1}} & \qw\\
	 	\nghost{{q}_{12} :  } & \lstick{{q}_{12} :  } & \qw & \ghost{\mathrm{\Delta_1}} & \qw\\
	 	\nghost{{q}_{13} :  } & \lstick{{q}_{13} :  } & \ctrl{2} & \ghost{\mathrm{\Delta_1}} & \qw\\
	 	\nghost{{q}_{14} :  } & \lstick{{q}_{14} :  } & \qw & \ghost{\mathrm{\Delta_1}}& \qw\\
	 	\nghost{{q}_{15} :  } & \lstick{{q}_{15} :  } & \ctrl{2} & \ghost{\mathrm{\Delta_1}} & \qw\\
	 	\nghost{{q}_{16} :  } & \lstick{{q}_{16} :  } & \qw & \ghost{\mathrm{\Delta_1}} & \qw\\
	 	\nghost{{q}_{17} :  } & \lstick{{q}_{17} :  } & \ctrl{2} & \ghost{\mathrm{\Delta_1}} & \qw\\
	 	\nghost{{q}_{18} :  } & \lstick{{q}_{18} :  } & \qw & \ghost{\mathrm{\Delta_1}}\ar @{-} [2,0] & \qw\\
	 	\nghost{{q}_{19} :  } & \lstick{{q}_{19} :  } & \gate{\mathrm{Z}}\ar @{-} [2,0] & \qw & \qw\\
	 	\nghost{{q}_{20} :  } & \lstick{{q}_{20} :  } & \qw & \gate{\mathrm{\mathrm{\Delta_1}}} & \qw\\
	 	\nghost{{q}_{21} :  } & \lstick{{q}_{21} :  } & \gate{\mathrm{Z}} & \qw & \qw\\
 }
  \hspace{5mm}\raisebox{-50mm}{=}\hspace{0mm}
 \Qcircuit @C=0.5em @R=0.15em @!R { 
	 	\nghost{{q}_{1} :  } & \lstick{{q}_{1} :  }  & \ctrl{2} \barrier[0em]{20} & \qw & \ctrl{2} & \ctrl{2} & \ctrl{2} & \ctrl{2} & \ctrl{2} & \ctrl{2} & \ctrl{2} & \ctrl{2} & \ctrl{1}  & \qw\\
	 	\nghost{{q}_{2} :  } & \lstick{{q}_{2} :  }  & \qw & \qw & \qw & \qw & \qw & \qw & \qw & \qw & \qw & \qw & \control\qw  & \qw\\
	 	\nghost{{q}_{3} :  } & \lstick{{q}_{3} :  }  & \ctrl{3} & \qw & \ctrl{3} & \ctrl{3} & \ctrl{3} & \ctrl{3} & \ctrl{3} & \ctrl{3} & \ctrl{3} & \ctrl{2} & \qw & \qw\\
	 	\nghost{{q}_{4} :  } & \lstick{{q}_{4} :  } & \qw & \qw & \qw & \qw & \qw & \qw & \qw & \qw & \qw & \qw & \qw  & \qw\\
	 	\nghost{{q}_{5} :  } & \lstick{{q}_{5} :  }  & \qw & \qw & \qw & \qw & \qw & \qw & \qw & \qw & \qw & \control\qw & \qw & \qw\\
	 	\nghost{{q}_{6} :  } & \lstick{{q}_{6} :  }  & \ctrl{2} & \qw & \ctrl{2} & \ctrl{2} & \ctrl{2} & \ctrl{2} & \ctrl{2} & \ctrl{2} & \ctrl{1} & \qw & \qw  & \qw\\
	 	\nghost{{q}_{7} :  } & \lstick{{q}_{7} :  }  & \qw & \qw & \qw & \qw & \qw & \qw & \qw & \qw & \control\qw & \qw & \qw  & \qw\\
	 	\nghost{{q}_{8} :  } & \lstick{{q}_{8} :  }  & \ctrl{3} & \qw & \ctrl{3} & \ctrl{3} & \ctrl{3} & \ctrl{3} & \ctrl{3} & \ctrl{1} & \qw & \qw & \qw  & \qw\\
	 	\nghost{{q}_{9} :  } & \lstick{{q}_{9} :  }  & \qw & \qw & \qw & \qw & \qw & \qw & \qw & \control\qw & \qw & \qw & \qw  & \qw\\
	 	\nghost{{q}_{10} :  } & \lstick{{q}_{10} :  }  & \qw & \qw & \qw & \qw & \qw & \qw & \qw & \qw & \qw & \qw & \qw & \qw\\
	 	\nghost{{q}_{11} :  } & \lstick{{q}_{11} :  }  & \ctrl{2} & \qw & \ctrl{2} & \ctrl{2} & \ctrl{2} & \ctrl{2} & \ctrl{1} & \qw & \qw & \qw & \qw  & \qw\\
	 	\nghost{{q}_{12} :  } & \lstick{{q}_{12} :  }  & \qw & \qw & \qw & \qw & \qw & \qw & \control\qw & \qw & \qw & \qw & \qw  & \qw\\
	 	\nghost{{q}_{13} :  } & \lstick{{q}_{13} :  }  & \ctrl{2} & \qw & \ctrl{2} & \ctrl{2} & \ctrl{2} & \ctrl{1} & \qw & \qw & \qw & \qw & \qw  & \qw\\
	 	\nghost{{q}_{14} :  } & \lstick{{q}_{14} :  }  & \qw & \qw & \qw & \qw & \qw & \control\qw & \qw & \qw & \qw & \qw & \qw  & \qw\\
	 	\nghost{{q}_{15} :  } & \lstick{{q}_{15} :  }  & \ctrl{2} & \qw & \ctrl{2} & \ctrl{2} & \ctrl{1} & \qw & \qw & \qw & \qw & \qw & \qw  & \qw\\
	 	\nghost{{q}_{16} :  } & \lstick{{q}_{16} :  }  & \qw & \qw & \qw & \qw & \control\qw & \qw & \qw & \qw & \qw & \qw & \qw & \qw & \qw\\
	 	\nghost{{q}_{17} :  } & \lstick{{q}_{17} :  }  & \ctrl{2} & \qw & \ctrl{3} & \ctrl{1} & \qw & \qw & \qw & \qw & \qw & \qw & \qw  & \qw\\
	 	\nghost{{q}_{18} :  } & \lstick{{q}_{18} :  }  & \qw & \qw & \qw & \control\qw & \qw & \qw & \qw & \qw & \qw & \qw & \qw  & \qw\\
	 	\nghost{{q}_{19} :  } & \lstick{{q}_{19} :  }  & \gate{\mathrm{Z}}\ar @{-} [2,0] & \qw & \qw & \qw & \qw & \qw & \qw & \qw & \qw & \qw & \qw & \qw\\
	 	\nghost{{q}_{20} :  } & \lstick{{q}_{20} :  }  & \qw & \qw & \control\qw & \qw & \qw & \qw & \qw & \qw & \qw & \qw & \qw  & \qw\\
	 	\nghost{{q}_{21} :  } & \lstick{{q}_{21} :  }  & \gate{\mathrm{Z}} & \qw & \qw & \qw & \qw & \qw & \qw & \qw & \qw & \qw & \qw  & \qw\\
 }
 \hspace{5mm}\raisebox{-50mm}{=}\hspace{0mm}
\Qcircuit @C=1.0em @R=0.53em @!R { 
	 	\nghost{{q}_{1} :  } & \lstick{{q}_{1} :  } & \qw & \ctrl{2} & \multigate{19}{\mathrm{\Delta'_1}} & \qw & \qw \\
	 	\nghost{{q}_{2} :  } & \lstick{{q}_{2} :  } & \qw & \qw & \ghost{\mathrm{\Delta_1}} & \qw & \qw \\
	 	\nghost{{q}_{3} :  } & \lstick{{q}_{3} :  } & \qw & \ctrl{3} & \ghost{\mathrm{\Delta_1}} & \qw & \qw \\
	 	\nghost{{q}_{4} :  } & \lstick{{q}_{4} :  } & \ctrl{1} & \qw & \ghost{\mathrm{\Delta_1}} & \ctrl{1} & \qw \\
	 	\nghost{{q}_{5} :  } & \lstick{{q}_{5} :  } & \targ & \qw & \ghost{\mathrm{\Delta_1}} & \targ & \qw \\
	 	\nghost{{q}_{6} :  } & \lstick{{q}_{6} :  } & \qw & \ctrl{2} & \ghost{\mathrm{\Delta_1}} & \qw & \qw \\
	 	\nghost{{q}_{7} :  } & \lstick{{q}_{7} :  } & \qw & \qw & \ghost{\mathrm{\Delta_1}} & \qw & \qw \\
	 	\nghost{{q}_{8} :  } & \lstick{{q}_{8} :  } & \qw & \ctrl{3} & \ghost{\mathrm{\Delta_1}} & \qw & \qw \\
	 	\nghost{{q}_{9} :  } & \lstick{{q}_{9} :  } & \targ & \qw & \ghost{\mathrm{\Delta_1}} & \targ & \qw \\
	 	\nghost{{q}_{10} :  } & \lstick{{q}_{10} :  } & \ctrl{-1} & \qw & \ghost{\mathrm{\Delta_1}} & \ctrl{-1} & \qw \\
	 	\nghost{{q}_{11} :  } & \lstick{{q}_{11} :  } & \qw & \ctrl{2} & \ghost{\mathrm{\Delta_1}} & \qw & \qw \\
	 	\nghost{{q}_{12} :  } & \lstick{{q}_{12} :  } & \qw & \qw & \ghost{\mathrm{\Delta_1}} & \qw & \qw \\
	 	\nghost{{q}_{13} :  } & \lstick{{q}_{13} :  } & \qw & \ctrl{2} & \ghost{\mathrm{\Delta_1}} & \qw & \qw \\
	 	\nghost{{q}_{14} :  } & \lstick{{q}_{14} :  } & \qw & \qw & \ghost{\mathrm{\Delta_1}}& \qw & \qw \\
	 	\nghost{{q}_{15} :  } & \lstick{{q}_{15} :  } & \qw & \ctrl{2} & \ghost{\mathrm{\Delta_1}} & \qw & \qw \\
	 	\nghost{{q}_{16} :  } & \lstick{{q}_{16} :  } & \qw & \qw & \ghost{\mathrm{\Delta_1}} & \qw & \qw \\
	 	\nghost{{q}_{17} :  } & \lstick{{q}_{17} :  } & \qw & \ctrl{4} & \ghost{\mathrm{\Delta_1}} & \qw & \qw \\
	 	\nghost{{q}_{18} :  } & \lstick{{q}_{18} :  } & \qw & \qw & \ghost{\mathrm{\Delta_1}} & \qw & \qw \\
	 	\nghost{{q}_{19} :  } & \lstick{{q}_{19} :  } & \qw & \qw & \ghost{\mathrm{\Delta_1}} & \qw & \qw \\
	 	\nghost{{q}_{20} :  } & \lstick{{q}_{20} :  } & \qw & \qw & \ghost{\mathrm{\Delta_1}} & \qw & \qw \\
	 	\nghost{{q}_{21} :  } & \lstick{{q}_{21} :  } & \qw & \ctrl{0} & \qw & \qw & \qw \\
 }
 }
 \qcref{mcz_LNN_remove_t}
\]
This process adds two CNOT gates per target which we "disconnect", and therefore increases the CNOT cost and depth by no more then $2m$.

Once these target qubits are not affected by any gate, we wish to add $\MCO{Z}{C_1}{q_j}$ or $\MCO{Z}{C_2}{q_j}$ gates for each $q_j\in\tau$ qubit which is located above the bottom control qubit. 
Using the identity described in \qc{add_mcz_new_target}, a pair of CNOT gates can be applied on both sides of the structure in order to "add a target" to the MCMTZ gate. For each added target, we utilize the nearest qubit on which an MCZ gates controlled by the entire control subset is already applied. The first pair of CNOT gates will therefore be applied on the qubit neighboring below the bottom control, and the nearest target qubit above it. The second pair will be applied on this added target and the nearest one above it and so on. 
The following demonstrates this on the structure controlled by $C_1$ in our example.

\[
\scalebox{0.65}{
\Qcircuit @C=1.0em @R=0.em @!R { 
	 	\nghost{{q}_{1} :  } & \lstick{{q}_{1} :  } & \ctrl{2} & \multigate{2}{\mathrm{\Delta_1}} & \qw\\
	 	\nghost{{q}_{2} :  } & \lstick{{q}_{2} :  } & \qw & \ghost{\mathrm{\Delta_1}} & \qw\\
	 	\nghost{{q}_{3} :  } & \lstick{{q}_{3} :  } & \ctrl{1} & \ghost{\mathrm{\Delta_1}}\ar @{-} [2,0]& \qw\\
	 	\nghost{{q}_{4} :  } & \lstick{{q}_{4} :  } & \gate{\mathrm{Z}}\ar @{-} [2,0] & \qw & \qw\\
	 	\nghost{{q}_{5} :  } & \lstick{{q}_{5} :  } & \qw & \multigate{4}{\mathrm{\Delta_1}} & \qw\\
	 	\nghost{{q}_{6} :  } & \lstick{{q}_{6} :  } & \ctrl{2} & \ghost{\mathrm{\Delta_1}} & \qw\\
	 	\nghost{{q}_{7} :  } & \lstick{{q}_{7} :  } & \qw & \ghost{\mathrm{\Delta_1}} & \qw\\
	 	\nghost{{q}_{8} :  } & \lstick{{q}_{8} :  } & \ctrl{2} & \ghost{\mathrm{\Delta_1}} & \qw\\
	 	\nghost{{q}_{9} :  } & \lstick{{q}_{9} :  } & \qw & \ghost{\mathrm{\Delta_1}}\ar @{-} [2,0] & \qw\\
	 	\nghost{{q}_{10} :  } & \lstick{{q}_{10} :  } & \gate{\mathrm{Z}}\ar @{-} [1,0] & \qw & \qw\\
	 	\nghost{{q}_{11} :  } & \lstick{{q}_{11} :  } & \ctrl{2} & \multigate{7}{\mathrm{\Delta_1}} & \qw\\
	 	\nghost{{q}_{12} :  } & \lstick{{q}_{12} :  } & \qw & \ghost{\mathrm{\Delta_1}} & \qw\\
	 	\nghost{{q}_{13} :  } & \lstick{{q}_{13} :  } & \ctrl{2} & \ghost{\mathrm{\Delta_1}} & \qw\\
	 	\nghost{{q}_{14} :  } & \lstick{{q}_{14} :  } & \qw & \ghost{\mathrm{\Delta_1}}& \qw\\
	 	\nghost{{q}_{15} :  } & \lstick{{q}_{15} :  } & \ctrl{2} & \ghost{\mathrm{\Delta_1}} & \qw\\
	 	\nghost{{q}_{16} :  } & \lstick{{q}_{16} :  } & \qw & \ghost{\mathrm{\Delta_1}} & \qw\\
	 	\nghost{{q}_{17} :  } & \lstick{{q}_{17} :  } & \ctrl{2} & \ghost{\mathrm{\Delta_1}} & \qw\\
	 	\nghost{{q}_{18} :  } & \lstick{{q}_{18} :  } & \qw & \ghost{\mathrm{\Delta_1}}\ar @{-} [2,0] & \qw\\
	 	\nghost{{q}_{19} :  } & \lstick{{q}_{19} :  } & \gate{\mathrm{Z}}\ar @{-} [2,0] & \qw & \qw\\
	 	\nghost{{q}_{20} :  } & \lstick{{q}_{20} :  } & \qw & \gate{\mathrm{\mathrm{\Delta_1}}} & \qw\\
	 	\nghost{{q}_{21} :  } & \lstick{{q}_{21} :  } & \gate{\mathrm{Z}} & \qw & \qw\\
 }
 \hspace{5mm}\raisebox{-50mm}{=}\hspace{0mm}
 \Qcircuit @C=1.0em @R=0.em @!R { 
	 	\nghost{{q}_{1} :  } & \lstick{{q}_{1} :  } & \qw & \qw & \qw & \ctrl{2} & \multigate{2}{\mathrm{\Delta_1}} & \qw & \qw & \qw & \qw\\
	 	\nghost{{q}_{2} :  } & \lstick{{q}_{2} :  } & \qw & \qw & \qw & \qw & \ghost{\mathrm{\Delta_1}}  & \qw & \qw & \qw & \qw\\
	 	\nghost{{q}_{3} :  } & \lstick{{q}_{3} :  } & \qw & \qw & \qw & \ctrl{3} & \ghost{\mathrm{\Delta_1}}\ar @{-} [2,0] & \qw & \qw & \qw & \qw\\
	 	\nghost{{q}_{4} :  } & \lstick{{q}_{4} :  } & \ctrl{6} & \qw & \qw & \qw & \qw & \qw & \ctrl{6} & \qw & \qw\\
	 	\nghost{{q}_{5} :  } & \lstick{{q}_{5} :  } & \qw & \qw & \qw & \qw & \multigate{4}{\mathrm{\Delta_1}} & \qw & \qw & \qw & \qw\\
	 	\nghost{{q}_{6} :  } & \lstick{{q}_{6} :  } & \qw & \qw & \qw & \ctrl{2} & \ghost{\mathrm{\Delta_1}} & \qw & \qw & \qw & \qw\\
	 	\nghost{{q}_{7} :  } & \lstick{{q}_{7} :  } & \qw & \qw & \qw & \qw & \ghost{\mathrm{\Delta_1}} & \qw & \qw & \qw & \qw\\
	 	\nghost{{q}_{8} :  } & \lstick{{q}_{8} :  } & \qw & \qw & \qw & \ctrl{3} & \ghost{\mathrm{\Delta_1}} & \qw & \qw & \qw & \qw\\
	 	\nghost{{q}_{9} :  } & \lstick{{q}_{9} :  } & \qw & \qw & \qw & \qw & \ghost{\mathrm{\Delta_1}}\ar @{-} [2,0] & \qw & \qw & \qw & \qw\\
	 	\nghost{{q}_{10} :  } & \lstick{{q}_{10} :  } & \targ & \ctrl{8} & \qw & \qw & \qw & \ctrl{8} & \targ & \qw & \qw\\
	 	\nghost{{q}_{11} :  } & \lstick{{q}_{11} :  } & \qw & \qw & \qw & \ctrl{2} & \multigate{7}{\mathrm{\Delta_1}} & \qw & \qw & \qw & \qw\\
	 	\nghost{{q}_{12} :  } & \lstick{{q}_{12} :  } & \qw & \qw & \qw & \qw & \ghost{\mathrm{\Delta_1}} & \qw & \qw & \qw & \qw\\
	 	\nghost{{q}_{13} :  } & \lstick{{q}_{13} :  } & \qw & \qw & \qw & \ctrl{2} & \ghost{\mathrm{\Delta_1}} & \qw & \qw & \qw & \qw\\
	 	\nghost{{q}_{14} :  } & \lstick{{q}_{14} :  } & \qw & \qw & \qw & \qw & \ghost{\mathrm{\Delta_1}} & \qw & \qw & \qw & \qw\\
	 	\nghost{{q}_{15} :  } & \lstick{{q}_{15} :  } & \qw & \qw & \qw & \ctrl{2} & \ghost{\mathrm{\Delta_1}} & \qw & \qw & \qw & \qw\\
	 	\nghost{{q}_{16} :  } & \lstick{{q}_{16} :  } & \qw & \qw & \qw & \qw & \ghost{\mathrm{\Delta_1}} & \qw & \qw & \qw & \qw\\
	 	\nghost{{q}_{17} :  } & \lstick{{q}_{17} :  } & \qw & \qw & \qw & \ctrl{2} & \ghost{\mathrm{\Delta_1}} & \qw & \qw & \qw & \qw\\
	 	\nghost{{q}_{18} :  } & \lstick{{q}_{18} :  } & \qw & \targ & \qw & \qw & \ghost{\mathrm{\Delta_1}}\ar @{-} [2,0] & \targ & \qw & \qw & \qw\\
	 	\nghost{{q}_{19} :  } & \lstick{{q}_{19} :  } & \qw & \qw & \qw & \gate{\mathrm{Z}}\ar @{-} [2,0] & \qw & \qw & \qw & \qw & \qw\\
	 	\nghost{{q}_{20} :  } & \lstick{{q}_{20} :  } & \qw & \qw & \qw & \qw & \gate{\mathrm{\mathrm{\Delta_1}}} & \qw & \qw & \qw & \qw\\
	 	\nghost{{q}_{21} :  } & \lstick{{q}_{21} :  } & \qw & \qw & \qw & \gate{\mathrm{Z}} & \qw & \qw & \qw & \qw & \qw\\
 }
 }
 \qcref{MCZ_LNN_add_t}
\]

This requires long range CNOT gates which can be achieved using SWAP gates. However, in order to lower the cost, we can allow for additional CNOT gates to be implemented, if those only result in relative phase gates applied on $Q\setminus \tau$. 

If a pair of CNOT gates controlled by a qubit from $Q\setminus\{C_1,\tau\}$ (the same holds for $C_2$) is applied on a qubit from $Q\setminus C_1$ , this effectively adds an MCZ gate controlled by $C_1$, and applied on the $Q\setminus\{C_1,\tau\}$ qubit, and therefore it can be included in the $[\Delta]$ gate. Moreover, if a CNOT gate controlled by a qubit from $C_1$ is applied on a qubit from $Q\setminus C_1$, this effectively adds an MCZ gate which is only applied on $C_1$, and thus can be included in the $[\Delta]$ gate as well. The following can therefore be used to replace long range CNOT gates in this case. 

\[
\scalebox{1.0}{
\Qcircuit @C=1.0em @R=0.6em @!R { 
	 	\nghost{} & \lstick{} & \ctrl{6} & \qw & \qw\\
	 	\nghost{} & \lstick{} & \qw & \qw & \qw\\
	 	\nghost{} & \lstick{} & \qw & \qw & \qw\\
	 	\nghost{} & \lstick{} & \qw & \qw & \qw\\
	 	\nghost{} & \lstick{} & \qw & \qw & \qw\\
	 	\nghost{} & \lstick{} & \qw & \qw & \qw\\
	 	\nghost{} & \lstick{} & \targ & \qw & \qw\\
 }
 \hspace{5mm}\raisebox{-15mm}{$\Rightarrow$}\hspace{0mm}
\Qcircuit @C=0.2em @R=0.6em @!R { 
	 	\nghost{} & \lstick{} & \ctrl{6} & \qw & \qw & \qw & \qw & \qw & \qw & \qw\\
	 	\nghost{} & \lstick{} & \qw & \ctrl{5} & \qw & \qw & \qw & \qw & \qw & \qw\\
	 	\nghost{} & \lstick{} & \qw & \qw & \ctrl{4} & \qw & \qw & \qw & \qw & \qw\\
	 	\nghost{} & \lstick{} & \qw & \qw & \qw & \ctrl{3} & \qw & \qw & \qw & \qw\\
	 	\nghost{} & \lstick{} & \qw & \qw & \qw & \qw & \ctrl{2} & \qw & \qw & \qw\\
	 	\nghost{} & \lstick{} & \qw & \qw & \qw & \qw & \qw & \ctrl{1} & \qw & \qw\\
	 	\nghost{} & \lstick{} & \targ & \targ & \targ & \targ & \targ & \targ & \qw & \qw\\
 }
 \hspace{5mm}\raisebox{-15mm}{=}\hspace{0mm}
 \Qcircuit @C=0.2em @R=0.6em @!R { 
	 	\nghost{} & \lstick{} & \ctrl{1} & \qw & \qw & \qw & \qw & \qw & \qw & \qw & \qw & \qw & \ctrl{1} & \qw & \qw\\
	 	\nghost{} & \lstick{} & \targ & \ctrl{1} & \qw & \qw & \qw & \qw & \qw & \qw & \qw & \ctrl{1} & \targ & \qw & \qw\\
	 	\nghost{} & \lstick{} & \qw & \targ & \ctrl{1} & \qw & \qw & \qw & \qw & \qw & \ctrl{1} & \targ & \qw & \qw & \qw\\
	 	\nghost{} & \lstick{} & \qw & \qw & \targ & \ctrl{1} & \qw & \qw & \qw & \ctrl{1} & \targ & \qw & \qw & \qw & \qw\\
	 	\nghost{} & \lstick{} & \qw & \qw & \qw & \targ & \ctrl{1} & \qw & \ctrl{1} & \targ & \qw & \qw & \qw & \qw & \qw\\
	 	\nghost{} & \lstick{} & \qw & \qw & \qw & \qw & \targ & \ctrl{1} & \targ & \qw & \qw & \qw & \qw & \qw & \qw\\
	 	\nghost{} & \lstick{} & \qw & \qw & \qw & \qw & \qw & \targ & \qw & \qw & \qw & \qw & \qw & \qw & \qw\\
 }
 }
 \qcref{C_3_LNN_decompose}
\]

Defining $d$ as the distance from the control qubit to the target qubit of the long-range CNOT gate, the cost and depth of this structure is $2d-1$. In an arbitrary permutation, this structure can be applied up-to $2m$ times for the implementation for each MCMTZ gate, at a total cost and depth of $2\sum_{j=1}^{m}(2d_j-1)$, with $d_j$ as the distance of each long range CNOT pair, such that $\sum_{j=1}^{m}(d_j)\leq k$. The added cost and depth of CNOT gates for this step is therefore no larger than $4k-2m$.

Adding all of these steps together, the MCMTSU2 costs and depths are equivalent to the MCSU2 ones, with $l'_{0,1}=l'_{0,2}=0$, such that CNOT cost and depth is increased by no more than $16k$, and $8(m-1)$ gates from $\{\rx,\rz\}$ are added. Since neighboring controls are not allowed in $C_1$, we can only guarantee $q_1\not\in C_2$ , and we can write $k_2\leq k-1$.
Adding the CNOT overhead of $4\floor{\frac{n}{2}}$ in depth $2\floor{\frac{n}{2}}+4$, the following is achieved similarly to \thm{lnn_su2_worst}.
\begin{theorem}\label{thm:lnn_mcmtsu2_worst}
    Any multi-controlled mutli-target $SU(2)$ gate with $n\geq 6$ controls and $m\geq 1$ targets, can be implemented over $k\geq n+m$ qubits in any permutation in LNN connectivity using $8m$ gates from $\{\rx,\rz\}$, in addition to no more than\\
CNOT cost $24k+14n-40$ and depth $24k+5n-4$,\\
T ~~~~~ cost $16n-32$ ~~~~~~~ and depth $4n$,\\
H ~~~~~ cost $8n-8$ ~~~~~~~~~ and depth $4n$.
\end{theorem}
We note that some depth reductions are achieved in many cases, considering the orientation of sequential "V"-shaped structures such as \qc{C_3_LNN_decompose} and \qc{C_1_LNN_decompose}. The structure in \qc{CX_d_LNN_decompose} can be used instead of \qc{C_3_LNN_decompose} in some cases in order to reduce the CNOT depth, while increasing the count of H gates. 
}

{\section{LNN MCMTX}\label{apx:lnn_mcmtx}

Similarly to the ATA case, if the number of targets is larger than one, an ancilla qubit is not required. 
We assume that a non-control qubit is located at the bottom and treat it as an ancilla qubit in order to implement an MCZ gate applied on the set $\{C,q_j\}$ such that $q_j\in\tau$ is the lowest target qubit which satisfies $j<k$. The MCZ gate is implemented using \thm{lnn_MCX_worst}. Each additional target qubit is added similarly to the ATA case by applying a CNOT gate on each side. 
In the LNN case we will add a new target using the closest non-control qubit to it, on which an MCZ controlled by $C$ is already implemented. It is clear that in this case, the long range CNOT gates can be implemented up-to a relative phase as follows.
\[
\scalebox{1.0}{
\Qcircuit @C=1.0em @R=0.6em @!R { 
	 	\nghost{} & \lstick{} & \ctrl{6} & \qw & \qw\\
	 	\nghost{} & \lstick{} & \qw & \qw & \qw\\
	 	\nghost{} & \lstick{} & \qw & \qw & \qw\\
	 	\nghost{} & \lstick{} & \qw & \qw & \qw\\
	 	\nghost{} & \lstick{} & \qw & \qw & \qw\\
	 	\nghost{} & \lstick{} & \qw & \qw & \qw\\
	 	\nghost{} & \lstick{} & \targ & \qw & \qw\\
 }
 \hspace{5mm}\raisebox{-15mm}{$\Rightarrow$}\hspace{0mm}
\Qcircuit @C=1.0em @R=0.2em @!R { 
	 	\nghost{} & \lstick{} & \ctrl{6} & \ctrl{1} & \qw \\
	 	\nghost{} & \lstick{} & \qw & \gate{\mathrm{Z}}\ar @{-} [1,0] & \qw \\
	 	\nghost{} & \lstick{} & \qw & \gate{\mathrm{Z}}\ar @{-} [1,0] & \qw \\
	 	\nghost{} & \lstick{} & \qw & \gate{\mathrm{Z}}\ar @{-} [1,0] & \qw \\
	 	\nghost{} & \lstick{} & \qw & \gate{\mathrm{Z}}\ar @{-} [1,0] & \qw \\
	 	\nghost{} & \lstick{} & \qw & \gate{\mathrm{Z}} & \qw \\
	 	\nghost{} & \lstick{} & \targ & \qw & \qw \\
 }
 \hspace{5mm}\raisebox{-15mm}{=}\hspace{0mm}
 \Qcircuit @C=0.2em @R=0.2em @!R { 
	 	\nghost{} & \lstick{} & \gate{\mathrm{H}} & \qw & \qw & \qw & \qw & \qw & \targ & \qw & \qw & \qw & \qw & \qw & \gate{\mathrm{H}} & \qw & \qw\\
	 	\nghost{} & \lstick{} & \qw & \qw & \qw & \qw & \qw & \targ & \ctrl{-1} & \targ & \qw & \qw & \qw & \qw & \qw & \qw & \qw\\
	 	\nghost{} & \lstick{} & \qw & \qw & \qw & \qw & \targ & \ctrl{-1} & \qw & \ctrl{-1} & \targ & \qw & \qw & \qw & \qw & \qw & \qw\\
	 	\nghost{} & \lstick{} & \qw & \qw & \qw & \targ & \ctrl{-1} & \qw & \qw & \qw & \ctrl{-1} & \targ & \qw & \qw & \qw & \qw & \qw\\
	 	\nghost{} & \lstick{} & \qw & \qw & \targ & \ctrl{-1} & \qw & \qw & \qw & \qw & \qw & \ctrl{-1} & \targ & \qw & \qw & \qw & \qw\\
	 	\nghost{} & \lstick{} & \qw & \targ & \ctrl{-1} & \qw & \qw & \qw & \qw & \qw & \qw & \qw & \ctrl{-1} & \targ & \qw & \qw & \qw\\
	 	\nghost{} & \lstick{} & \gate{\mathrm{H}} & \ctrl{-1} & \qw & \qw & \qw & \qw & \qw & \qw & \qw & \qw & \qw & \ctrl{-1} & \gate{\mathrm{H}} & \qw & \qw\\
 }
 }
 \qcref{CX_d_LNN_decompose}
\]
Finally, a pair of H gates is applied on each target qubit in order to transform the MCMTZ to MCMTX. Each long distance CNOT gate is implemented using four H gates, however, many of these gates cancel out, resulting in an addition of no more than $2m-2$ H gates. 
Noticing the similarities to the process described in \apx{MCMTSU2_lnn}, it can be realised that in addition to the costs described in \thm{lnn_MCX_worst}, this process requires no more than $4k-2m$ CNOT gates.
\begin{theorem}\label{thm:lnn_MCMTX_worst}
    Any multi-controlled multi-target $X$ gate with $n\geq 5$ controls and $m\geq 2$, can be implemented over $k\geq n+m$ qubits in any permutation in LNN connectivity using no more than\\
    CNOT cost $12k+14n-2m-34$ and depth $12k+5n-2m+1$,\\
T ~~~~~ cost $16n-16$ ~~~~~~~~~~~~~~ and depth $4n+4$,\\
H ~~~~~ cost $8n+2m+2$ ~~~~~~~~~ and depth $4n+2m+1$.
\end{theorem}
}


\begin{thebibliography}{0}%
\makeatletter
\providecommand \@ifxundefined [1]{%
 \@ifx{#1\undefined}
}%
\providecommand \@ifnum [1]{%
 \ifnum #1\expandafter \@firstoftwo
 \else \expandafter \@secondoftwo
 \fi
}%
\providecommand \@ifx [1]{%
 \ifx #1\expandafter \@firstoftwo
 \else \expandafter \@secondoftwo
 \fi
}%
\providecommand \natexlab [1]{#1}%
\providecommand \enquote  [1]{``#1''}%
\providecommand \bibnamefont  [1]{#1}%
\providecommand \bibfnamefont [1]{#1}%
\providecommand \citenamefont [1]{#1}%
\providecommand \href@noop [0]{\@secondoftwo}%
\providecommand \href [0]{\begingroup \@sanitize@url \@href}%
\providecommand \@href[1]{\@@startlink{#1}\@@href}%
\providecommand \@@href[1]{\endgroup#1\@@endlink}%
\providecommand \@sanitize@url [0]{\catcode `\\12\catcode `\$12\catcode `\&12\catcode `\#12\catcode `\^12\catcode `\_12\catcode `\%12\relax}%
\providecommand \@@startlink[1]{}%
\providecommand \@@endlink[0]{}%
\providecommand \url  [0]{\begingroup\@sanitize@url \@url }%
\providecommand \@url [1]{\endgroup\@href {#1}{\urlprefix }}%
\providecommand \urlprefix  [0]{URL }%
\providecommand \Eprint [0]{\href }%
\providecommand \doibase [0]{https://doi.org/}%
\providecommand \selectlanguage [0]{\@gobble}%
\providecommand \bibinfo  [0]{\@secondoftwo}%
\providecommand \bibfield  [0]{\@secondoftwo}%
\providecommand \translation [1]{[#1]}%
\providecommand \BibitemOpen [0]{}%
\providecommand \bibitemStop [0]{}%
\providecommand \bibitemNoStop [0]{.\EOS\space}%
\providecommand \EOS [0]{\spacefactor3000\relax}%
\providecommand \BibitemShut  [1]{\csname bibitem#1\endcsname}%
\let\auto@bib@innerbib\@empty
\end{thebibliography}%


\begin{thebibliography}{10}
\providecommand{\url}[1]{#1}
\csname url@samestyle\endcsname
\providecommand{\newblock}{\relax}
\providecommand{\bibinfo}[2]{#2}
\providecommand{\BIBentrySTDinterwordspacing}{\spaceskip=0pt\relax}
\providecommand{\BIBentryALTinterwordstretchfactor}{4}
\providecommand{\BIBentryALTinterwordspacing}{\spaceskip=\fontdimen2\font plus
\BIBentryALTinterwordstretchfactor\fontdimen3\font minus \fontdimen4\font\relax}
\providecommand{\BIBforeignlanguage}[2]{{%
\expandafter\ifx\csname l@#1\endcsname\relax
\typeout{** WARNING: IEEEtran.bst: No hyphenation pattern has been}%
\typeout{** loaded for the language `#1'. Using the pattern for}%
\typeout{** the default language instead.}%
\else
\language=\csname l@#1\endcsname
\fi
#2}}
\providecommand{\BIBdecl}{\relax}
\BIBdecl

\bibitem{nielsen_quantum_2010}
\BIBentryALTinterwordspacing
M.~A. Nielsen and I.~L. Chuang, ``\BIBforeignlanguage{en}{Quantum {Computation} and {Quantum} {Information}: 10th {Anniversary} {Edition}},'' Dec. 2010, iSBN: 9780511976667 Publisher: Cambridge University Press. [Online]. Available: \url{https://www.cambridge.org/highereducation/books/quantum-computation-and-quantum-information/01E10196D0A682A6AEFFEA52D53BE9AE}
\BIBentrySTDinterwordspacing

\bibitem{shende_synthesis_2005}
\BIBentryALTinterwordspacing
V.~V. Shende, S.~S. Bullock, and I.~L. Markov, ``Synthesis of quantum logic circuits,'' in \emph{Proceedings of the 2005 {Asia} and {South} {Pacific} {Design} {Automation} {Conference}}, ser. {ASP}-{DAC} '05.\hskip 1em plus 0.5em minus 0.4em\relax New York, NY, USA: Association for Computing Machinery, Jan. 2005, pp. 272--275. [Online]. Available: \url{https://dl.acm.org/doi/10.1145/1120725.1120847}
\BIBentrySTDinterwordspacing

\bibitem{arrazola_universal_2022}
\BIBentryALTinterwordspacing
J.~M. Arrazola, O.~D. Matteo, N.~Quesada, S.~Jahangiri, A.~Delgado, and N.~Killoran, ``\BIBforeignlanguage{en-GB}{Universal quantum circuits for quantum chemistry},'' \emph{\BIBforeignlanguage{en-GB}{Quantum}}, vol.~6, p. 742, Jun. 2022, publisher: Verein zur Förderung des Open Access Publizierens in den Quantenwissenschaften. [Online]. Available: \url{https://quantum-journal.org/papers/q-2022-06-20-742/}
\BIBentrySTDinterwordspacing

\bibitem{de_carvalho_parametrized_2024}
\BIBentryALTinterwordspacing
J.~H.~A. De~Carvalho and F.~M. D.~P. Neto, ``\BIBforeignlanguage{en}{Parametrized {Constant}-{Depth} {Quantum} {Neuron}},'' \emph{\BIBforeignlanguage{en}{IEEE Transactions on Neural Networks and Learning Systems}}, pp. 1--12, 2024. [Online]. Available: \url{https://ieeexplore.ieee.org/document/10180218/}
\BIBentrySTDinterwordspacing

\bibitem{ni_progressive_2024}
\BIBentryALTinterwordspacing
X.-H. Ni, L.-X. Li, Y.-Q. Song, Z.-P. Jin, S.-J. Qin, and F.~Gao, ``Progressive {Quantum} {Algorithm} for {Maximum} {Independent} {Set} with {Quantum} {Alternating} {Operator} {Ansatz},'' Sep. 2024, arXiv:2405.04303 [quant-ph]. [Online]. Available: \url{http://arxiv.org/abs/2405.04303}
\BIBentrySTDinterwordspacing

\bibitem{iten_quantum_2016}
\BIBentryALTinterwordspacing
R.~Iten, R.~Colbeck, I.~Kukuljan, J.~Home, and M.~Christandl, ``\BIBforeignlanguage{en}{Quantum circuits for isometries},'' \emph{\BIBforeignlanguage{en}{Physical Review A}}, vol.~93, no.~3, p. 032318, Mar. 2016. [Online]. Available: \url{https://link.aps.org/doi/10.1103/PhysRevA.93.032318}
\BIBentrySTDinterwordspacing

\bibitem{malvetti_quantum_2021}
\BIBentryALTinterwordspacing
E.~Malvetti, R.~Iten, and R.~Colbeck, ``\BIBforeignlanguage{en-GB}{Quantum {Circuits} for {Sparse} {Isometries}},'' \emph{\BIBforeignlanguage{en-GB}{Quantum}}, vol.~5, p. 412, Mar. 2021, publisher: Verein zur Förderung des Open Access Publizierens in den Quantenwissenschaften. [Online]. Available: \url{https://quantum-journal.org/papers/q-2021-03-15-412/}
\BIBentrySTDinterwordspacing

\bibitem{grinko_efficient_2023}
\BIBentryALTinterwordspacing
D.~Grinko, A.~Burchardt, and M.~Ozols, ``\BIBforeignlanguage{en}{Efficient quantum circuits for port-based teleportation},'' 2023, version Number: 2. [Online]. Available: \url{https://arxiv.org/abs/2312.03188}
\BIBentrySTDinterwordspacing

\bibitem{tanasescu_distribution_2022}
\BIBentryALTinterwordspacing
A.~Tănăsescu, D.~Constantinescu, and P.~G. Popescu, ``\BIBforeignlanguage{en}{Distribution of controlled unitary quantum gates towards factoring large numbers on today’s small-register devices},'' \emph{\BIBforeignlanguage{en}{Scientific Reports}}, vol.~12, no.~1, p. 21310, Dec. 2022, publisher: Nature Publishing Group. [Online]. Available: \url{https://www.nature.com/articles/s41598-022-25812-z}
\BIBentrySTDinterwordspacing

\bibitem{park_circuit-based_2019}
\BIBentryALTinterwordspacing
D.~K. Park, F.~Petruccione, and J.-K.~K. Rhee, ``\BIBforeignlanguage{en}{Circuit-{Based} {Quantum} {Random} {Access} {Memory} for {Classical} {Data}},'' \emph{\BIBforeignlanguage{en}{Scientific Reports}}, vol.~9, no.~1, p. 3949, Mar. 2019. [Online]. Available: \url{https://www.nature.com/articles/s41598-019-40439-3}
\BIBentrySTDinterwordspacing

\bibitem{grover_quantum_1997}
\BIBentryALTinterwordspacing
L.~K. Grover, ``Quantum {Mechanics} {Helps} in {Searching} for a {Needle} in a {Haystack},'' \emph{Physical Review Letters}, vol.~79, no.~2, pp. 325--328, Jul. 1997, publisher: American Physical Society. [Online]. Available: \url{https://link.aps.org/doi/10.1103/PhysRevLett.79.325}
\BIBentrySTDinterwordspacing

\bibitem{ali_function_2018}
\BIBentryALTinterwordspacing
M.~B. Ali, T.~Hirayama, K.~Yamanaka, and Y.~Nishitani, ``\BIBforeignlanguage{en}{Function {Design} for {Minimum} {Multiple}-{Control} {Toffoli} {Circuits} of {Reversible} {Adder}/{Subtractor} {Blocks} and {Arithmetic} {Logic} {Units}},'' \emph{\BIBforeignlanguage{en}{IEICE Transactions on Fundamentals of Electronics, Communications and Computer Sciences}}, vol. E101.A, no.~12, pp. 2231--2243, Dec. 2018. [Online]. Available: \url{https://www.jstage.jst.go.jp/article/transfun/E101.A/12/E101.A_2231/_article}
\BIBentrySTDinterwordspacing

\bibitem{shafaei_cofactor_2014}
\BIBentryALTinterwordspacing
A.~Shafaei, M.~Saeedi, and M.~Pedram, ``\BIBforeignlanguage{en}{Cofactor {Sharing} for {Reversible} {Logic} {Synthesis}},'' \emph{\BIBforeignlanguage{en}{ACM Journal on Emerging Technologies in Computing Systems}}, vol.~11, no.~2, pp. 1--21, Nov. 2014. [Online]. Available: \url{https://dl.acm.org/doi/10.1145/2629524}
\BIBentrySTDinterwordspacing

\bibitem{lewis_matrix_2022}
\BIBentryALTinterwordspacing
M.~Lewis, S.~Soudjani, and P.~Zuliani, ``Matrix {Representation} of {Arbitrarily} {Controlled} {Quantum} {Gates},'' May 2022, arXiv:2205.02525 [quant-ph]. [Online]. Available: \url{http://arxiv.org/abs/2205.02525}
\BIBentrySTDinterwordspacing

\bibitem{barenco_elementary_1995}
\BIBentryALTinterwordspacing
A.~Barenco, C.~H. Bennett, R.~Cleve, D.~P. DiVincenzo, N.~Margolus, P.~Shor, T.~Sleator, J.~A. Smolin, and H.~Weinfurter, ``Elementary gates for quantum computation,'' \emph{Physical Review A}, vol.~52, no.~5, pp. 3457--3467, Nov. 1995, publisher: American Physical Society. [Online]. Available: \url{https://link.aps.org/doi/10.1103/PhysRevA.52.3457}
\BIBentrySTDinterwordspacing

\bibitem{isenhower_multibit_2011}
\BIBentryALTinterwordspacing
L.~Isenhower, M.~Saffman, and K.~Mølmer, ``\BIBforeignlanguage{en}{Multibit {CkNOT} quantum gates via {Rydberg} blockade},'' \emph{\BIBforeignlanguage{en}{Quantum Information Processing}}, vol.~10, no.~6, p. 755, Sep. 2011. [Online]. Available: \url{https://doi.org/10.1007/s11128-011-0292-4}
\BIBentrySTDinterwordspacing

\bibitem{petrosyan_fast_2024}
\BIBentryALTinterwordspacing
D.~Petrosyan, S.~Norrell, C.~Poole, and M.~Saffman, ``\BIBforeignlanguage{en}{Fast measurements and multiqubit gates in dual-species atomic arrays},'' \emph{\BIBforeignlanguage{en}{Physical Review A}}, vol. 110, no.~4, p. 042404, Oct. 2024. [Online]. Available: \url{https://link.aps.org/doi/10.1103/PhysRevA.110.042404}
\BIBentrySTDinterwordspacing

\bibitem{jandura_time-optimal_2022}
\BIBentryALTinterwordspacing
S.~Jandura and G.~Pupillo, ``Time-{Optimal} {Two}- and {Three}-{Qubit} {Gates} for {Rydberg} {Atoms},'' \emph{Quantum}, vol.~6, p. 712, May 2022, arXiv:2202.00903 [quant-ph]. [Online]. Available: \url{http://arxiv.org/abs/2202.00903}
\BIBentrySTDinterwordspacing

\bibitem{abughanem_ibm_2024}
\BIBentryALTinterwordspacing
M.~AbuGhanem, ``{IBM} {Quantum} {Computers}: {Evolution}, {Performance}, and {Future} {Directions},'' Sep. 2024, arXiv:2410.00916 [quant-ph]. [Online]. Available: \url{http://arxiv.org/abs/2410.00916}
\BIBentrySTDinterwordspacing

\bibitem{abughanem_toffoli_2025}
\BIBentryALTinterwordspacing
------, ``A {Toffoli} {Gate} {Decomposition} via {Echoed} {Cross}-{Resonance} {Gates},'' Jan. 2025, arXiv:2501.02222 [quant-ph]. [Online]. Available: \url{http://arxiv.org/abs/2501.02222}
\BIBentrySTDinterwordspacing

\bibitem{xue_quantum_2022}
\BIBentryALTinterwordspacing
X.~Xue, M.~Russ, N.~Samkharadze, B.~Undseth, A.~Sammak, G.~Scappucci, and L.~M.~K. Vandersypen, ``\BIBforeignlanguage{en}{Quantum logic with spin qubits crossing the surface code threshold},'' \emph{\BIBforeignlanguage{en}{Nature}}, vol. 601, no. 7893, pp. 343--347, Jan. 2022. [Online]. Available: \url{https://www.nature.com/articles/s41586-021-04273-w}
\BIBentrySTDinterwordspacing

\bibitem{bravyi_universal_2005}
\BIBentryALTinterwordspacing
S.~Bravyi and A.~Kitaev, ``\BIBforeignlanguage{en}{Universal quantum computation with ideal {Clifford} gates and noisy ancillas},'' \emph{\BIBforeignlanguage{en}{Physical Review A}}, vol.~71, no.~2, p. 022316, Feb. 2005. [Online]. Available: \url{https://link.aps.org/doi/10.1103/PhysRevA.71.022316}
\BIBentrySTDinterwordspacing

\bibitem{nielsen_quantum_2002}
\BIBentryALTinterwordspacing
M.~A. Nielsen and I.~Chuang, ``Quantum {Computation} and {Quantum} {Information},'' \emph{American Journal of Physics}, vol.~70, no.~5, pp. 558--559, May 2002, publisher: American Association of Physics Teachers. [Online]. Available: \url{https://aapt.scitation.org/doi/10.1119/1.1463744}
\BIBentrySTDinterwordspacing

\bibitem{jones_logic_2013}
\BIBentryALTinterwordspacing
N.~C. Jones, ``Logic {Synthesis} for {Fault}-{Tolerant} {Quantum} {Computers},'' Oct. 2013, arXiv:1310.7290 [quant-ph]. [Online]. Available: \url{http://arxiv.org/abs/1310.7290}
\BIBentrySTDinterwordspacing

\bibitem{vasmer_three-dimensional_2019}
\BIBentryALTinterwordspacing
M.~Vasmer and D.~E. Browne, ``\BIBforeignlanguage{en}{Three-dimensional surface codes: {Transversal} gates and fault-tolerant architectures},'' \emph{\BIBforeignlanguage{en}{Physical Review A}}, vol. 100, no.~1, p. 012312, Jul. 2019. [Online]. Available: \url{https://link.aps.org/doi/10.1103/PhysRevA.100.012312}
\BIBentrySTDinterwordspacing

\bibitem{butt_fault-tolerant_2024}
\BIBentryALTinterwordspacing
F.~Butt, S.~Heußen, M.~Rispler, and M.~Müller, ``Fault-{Tolerant} {Code} {Switching} {Protocols} for {Near}-{Term} {Quantum} {Processors},'' \emph{PRX Quantum}, vol.~5, no.~2, p. 020345, May 2024, arXiv:2306.17686 [quant-ph]. [Online]. Available: \url{http://arxiv.org/abs/2306.17686}
\BIBentrySTDinterwordspacing

\bibitem{cheng_mapping_2018}
\BIBentryALTinterwordspacing
X.~Cheng, Z.~Guan, and W.~Ding, ``\BIBforeignlanguage{en}{Mapping from multiple-control {Toffoli} circuits to linear nearest neighbor quantum circuits},'' \emph{\BIBforeignlanguage{en}{Quantum Information Processing}}, vol.~17, no.~7, p. 169, May 2018. [Online]. Available: \url{https://doi.org/10.1007/s11128-018-1908-8}
\BIBentrySTDinterwordspacing

\bibitem{schoenberger_shuttling_2024}
\BIBentryALTinterwordspacing
D.~Schoenberger, S.~Hillmich, M.~Brandl, and R.~Wille, ``Shuttling for {Scalable} {Trapped}-{Ion} {Quantum} {Computers},'' Oct. 2024, arXiv:2402.14065 [quant-ph]. [Online]. Available: \url{http://arxiv.org/abs/2402.14065}
\BIBentrySTDinterwordspacing

\bibitem{bluvstein_logical_2024}
\BIBentryALTinterwordspacing
D.~Bluvstein, S.~J. Evered, A.~A. Geim, S.~H. Li, H.~Zhou, T.~Manovitz, S.~Ebadi, M.~Cain, M.~Kalinowski, D.~Hangleiter, J.~P. Bonilla~Ataides, N.~Maskara, I.~Cong, X.~Gao, P.~Sales~Rodriguez, T.~Karolyshyn, G.~Semeghini, M.~J. Gullans, M.~Greiner, V.~Vuletić, and M.~D. Lukin, ``\BIBforeignlanguage{en}{Logical quantum processor based on reconfigurable atom arrays},'' \emph{\BIBforeignlanguage{en}{Nature}}, vol. 626, no. 7997, pp. 58--65, Feb. 2024, publisher: Nature Publishing Group. [Online]. Available: \url{https://www.nature.com/articles/s41586-023-06927-3}
\BIBentrySTDinterwordspacing

\bibitem{pedram_layout_2016}
\BIBentryALTinterwordspacing
M.~Pedram and A.~Shafaei, ``\BIBforeignlanguage{en}{Layout {Optimization} for {Quantum} {Circuits} with {Linear} {Nearest} {Neighbor} {Architectures}},'' \emph{\BIBforeignlanguage{en}{IEEE Circuits and Systems Magazine}}, vol.~16, no.~2, pp. 62--74, 2016. [Online]. Available: \url{http://ieeexplore.ieee.org/document/7476978/}
\BIBentrySTDinterwordspacing

\bibitem{saeedi_synthesis_2011}
\BIBentryALTinterwordspacing
M.~Saeedi, R.~Wille, and R.~Drechsler, ``\BIBforeignlanguage{en}{Synthesis of quantum circuits for linear nearest neighbor architectures},'' \emph{\BIBforeignlanguage{en}{Quantum Information Processing}}, vol.~10, no.~3, pp. 355--377, Jun. 2011. [Online]. Available: \url{https://doi.org/10.1007/s11128-010-0201-2}
\BIBentrySTDinterwordspacing

\bibitem{he_mapping_2019}
\BIBentryALTinterwordspacing
X.~He, Z.~Guan, and F.~Ding, ``\BIBforeignlanguage{en}{The {Mapping} and {Optimization} {Method} of {Quantum} {Circuits} for {Clifford} + {T} {Gate}},'' \emph{\BIBforeignlanguage{en}{Journal of Applied Mathematics and Physics}}, vol.~07, no.~11, pp. 2796--2810, 2019. [Online]. Available: \url{http://www.scirp.org/journal/doi.aspx?DOI=10.4236/jamp.2019.711192}
\BIBentrySTDinterwordspacing

\bibitem{khan_cost_2008}
M.~H.~A. Khan, ``\BIBforeignlanguage{en}{Cost {Reduction} in {Nearest} {Neighbour} {Based} {Synthesis} of {Quantum} {Boolean} {Circuits}},'' \emph{\BIBforeignlanguage{en}{Engineering Letters}}, vol.~16, no.~1, 2008.

\bibitem{lukac_optimization_2021}
\BIBentryALTinterwordspacing
M.~Lukac, P.~Kerntopf, and M.~Kameyama, ``Optimization of {LNN} {Reversible} {Circuits} {Using} an {Analytic} {Sifting} {Method},'' \emph{Journal of Circuits, Systems and Computers}, vol.~30, no.~09, p. 2150166, Jul. 2021, publisher: World Scientific Publishing Co. [Online]. Available: \url{https://www.worldscientific.com/doi/abs/10.1142/S0218126621501668}
\BIBentrySTDinterwordspacing

\bibitem{ding_fast_2019}
\BIBentryALTinterwordspacing
F.~Ding, Z.~Guau, and F.~Ren, ``A {Fast} {Optimization} {Algorithm} for {Nearest} {Neighbor} {Architecture} {Based} on {Quantum} {Weight},'' in \emph{2019 {IEEE} 6th {International} {Conference} on {Cloud} {Computing} and {Intelligence} {Systems} ({CCIS})}, Dec. 2019, pp. 73--78. [Online]. Available: \url{https://ieeexplore.ieee.org/document/9073740/?arnumber=9073740}
\BIBentrySTDinterwordspacing

\bibitem{zhang_method_2022}
\BIBentryALTinterwordspacing
C.~Zhang, Z.~Guan, Y.~Qian, and S.~Feng, ``A {Method} of {Mapping} and {Nearest}-{Neighbor} for {IBM} {QX} {Architecture},'' in \emph{2022 {International} {Conference} on {Computing}, {Communication}, {Perception} and {Quantum} {Technology} ({CCPQT})}, Aug. 2022, pp. 402--409. [Online]. Available: \url{https://ieeexplore.ieee.org/document/9939331/?arnumber=9939331}
\BIBentrySTDinterwordspacing

\bibitem{vale_decomposition_2023}
\BIBentryALTinterwordspacing
R.~Vale, T.~M.~D. Azevedo, I.~C.~S. Araújo, I.~F. Araujo, and A.~J. da~Silva, ``Decomposition of {Multi}-controlled {Special} {Unitary} {Single}-{Qubit} {Gates},'' Feb. 2023, arXiv:2302.06377 [quant-ph]. [Online]. Available: \url{http://arxiv.org/abs/2302.06377}
\BIBentrySTDinterwordspacing

\bibitem{he_decompositions_2017}
\BIBentryALTinterwordspacing
Y.~He, M.-X. Luo, E.~Zhang, H.-K. Wang, and X.-F. Wang, ``\BIBforeignlanguage{en}{Decompositions of n-qubit {Toffoli} {Gates} with {Linear} {Circuit} {Complexity}},'' \emph{\BIBforeignlanguage{en}{International Journal of Theoretical Physics}}, vol.~56, no.~7, pp. 2350--2361, Jul. 2017. [Online]. Available: \url{https://doi.org/10.1007/s10773-017-3389-4}
\BIBentrySTDinterwordspacing

\bibitem{balauca_efficient_2022}
S.~Balauca and A.~Arusoaie, ``\BIBforeignlanguage{en}{Efficient {Constructions} for {Simulating} {Multi} {Controlled} {Quantum} {Gates}},'' in \emph{\BIBforeignlanguage{en}{Computational {Science} – {ICCS} 2022}}, ser. Lecture {Notes} in {Computer} {Science}, D.~Groen, C.~de~Mulatier, M.~Paszynski, V.~V. Krzhizhanovskaya, J.~J. Dongarra, and P.~M.~A. Sloot, Eds.\hskip 1em plus 0.5em minus 0.4em\relax Cham: Springer International Publishing, 2022, pp. 179--194.

\bibitem{maslov_reversible_2011}
\BIBentryALTinterwordspacing
D.~Maslov and M.~Saeedi, ``Reversible {Circuit} {Optimization} via {Leaving} the {Boolean} {Domain},'' \emph{IEEE Transactions on Computer-Aided Design of Integrated Circuits and Systems}, vol.~30, no.~6, pp. 806--816, Jun. 2011, arXiv:1103.0215 [quant-ph]. [Online]. Available: \url{http://arxiv.org/abs/1103.0215}
\BIBentrySTDinterwordspacing

\bibitem{kole_improved_2017}
\BIBentryALTinterwordspacing
A.~Kole and K.~Datta, ``Improved {NCV} {Gate} {Realization} of {Arbitrary} {Size} {Toffoli} {Gates},'' in \emph{2017 30th {International} {Conference} on {VLSI} {Design} and 2017 16th {International} {Conference} on {Embedded} {Systems} ({VLSID})}, Jan. 2017, pp. 289--294, iSSN: 2380-6923. [Online]. Available: \url{https://ieeexplore.ieee.org/document/7884793/?arnumber=7884793}
\BIBentrySTDinterwordspacing

\bibitem{sasanian_ncv_2011}
\BIBentryALTinterwordspacing
Z.~Sasanian and D.~M. Miller, ``{NCV} realization of {MCT} gates with mixed controls,'' in \emph{Proceedings of 2011 {IEEE} {Pacific} {Rim} {Conference} on {Communications}, {Computers} and {Signal} {Processing}}, Aug. 2011, pp. 567--571, iSSN: 2154-5952. [Online]. Available: \url{https://ieeexplore.ieee.org/document/6032956/?arnumber=6032956}
\BIBentrySTDinterwordspacing

\bibitem{biswal_improving_2016}
\BIBentryALTinterwordspacing
L.~Biswal, C.~Bandyopadhyay, R.~Wille, R.~Drechsler, and H.~Rahaman, ``Improving the {Realization} of {Multiple}-{Control} {Toffoli} {Gates} {Using} the {NCVW} {Quantum} {Gate} {Library},'' in \emph{2016 29th {International} {Conference} on {VLSI} {Design} and 2016 15th {International} {Conference} on {Embedded} {Systems} ({VLSID})}, Jan. 2016, pp. 573--574, iSSN: 2380-6923. [Online]. Available: \url{https://ieeexplore.ieee.org/document/7435022/?arnumber=7435022}
\BIBentrySTDinterwordspacing

\bibitem{ali_quantum_2015}
\BIBentryALTinterwordspacing
M.~B. Ali, T.~Hirayama, K.~Yamanaka, and Y.~Nishitani, ``Quantum {Cost} {Reduction} of {Reversible} {Circuits} {Using} {New} {Toffoli} {Decomposition} {Techniques},'' in \emph{2015 {International} {Conference} on {Computational} {Science} and {Computational} {Intelligence} ({CSCI})}, Dec. 2015, pp. 59--64. [Online]. Available: \url{https://ieeexplore.ieee.org/document/7424064/?arnumber=7424064}
\BIBentrySTDinterwordspacing

\bibitem{niemann_t-depth_2019}
\BIBentryALTinterwordspacing
P.~Niemann, A.~Gupta, and R.~Drechsler, ``T-depth {Optimization} for {Fault}-{Tolerant} {Quantum} {Circuits},'' in \emph{2019 {IEEE} 49th {International} {Symposium} on {Multiple}-{Valued} {Logic} ({ISMVL})}, May 2019, pp. 108--113, iSSN: 2378-2226. [Online]. Available: \url{https://ieeexplore.ieee.org/document/8758745/?arnumber=8758745}
\BIBentrySTDinterwordspacing

\bibitem{leng_decomposing_2024}
\BIBentryALTinterwordspacing
J.~Leng, F.~Yang, and X.-B. Wang, ``\BIBforeignlanguage{en}{Decomposing -{Qubit} {Toffoli} {Gate} with {Shallow} {Circuit} {Depth} and {No} {Ancilla}},'' \emph{\BIBforeignlanguage{en}{Advanced Quantum Technologies}}, vol.~7, no.~2, p. 2300370, 2024. [Online]. Available: \url{https://onlinelibrary.wiley.com/doi/abs/10.1002/qute.202300370}
\BIBentrySTDinterwordspacing

\bibitem{abdessaied_technology_2016}
\BIBentryALTinterwordspacing
N.~Abdessaied, M.~Amy, M.~Soeken, and R.~Drechsler, ``Technology {Mapping} of {Reversible} {Circuits} to {Clifford}+{T} {Quantum} {Circuits},'' in \emph{2016 {IEEE} 46th {International} {Symposium} on {Multiple}-{Valued} {Logic} ({ISMVL})}, May 2016, pp. 150--155, iSSN: 2378-2226. [Online]. Available: \url{https://ieeexplore.ieee.org/document/7515539/?arnumber=7515539}
\BIBentrySTDinterwordspacing

\bibitem{arsoski_implementing_2024}
\BIBentryALTinterwordspacing
V.~V. Arsoski, ``\BIBforeignlanguage{en}{Implementing multi-controlled {X} gates using the quantum {Fourier} transform},'' \emph{\BIBforeignlanguage{en}{Quantum Information Processing}}, vol.~23, no.~9, p. 305, Aug. 2024. [Online]. Available: \url{https://doi.org/10.1007/s11128-024-04511-w}
\BIBentrySTDinterwordspacing

\bibitem{arsoski_multi-controlled_2025}
\BIBentryALTinterwordspacing
------, ``Multi-controlled single-qubit unitary gates based on the quantum {Fourier} transform and deep decomposition,'' Feb. 2025, arXiv:2408.00935 [quant-ph]. [Online]. Available: \url{http://arxiv.org/abs/2408.00935}
\BIBentrySTDinterwordspacing

\bibitem{chakrabarti_nearest_2007}
\BIBentryALTinterwordspacing
A.~Chakrabarti and S.~Sur, ``\BIBforeignlanguage{en}{Nearest {Neighbour} based {Synthesis} of {Quantum} {Boolean} {Circuits}},'' \emph{\BIBforeignlanguage{en}{Engineering Letters}}, 2007. [Online]. Available: \url{https://www.engineeringletters.com/issues_v15/issue_2/EL_15_2_26.pdf}
\BIBentrySTDinterwordspacing

\bibitem{miller_elementary_2011}
\BIBentryALTinterwordspacing
D.~M. Miller, R.~Wille, and Z.~Sasanian, ``Elementary {Quantum} {Gate} {Realizations} for {Multiple}-{Control} {Toffoli} {Gates},'' in \emph{2011 41st {IEEE} {International} {Symposium} on {Multiple}-{Valued} {Logic}}, May 2011, pp. 288--293, iSSN: 2378-2226. [Online]. Available: \url{https://ieeexplore.ieee.org/document/5954249/?arnumber=5954249&tag=1}
\BIBentrySTDinterwordspacing

\bibitem{li_quantum_2023}
\BIBentryALTinterwordspacing
Y.~Li, W.~Liu, M.~Li, and Y.~Li, ``\BIBforeignlanguage{en}{Quantum circuit compilation for nearest-neighbor architecture based on reinforcement learning},'' \emph{\BIBforeignlanguage{en}{Quantum Information Processing}}, vol.~22, no.~8, p. 295, Jul. 2023. [Online]. Available: \url{https://doi.org/10.1007/s11128-023-04050-w}
\BIBentrySTDinterwordspacing

\bibitem{tan_multi-strategy_2018}
\BIBentryALTinterwordspacing
Y.-y. Tan, X.-y. Cheng, Z.-j. Guan, Y.~Liu, and H.~Ma, ``\BIBforeignlanguage{en}{Multi-strategy based quantum cost reduction of linear nearest-neighbor quantum circuit},'' \emph{\BIBforeignlanguage{en}{Quantum Information Processing}}, vol.~17, no.~3, p.~61, Jan. 2018. [Online]. Available: \url{https://doi.org/10.1007/s11128-018-1832-y}
\BIBentrySTDinterwordspacing

\bibitem{silva_linear_2023}
\BIBentryALTinterwordspacing
J.~D.~S. Silva, T.~M.~D. Azevedo, I.~F. Araujo, and A.~J. da~Silva, ``Linear decomposition of approximate multi-controlled single qubit gates,'' Oct. 2023, arXiv:2310.14974 [quant-ph]. [Online]. Available: \url{http://arxiv.org/abs/2310.14974}
\BIBentrySTDinterwordspacing

\bibitem{da_silva_linear-depth_2022}
\BIBentryALTinterwordspacing
A.~J. da~Silva and D.~K. Park, ``Linear-depth quantum circuits for multiqubit controlled gates,'' \emph{Physical Review A}, vol. 106, no.~4, p. 042602, Oct. 2022, publisher: American Physical Society. [Online]. Available: \url{https://link.aps.org/doi/10.1103/PhysRevA.106.042602}
\BIBentrySTDinterwordspacing

\bibitem{saeedi_linear-depth_2013}
\BIBentryALTinterwordspacing
M.~Saeedi and M.~Pedram, ``Linear-depth quantum circuits for \$n\$-qubit {Toffoli} gates with no ancilla,'' \emph{Physical Review A}, vol.~87, no.~6, p. 062318, Jun. 2013, publisher: American Physical Society. [Online]. Available: \url{https://link.aps.org/doi/10.1103/PhysRevA.87.062318}
\BIBentrySTDinterwordspacing

\bibitem{gidney_constructing_nodate}
\BIBentryALTinterwordspacing
C.~Gidney, ``Constructing {Large} {Controlled} {Nots}.'' [Online]. Available: \url{https://algassert.com/circuits/2015/06/05/Constructing-Large-Controlled-Nots.html}
\BIBentrySTDinterwordspacing

\bibitem{cruz_shallow_2023}
\BIBentryALTinterwordspacing
P.~M.~Q. Cruz and B.~Murta, ``Shallow unitary decompositions of quantum {Fredkin} and {Toffoli} gates for connectivity-aware equivalent circuit averaging,'' May 2023, arXiv:2305.18128 [quant-ph]. [Online]. Available: \url{http://arxiv.org/abs/2305.18128}
\BIBentrySTDinterwordspacing

\bibitem{gwinner_benchmarking_2021}
\BIBentryALTinterwordspacing
J.~Gwinner, M.~Briański, W.~Burkot, L.~Czerwiński, and V.~Hlembotskyi, ``Benchmarking 16-element quantum search algorithms on superconducting quantum processors,'' Jan. 2021, arXiv:2007.06539 [quant-ph]. [Online]. Available: \url{http://arxiv.org/abs/2007.06539}
\BIBentrySTDinterwordspacing

\bibitem{nakanishi_decompositions_2024}
\BIBentryALTinterwordspacing
K.~M. Nakanishi, T.~Satoh, and S.~Todo, ``\BIBforeignlanguage{en}{Decompositions of multiple controlled- {Z} gates on various qubit-coupling graphs},'' \emph{\BIBforeignlanguage{en}{Physical Review A}}, vol. 110, no.~1, p. 012604, Jul. 2024. [Online]. Available: \url{https://link.aps.org/doi/10.1103/PhysRevA.110.012604}
\BIBentrySTDinterwordspacing

\bibitem{nakanishi_quantum-gate_2021}
\BIBentryALTinterwordspacing
------, ``Quantum-gate decomposer,'' Sep. 2021, arXiv:2109.13223 [quant-ph]. [Online]. Available: \url{http://arxiv.org/abs/2109.13223}
\BIBentrySTDinterwordspacing

\bibitem{nemkov_efficient_2023}
\BIBentryALTinterwordspacing
N.~A. Nemkov, E.~O. Kiktenko, I.~A. Luchnikov, and A.~K. Fedorov, ``\BIBforeignlanguage{en}{Efficient variational synthesis of quantum circuits with coherent multi-start optimization},'' \emph{\BIBforeignlanguage{en}{Quantum}}, vol.~7, p. 993, May 2023. [Online]. Available: \url{https://quantum-journal.org/papers/q-2023-05-04-993/}
\BIBentrySTDinterwordspacing

\bibitem{song_optimal_2002}
\BIBentryALTinterwordspacing
G.~Song and A.~Klappenecker, ``Optimal {Realizations} of {Controlled} {Unitary} {Gates},'' Jul. 2002, arXiv:quant-ph/0207157. [Online]. Available: \url{http://arxiv.org/abs/quant-ph/0207157}
\BIBentrySTDinterwordspacing

\bibitem{zindorf_quantum_2024}
\BIBentryALTinterwordspacing
B.~Zindorf and S.~Bose, ``Quantum {Computing} with {Hermitian} {Gates},'' Feb. 2024, arXiv:2402.12356 [quant-ph]. [Online]. Available: \url{http://arxiv.org/abs/2402.12356}
\BIBentrySTDinterwordspacing

\bibitem{heusen_measurement-free_2024}
\BIBentryALTinterwordspacing
S.~Heußen, D.~F. Locher, and M.~Müller, ``\BIBforeignlanguage{en}{Measurement-{Free} {Fault}-{Tolerant} {Quantum} {Error} {Correction} in {Near}-{Term} {Devices}},'' \emph{\BIBforeignlanguage{en}{PRX Quantum}}, vol.~5, no.~1, p. 010333, Feb. 2024. [Online]. Available: \url{https://link.aps.org/doi/10.1103/PRXQuantum.5.010333}
\BIBentrySTDinterwordspacing

\bibitem{claudon_polylogarithmic-depth_2024}
\BIBentryALTinterwordspacing
B.~Claudon, J.~Zylberman, C.~Feniou, F.~Debbasch, A.~Peruzzo, and J.-P. Piquemal, ``\BIBforeignlanguage{en}{Polylogarithmic-depth controlled-{NOT} gates without ancilla qubits},'' \emph{\BIBforeignlanguage{en}{Nature Communications}}, vol.~15, no.~1, p. 5886, Jul. 2024, publisher: Nature Publishing Group. [Online]. Available: \url{https://www.nature.com/articles/s41467-024-50065-x}
\BIBentrySTDinterwordspacing

\bibitem{nie_quantum_2024}
\BIBentryALTinterwordspacing
J.~Nie, W.~Zi, and X.~Sun, ``Quantum circuit for multi-qubit {Toffoli} gate with optimal resource,'' Feb. 2024, arXiv:2402.05053 [quant-ph]. [Online]. Available: \url{http://arxiv.org/abs/2402.05053}
\BIBentrySTDinterwordspacing

\bibitem{figgatt_parallel_2019}
\BIBentryALTinterwordspacing
C.~Figgatt, A.~Ostrander, N.~M. Linke, K.~A. Landsman, D.~Zhu, D.~Maslov, and C.~Monroe, ``\BIBforeignlanguage{en}{Parallel entangling operations on a universal ion-trap quantum computer},'' \emph{\BIBforeignlanguage{en}{Nature}}, vol. 572, no. 7769, pp. 368--372, Aug. 2019. [Online]. Available: \url{https://www.nature.com/articles/s41586-019-1427-5}
\BIBentrySTDinterwordspacing

\bibitem{liu_qcontext_2023}
\BIBentryALTinterwordspacing
J.~Liu, M.~Bowman, P.~Gokhale, S.~Dangwal, J.~Larson, F.~T. Chong, and P.~D. Hovland, ``{QContext}: {Context}-{Aware} {Decomposition} for {Quantum} {Gates},'' in \emph{2023 {IEEE} {International} {Symposium} on {Circuits} and {Systems} ({ISCAS})}, May 2023, pp. 1--5, iSSN: 2158-1525. [Online]. Available: \url{https://ieeexplore.ieee.org/document/10181370/?arnumber=10181370}
\BIBentrySTDinterwordspacing

\bibitem{meijer_exploiting_2024}
\BIBentryALTinterwordspacing
F.~d. Meijer, D.~Gijswijt, and R.~Sotirov, ``Exploiting {Symmetries} in {Optimal} {Quantum} {Circuit} {Design},'' Jan. 2024, arXiv:2401.08262 [math]. [Online]. Available: \url{http://arxiv.org/abs/2401.08262}
\BIBentrySTDinterwordspacing

\bibitem{giles_exact_2013}
\BIBentryALTinterwordspacing
B.~Giles and P.~Selinger, ``Exact synthesis of multiqubit {Clifford}+{T} circuits,'' \emph{Physical Review A}, vol.~87, no.~3, p. 032332, Mar. 2013, arXiv:1212.0506 [quant-ph]. [Online]. Available: \url{http://arxiv.org/abs/1212.0506}
\BIBentrySTDinterwordspacing

\bibitem{maslov_advantages_2016}
\BIBentryALTinterwordspacing
D.~Maslov, ``Advantages of using relative-phase {Toffoli} gates with an application to multiple control {Toffoli} optimization,'' \emph{Physical Review A}, vol.~93, no.~2, p. 022311, Feb. 2016, publisher: American Physical Society. [Online]. Available: \url{https://link.aps.org/doi/10.1103/PhysRevA.93.022311}
\BIBentrySTDinterwordspacing

\bibitem{maslov_toffoli_2005}
\BIBentryALTinterwordspacing
D.~Maslov, G.~Dueck, and D.~Miller, ``Toffoli network synthesis with templates,'' \emph{IEEE Transactions on Computer-Aided Design of Integrated Circuits and Systems}, vol.~24, no.~6, pp. 807--817, Jun. 2005, conference Name: IEEE Transactions on Computer-Aided Design of Integrated Circuits and Systems. [Online]. Available: \url{https://ieeexplore.ieee.org/abstract/document/1432873}
\BIBentrySTDinterwordspacing

\bibitem{maslov_depth_2022}
\BIBentryALTinterwordspacing
D.~Maslov and B.~Zindorf, ``Depth {Optimization} of {CZ}, {CNOT}, and {Clifford} {Circuits},'' \emph{IEEE Transactions on Quantum Engineering}, vol.~3, pp. 1--8, 2022, conference Name: IEEE Transactions on Quantum Engineering. [Online]. Available: \url{https://ieeexplore.ieee.org/abstract/document/9792395}
\BIBentrySTDinterwordspacing

\end{thebibliography}
\end{document}